\documentclass[10pt]{article}
\usepackage{multicol}
\usepackage{lettrine}
\usepackage{braket}
\usepackage{framed}
\usepackage{graphicx}
\usepackage{tabularx}
\usepackage{subcaption} 
\usepackage{xcolor}
\usepackage{amsmath}
\numberwithin{equation}{section}
\usepackage{float}
\usepackage{booktabs}
\usepackage{amssymb}
\usepackage{xcolor}
\usepackage{mathrsfs}
\usepackage[ruled,vlined,linesnumbered]
{algorithm2e}
\usepackage{tikz}
\usepackage{amsthm}   
\usepackage{aliascnt}
\theoremstyle{plain}
\newtheorem{theorem}{Theorem}[section]
\newtheorem{lemma}[theorem]{Lemma}
\newtheorem{corollary}[theorem]{Corollary}

\newaliascnt{result}{theorem}
\newtheorem{result}[result]{Result}
\aliascntresetthe{result}
\usepackage[
  backend=biber,
  style=numeric-comp,
  sorting=none,
  sortcites=true,
  giveninits=false,
  maxbibnames=99,
  url=false,
  doi=false,
  eprint=false
]{biblatex}
\usepackage{hyperref}

\usepackage[capitalize]{cleveref}

\crefname{theorem}{Theorem}{Theorems}
\crefname{lemma}{Lemma}{Lemmas}
\crefname{corollary}{Corollary}{Corollaries}
\crefname{proposition}{Proposition}{Propositions}
\crefname{result}{Result}{Results}
\Crefname{theorem}{Theorem}{Theorems}
\Crefname{lemma}{Lemma}{Lemmas}
\Crefname{corollary}{Corollary}{Corollaries}
\Crefname{proposition}{Proposition}{Propositions}
\Crefname{result}{Result}{Results}

\usetikzlibrary{tikzmark,calc}
\title{Quantum simulation of slow analytic time-dependent Hamiltonians} 

\author{Chenhao Zhao$^{1}$, \quad Yinan Li$^{2}$, \quad Dong An$^{3}$ \\ 
\footnotesize $^{1}$ School of Mathematics and Statistics, Wuhan University, Wuhan, China\\
\footnotesize $^{2}$ School of Artificial Intelligence, Wuhan University, Wuhan, China\\
\footnotesize $^{3}$ Beijing International Center for Mathematical Research, Peking University, Beijing, China\\
}

\date{ }

\begin{document}

\maketitle

\begin{abstract}
    We develop a quantum algorithm for slow analytic Hamiltonians $\widetilde H(t)=H(t/T)$ with $\|H(s)\|\leq\alpha$ that achieves nearly additive query complexity and low gate overhead. 
    Our main technical contribution is a periodic Gevrey extension of $H(s)$, together with Fourier component decay and truncation bounds that enable an efficient finite-dimensional simulation. 
    Combined with Floquet embedding and optimal time-independent Hamiltonian simulation technique, this gives query complexity $\widetilde{\mathcal O}\!\left(\alpha T+\log(1/\varepsilon)\right)$ and additional gate complexity $\widetilde{\mathcal O}\!\left((\alpha T+\log(1/\varepsilon))^2\log(1/\varepsilon)\right)$, assuming coherent access to $H'(s)$ and endpoint derivatives. 
    For slow analytic control Hamiltonians, only block encodings of the time-independent control operators are required, with the same query complexity and lower gate overhead. 
    Our method also extends to Gevrey Hamiltonians and improves the precision dependence for simulating slow analytic semi-dissipative linear differential equations. 
\end{abstract}

\tableofcontents

\section{Introduction}\label{sec:intro}

Simulating the time evolution of a physical quantum system, also known as \emph{Hamiltonian simulation} problem, is a central application of quantum computers~\cite{Feynman1982}. 
It is also a basic primitive in quantum algorithms for other tasks, such as linear systems~\cite{HarrowHassidimLloyd2009,ChildsKothariSomma2017,SubasiSommaOrsucci2019,AnLin2022}, differential equations~\cite{AnLiuLin2023,AnChildsLin2023,JinLiuYu2022}, and adiabatic quantum computation and state preparation~\cite{AlbashLidar2018}. 
Given a time-dependent Hermitian Hamiltonian $\widetilde H(t)$, the goal is to solve the Schr\"odinger equation 
\begin{equation}\label{eqn:ham_sim}
    i \frac{d}{dt} \ket{\psi(t)} = \widetilde{H}(t) \ket{\psi(t)}, \quad \ket{\psi(0)} = \ket{\psi_0},  
\end{equation}
by preparing a state that approximates $\ket{\psi(T)}$ to $2$-norm error at most $\varepsilon$.

A special case of the Hamiltonian simulation problem is the \emph{time-independent Hamiltonian simulation}, where $\widetilde{H}(t) \equiv \widetilde{H}$ is a constant Hermitian matrix and the solution can be written as $\ket{\psi(T)} = e^{-i\widetilde{H} T} \ket{\psi_0}$. 
The time-independent case has been studied extensively~\cite{BerryAhokasCleveEtAl2007,BerryChilds2012,BerryCleveGharibian2014,BerryChildsCleveEtAl2014,BerryChildsCleveEtAl2015,BerryChildsKothari2015,LowChuang2017,ChildsMaslovNamEtAl2018,ChildsOstranderSu2019,Campbell2019,Low2019,ChildsSu2019,LowChuang2019,GilyenSuLowEtAl2019,ChenHuangKuengEtAl2020,SahinogluSomma2020,ChildsSuTranEtAl2021,Faehrmann2022Randomizing,MartynLiuChinEtAl2023,Zhang2024Parallel,Zlokapa2024Hamiltonian,Bosse2025Efficient,Rendon2024Improved,WatsonWatkins2024}. 
Among them, the quantum singular value transformation (QSVT)~\cite{GilyenSuLowEtAl2019} solves the problem by implementing the matrix function $e^{-i\widetilde{H} T}$, with queries to the Hamiltonian for $\mathcal{O}(\alpha T + \log(1/\varepsilon))$ times\footnote{The actual scaling of QSVT as well as some other methods also has a $\log\log(1/\varepsilon)$ factor on the denominator, which we ignore for simplicity throughout this work.}, where $\alpha \geq \|\widetilde{H}(t)\|$ is a block-encoding normalization factor. 
This matches the known lower bound for time-independent Hamiltonian simulation~\cite{BerryChildsKothari2015}. 
Additionally, product formulas provide a complementary approach and can exploit structures such as locality to obtain better system-size dependence, although their generic time and precision dependence is suboptimal~\cite{ChildsSuTranEtAl2021}.

The general time-dependent problem is less well understood~\cite{HuyghebaertDeRaedt1990,WiebeBerryHoyerEtAl2010,PoulinQarrySommaEtAl2011,WeckerHastingsWiebeEtAl2015,LowWiebe2019,kieferova2019simulating,BerryChildsSuEtAl2020,AnFangLin2021,AnFangLin2022,mizuta2023optimal,mizuta2023optimalmulti,watkins2024time,Cao2025Unifying,Li2025TimeDependent,chen2026optimal}. 
The truncated Dyson series method~\cite{kieferova2019simulating,LowWiebe2019}, which generalizes truncated Taylor series simulation, achieves query complexity
$\widetilde{\mathcal O}\!\left(\alpha T\log(1/\varepsilon)\right)$. 
This dependence is nearly optimal in time and precision when each is considered separately. 
However, unlike the time-independent case, the two factors of time and precision enter multiplicatively, leading to significant computational overhead for an accurate long-time simulation and motivating algorithms with additive dependence on $\alpha T$ and $\log(1/\varepsilon)$. 

To design such an algorithm with additive scaling, a natural route is to reduce the time-dependent Hamiltonian simulation tasks to time-independent ones and apply the optimal QSVT approach. 
Floquet constructions achieve this for periodic and multiperiodic Hamiltonians, yielding optimal additive query complexity under suitable access assumptions~\cite{mizuta2023optimal,mizuta2023optimalmulti}. 
Other works~\cite{watkins2024time,Cao2025Unifying,Li2025TimeDependent} study more general clock embedding techniques and clarify the relationship between time-dependent and time-independent simulations, but their implementations generally introduce additional overhead that prevents comparable additive bounds. 

Very recently, a concurrent and independent work~\cite{chen2026optimal} makes a remarkable breakthrough by achieving query complexity $\mathcal{O}(\alpha T + \log(1/\varepsilon))$, matching the time-independent lower bound for general Lipschitz-continuous Hamiltonians in the standard HAM-T model. 
The algorithm is based on a novel extension of the transducer framework~\cite{BelovsJefferyYolcu2024} to encode the complete sequence of time-ordered steps. 
Its direct circuit implementation, however, has additional gate complexity with polynomial dependence on $1/\varepsilon$. 
As a comparison, the truncated Dyson series approach needs only $\mathcal{O}(\alpha T \operatorname{polylog}(\alpha T/\varepsilon))$ additional gates. 
This leaves open the following central question: 
\begin{center}
   \emph{Can time-dependent Hamiltonians be simulated with nearly additive query complexity
    $\widetilde{\mathcal O}\!\left(\alpha T+\log(1/\varepsilon)\right)$
    while keeping the additional gate complexity polylogarithmic in $1/\varepsilon$?}
\end{center}

In this work, we answer this question affirmatively for slow analytic Hamiltonian of the form
$\widetilde H(t)=H(t/T)$, which describes a fixed analytic control path traversed over a physical time $T$. 
Such dynamics arise in adiabatic quantum computation~\cite{FarhiGoldstoneGutmannEtAl2000,AlbashLidar2018} and quantum annealing~\cite{kadowaki1998quantum,hauke2020perspectives}, with wide applications to molecular ground-state preparation~\cite{kremenetski2021simulation} and adiabatic state transfer in quantum-dot spin chains~\cite{kandel2021adiabatic} and topological pumping~\cite{deng2024highorder}. 
Our main technical contribution is the construction of a periodic Gevrey extension that agrees with $\widetilde{H}(t)$ on the physical time interval, together with Fourier decay and Floquet space truncation bounds for this extension. 
These results allow standard Floquet embedding and QSVT to yield the desired nearly additive query scaling while maintaining polylogarithmic gate dependence on the precision. 
Our approach also applies to lower-regularity Gevrey Hamiltonians, but the query complexity becomes suboptimal. 
Additionally, we combine our approach with the linear combination of Hamiltonian simulation (LCHS) method~\cite{LowSomma2025} to improve the precision dependence for slow analytic non-unitary dynamics.

\subsection{Results}

We first state our main result for general slow analytic Hamiltonians. 
Throughout this subsection, the regularity parameters of the Hamiltonian are treated as constants independent of $T$ and $\varepsilon$.

\begin{result}[Slow analytic Hamiltonians, informal version of~\cref{th:query_gate_complex_gene}]\label{res:main}
    Let $\widetilde{H}(t) = H(t/T)$ for an analytic Hamiltonian $H(s)$ with $\|H(s)\| \leq \alpha$. 
    Suppose that we are given access to the HAM-T oracle of $H'(s)$, and oracles for the derivatives $H^{(k)}(0)$ and $H^{(k)}(1)$ up to $\widetilde{\mathcal{O}}(\log \alpha T+\log(1/\varepsilon))$-th order. 
    Then there exists a quantum algorithm for solving~\cref{eqn:ham_sim} up to time $T$ with error at most $\varepsilon$, using $\widetilde{\mathcal{O}}\left(\alpha T + \log(1/\varepsilon)\right)$ oracle queries and $\widetilde{\mathcal{O}}\left((\alpha T + \log(1/\varepsilon))^2\log(1/\varepsilon)\right)$ additional gates. 
\end{result}

In~\cref{res:main}, we make two important assumptions on the time-dependent Hamiltonian. 
First, we assume it can be written as $\widetilde{H}(t) = H(t/T)$ for another Hamiltonian defined through the rescaled time $s = t/T$, which means a slow condition since $\frac{d}{dt}\widetilde H(t)=\frac{1}{T}H'(t/T)$ is on the order of $1/T$.  
The second assumption is the analytic condition on $H(s)$. 
We quantify analyticity by assuming infinite differentiability and $\|H^{(k)}(s)\| \leq C D^k k!$ for constants $C > 0, D > 0$, which is equivalent to assuming point-wise non-trivial convergence of the Taylor series of $H(s)$ for any $s \in [0,1]$. 
Under these two assumptions, the query complexity of our algorithm depends nearly additively on the evolution scale $\alpha T$ and the precision parameter $\log(1/\varepsilon)$. 
The additional gate overhead is poly-logarithmic in $1/\varepsilon$, making the algorithm gate efficient for highly accurate simulation, although it contains a nearly quadratic time dependence.

However, we remark that the oracle assumptions in~\cref{res:main} (see~\cref{subsubsec:oracle_general_case} for detailed oracle definitions) differ from the standard HAM-T model, which provides coherent access directly to $H(s)$ or $\widetilde H(t)$~\cite{LowWiebe2019,chen2026optimal}. 
Our general construction instead requires access to $H'(s)$ and to high-order endpoint derivatives. 
Thus, the result should not be interpreted as an improvement under the same input model. 
When $H(s)$ is specified through an explicit analytic formula or a structured decomposition, the required derivative block encodings may be constructed from the same underlying primitives with comparable costs of standard HAM-T (because $H(s)$ and $H'(s)$ would have the same structure), but actual costs must be assessed on a case-by-case basis.

For an important structured class of Hamiltonians, the derivative access assumption can be removed. 
Consider a control Hamiltonian
\begin{equation}
    H(s)=\sum_{j=0}^{j_{\max}-1}\alpha_j(s)M_j,
\end{equation}
where the coefficient functions $\alpha_j(s)$ are real scalar-valued analytic and the $M_j$ are time-independent Hermitian matrices. 
In this setting, the coefficients required by the algorithm can be computed classically, while the quantum circuit only queries block encodings of the operators $M_j$. 
The result is summarized as follows. 

\begin{result}[Slow analytic control Hamiltonians, informal version of~\cref{th:query_gate_complex_cont}]\label{res:control}
    Under the setting of~\cref{res:main}, suppose additionally that $H(s) = \sum_{j=0}^{j_{\max}-1} \alpha_j(s) M_j$ for some time-independent Hamiltonians $M_j$ and real analytic scalar functions $\alpha_j(s)$. 
    Then, given controlled block encodings of the $M_j$,  the evolution can be simulated to error $\varepsilon$ using $\widetilde{\mathcal{O}}\left(\alpha T + \log(1/\varepsilon)\right)$ oracle queries and $\widetilde{\mathcal{O}}\left((\alpha T + \log(1/\varepsilon))^2\right)$ additional gates. 
\end{result}

\subsubsection{Main ideas of the algorithm}

We briefly sketch the main ideas and steps of our algorithm. 
Recall that the original physical-time Hamiltonian is
$\widetilde{H}(t)=H(t/T)$, where $s=t/T\in[0,1]$ is the rescaled time.

\paragraph{Step 1: Constructing a periodic extension.}
We extend $H(s)$ from $[0,1]$ to a $2$-periodic Hamiltonian $\hat{H}(s)$. A direct periodic continuation would generally be nonsmooth because the derivatives of $H(s)$ at $s=0$ and $s=1$ need not agree. 
We therefore define
\begin{equation}
\hat{H}(s)=
\begin{cases}
H(s), & 0\leq s\leq 1,\\
F_{\tau}(s-1)+G_{\tau}(s-2), & 1\leq s\leq 2,
\end{cases}
\qquad
\hat{H}(s+2)=\hat{H}(s),
\end{equation}
where
$F_{\tau}(x)
=
\sum_{k=0}^{\infty}
\frac{H^{(k)}(1)}{k!}x^k\chi_{\tau}(R_kx)$
and 
$G_{\tau}(x)
=
\sum_{k=0}^{\infty}
\frac{H^{(k)}(0)}{k!}x^k\chi_{\tau}(R_kx)$. 
Here, $\chi_{\tau}$ is a compactly supported cut-off function such that $\chi_{\tau}(0) = 1$, and $R_k$ is chosen to guarantee convergence with controlled derivatives. 
This construction matches all derivatives at $s=1$ and at the periodic boundary $s=0,2$. 
In particular,
$\hat{H}(t/T)=\widetilde{H}(t)$ for $t\in[0,T]$, so the extension does not change the dynamics that we wish to simulate, and we will use $\widetilde{H}(t)$ and $H(s)$ to denote the extended Hamiltonians with an abuse of notations. 
A detailed discussion on this extension is provided in~\cref{sec:Periodic_Extension}. 

\paragraph{Step 2: Fourier expansion and Floquet embedding.}
We Fourier expand the periodic extension as
\begin{equation}\label{eqn:intro_Fourier}
\widetilde{H}(t) = \sum_{m\in\mathbb{Z}} \widetilde{H}_m e^{im\omega t},
\qquad
\widetilde{H}_m
=
\frac{1}{2T}\int_0^{2T}
\widetilde{H}(t)e^{-im\pi t/T}\,dt.
\end{equation}
Floquet theory~\cite{mizuta2023optimal} converts the periodic time-dependent problem into a time-independent problem on an enlarged infinite-dimensional Hilbert space. 
Specifically, the corresponding effective infinitely-dimensional Hamiltonian is
\begin{equation}
\mathscr{H}_{\mathrm{eff}}
=
\sum_{l\in\mathbb{Z}}
|l\rangle\langle l|
\otimes
\bigl(\widetilde{H}_0-l \pi I/T\bigr)
+
\sum_{\substack{l,m\in\mathbb{Z}}}
|l\rangle\langle l+m|
\otimes
\widetilde{H}_m, 
\end{equation}
and the desired solution can be recovered through
\begin{equation}
|\psi(t)\rangle
=
\sum_{l\in\mathbb{Z}}
e^{-il\omega t}
\langle l|
e^{-i\mathscr{H}_{\mathrm{eff}}t}
|0\rangle
|\psi(0)\rangle .
\end{equation}

\paragraph{Step 3: Truncation and quantum simulation.}
We truncate the infinite effective Hamiltonian to a finite-dimensional one with $|l|\leq l_{\max}$ (see~\cref{sec:Truncation_order_Four} for a detailed analysis on the truncation order). 
We then construct a block encoding of this truncated effective Hamiltonian using linear combination of unitaries (LCU)~\cite{ChildsWiebe2012} for the Fourier coefficients. 
Notice that computing the Fourier coefficients efficiently for general $H(s)$ would require queries to $H'$ at suitable quadrature nodes and to the endpoint derivatives $H^{(k)}(0)$ and $H^{(k)}(1)$, otherwise the additive query scaling would be ruined. 
Meanwhile, for control Hamiltonian, Fourier coefficients can be efficiently computed directly from~\cref{eqn:intro_Fourier} using queries only to $M_j$'s. 
See~\cref{sec:Algorithm_implementation} for detailed constructions. 
After the truncated effective Hamiltonian is built, we use QSVT to implement the time-independent evolution $e^{-i\mathscr{H}_{\mathrm{eff}}T}$. 

At the circuit level, the algorithm first prepares a nearly translation-invariant superposition over the Floquet indices $\ket{l}$, performs the effective-Hamiltonian evolution together with the phase operation $e^{-il\omega T}$, and then unprepares the Floquet register. 
Approximate translation symmetry~\cite{mizuta2023optimal} gives the desired physical state with constant amplitude, which is further amplified using oblivious amplitude amplification~\cite{BerryChildsKothari2015}. 
See~\cref{sec:Error_analysis_5} for details.

\subsubsection{Additional results}

Our algorithm also applies to a broader Gevrey class of Hamiltonians. 
For a parameter $\sigma \geq 1$, the Gevrey class $G^{\sigma}([0,1])$ is defined as the set of Hamiltonians defined on $[0,1]$ such that $\|H^{(k)}(s)\| \leq C D^k (k!)^{\sigma}$ for constants $C>0$, $D>0$ and any $k \geq 0$. 
Notice that $G^{1}([0,1])$ corresponds to the real analytic Hamiltonian class, and $G^{\sigma}([0,1])$ with $\sigma > 1$ is a Hamiltonian class whose regularity is between infinitely differentiable and real analytic. 
The complexity of our algorithm applied to $G^{\sigma}([0,1])$ is as follows. 

\begin{result}[Slow Gevrey Hamiltonians, informal version of \cref{th:query_gate_complex_gene,th:query_gate_complex_cont}]\label{res:gevrey}
    Let $1 \leq \sigma < 2$ and let $\widetilde H(t)=H(t/T)$, where $H(s) \in G^\sigma([0,1])$ and $\|H(s)\| \leq \alpha$. 
    Under the same oracle assumptions of \cref{res:main} and~\cref{res:control}, the evolution up to time $T$ can be simulated to error $\varepsilon$ using $\widetilde{\mathcal O}\!\left(\left(\alpha T+\log(1/\varepsilon)\right)^\sigma\right)$
    oracle queries. 
    The additional gate overhead is $\widetilde{\mathcal{O}}\left((\alpha T + \log(1/\varepsilon))^{2\sigma}\log^{\sigma}(1/\varepsilon)\right)$ for the general Hamiltonian, and $\widetilde{\mathcal{O}}\left((\alpha T + \log(1/\varepsilon))^{2\sigma}\right)$ for the control Hamiltonian. 
\end{result}

\cref{res:gevrey} shows that the complexity of our approach grows as the regularity of the Hamiltonian becomes worse. 
This is because, for a Gevrey-class Hamiltonian of index $\sigma>1$, the Fourier coefficients of the periodic extension decay subexponentially rather than exponentially. 
The Floquet truncation order consequently grows as a power of the evolution time and logarithmic precision. 
Thus, the analytic case $\sigma=1$ gives the nearly additive scaling stated in \cref{res:main}, while lower regularity incurs a polynomial loss.

As a second additional result, we combine our Hamiltonian simulation algorithm with the optimal-scaling LCHS representation~\cite{LowSomma2025}. 
Consider the homogeneous linear differential equation
\begin{equation}\label{eq:intro_linear_ode}
    \frac{d}{dt}u(t)=-\widetilde{A}(t)u(t),
    \qquad u(0)=u_0,
\end{equation}
where $\widetilde{A}(t) = A(t/T)$ for an analytic matrix-valued function $A(s)$ such that $A(s)=L(s)+iH(s)$ for Hermitian matrices $L(s)$ and $H(s)$ and $L(s)\succeq 0$. 
LCHS expresses its non-unitary propagator $\mathcal{T}e^{-\int_0^t \widetilde{A}(\tau) d\tau}$ of~\cref{eq:intro_linear_ode} as a linear combination of unitary evolution operators generated by the Hamiltonian $kL(s)+H(s)$, so~\cref{eq:intro_linear_ode} can be simulated through LCU of Hamiltonian simulation problems. 
Since our approach can improve the time-dependent Hamiltonian simulation subroutine, the overall complexity of solving general linear differential equation is thereby reduced. 

\begin{result}[Slow analytic linear ODEs, informal version of \cref{th:slow_ode_complexity}]\label{res:linear_ode}
    Suppose $\widetilde{A}(t) = A(t/T)$ for an analytic matrix $A(s)$ with $\|A(s)\| \leq \alpha$. 
    Assume access to either the derivative oracles or the control-Hamiltonian oracles required by our Hamiltonian simulation algorithm. 
    Then a quantum state approximating the normalized solution of~\cref{eq:intro_linear_ode} at time $T$ to error $\varepsilon$ can be prepared using $\widetilde{\mathcal O}\!\left(
    \frac{\|u(0)\|}{\|u(T)\|}\alpha T \log\left(\frac{1}{\varepsilon}\right)
    \right)$ queries. 
\end{result}

For generic linear differential equations with time-dependent coefficients, the previous best scaling is $\widetilde{\mathcal O}\!\left(
    \frac{\|u(0)\|}{\|u(T)\|}\alpha T \left(\log\left(\frac{1}{\varepsilon}\right)\right)^2
    \right)$~\cite{LowSomma2025}. 
Therefore, \cref{res:linear_ode} quadratically improves the precision dependence when the coefficient matrices are restricted to be slow analytic. 
This also matches the time and precision dependence of optimal-scaling LCHS in~\cite{LowSomma2025} for a time-independent matrix up to polylogarithmic factors. 
We give the construction and a proof sketch in \cref{sec:linear_odes}.

\subsection{Comparison and related works}
\label{subsec:comparison_related}

Floquet theory has already been previously used to obtain efficient simulation algorithms for periodic and multiperiodic Hamiltonians. 
Specifically,~\cite{mizuta2023optimal} considers a periodic Hamiltonian with exponentially decaying Fourier components, and~\cite{mizuta2023optimalmulti} generalizes this construction to Hamiltonians driven by a constant number of frequencies and Fourier modes. 
For evolution over a constant number of periods, both algorithms obtain an additive dependence on the evolution time and precision, together with only logarithmically many additional gates per oracle query. 
A more general Sambe–Howland clock framework has also been well explored~\cite{Cao2025Unifying} to establish closely relation between time-dependent and time-independent Hamiltonian simulations, and construct efficient time-dependent simulation algorithms. 

Our algorithm follows the same broad Floquet strategy and translation-symmetry amplitude amplification technique proposed in~\cite{mizuta2023optimal}, but treats a different class of inputs. 
Unlike existing works, we do not assume that the Hamiltonian $\widetilde H(t)$ is periodic or multiperiodic, or that its Fourier components are supplied as input. 
Instead, starting from a slow Hamiltonian $\widetilde H(t)=H(t/T)$, we construct a periodic Gevrey extension $\hat H(s)$ of the rescaled Hamiltonian $H(s)$ and explicitly construct the Fourier components. 
Furthermore, the extended Hamiltonian $\hat H(s)$ is no longer real analytic in general due to the identity theorem for real analytic functions, so the Fourier components do not decay exponentially and the existing complexity analysis in~\cite{mizuta2023optimal,mizuta2023optimalmulti} does not apply. 
Therefore, we derive the decay of the Fourier coefficients and the corresponding Lie-Robinson bound for Gevrey Hamiltonians, making it possible to truncate the Floquet space and construct its effective Hamiltonian for this broader class of Hamiltonians. 

A closely related periodic extension idea appears in the product and multi-product formula analysis in~\cite{Mizuta2024Explicit}, where a Hamiltonian defined on a finite time interval is extended to a sufficiently smooth periodic Hamiltonian, and the corresponding Floquet representation is used to derive explicit error bounds with commutator scaling from its time-independent counterpart in~\cite{ChildsSuTranEtAl2021}. 
The extension there mainly serves as an analytical tool, and is only differentiable up to a finite order (i.e., the convergence order of the product formula being analyzed). 
In contrast, our extension technique is an explicit ingredient of the simulation algorithm, and the extended Hamiltonian is infinitely differentiable, which is the key to achieving poly-logarithmic scaling in precision. 

During the final stages of preparing the present paper, we became aware of a recent concurrent work~\cite{chen2026optimal}, which, as discussed earlier, obtained the optimal query complexity $\mathcal O\left(
        \alpha T+
        \log(1/\varepsilon)
    \right)$
for general Lipschitz-continuous Hamiltonians in the HAM-T model, without periodicity or slow-variation assumptions. 
Their Hamiltonian assumption is therefore substantially weaker, and their query bound is tighter than ours by polylogarithmic factors. 
Their direct circuit implementation, however, uses a number of additional gates that is polynomial in $1/\varepsilon$. 
Therefore, our result is complementary that under the stronger slow analytic assumption and derivative-access model, it retains
$\mathcal O(\operatorname{polylog}(1/\varepsilon))$
additional gates. 
In particular, our main advantage lies in simultaneously obtaining nearly additive query complexity and polylogarithmic gate overhead in the precision.

Table~\ref{tab:related_work_comparison} summarizes these comparisons as well as a few other existing algorithms. 
Since the cited algorithms use different oracles and are applicable to different class of Hamiltonians, we state those assumptions explicitly. 

\begin{table}[t]
\centering
\scriptsize
\setlength{\tabcolsep}{2pt}
\renewcommand{\arraystretch}{1.25}

\begin{tabularx}{\textwidth}{
@{}
>{\raggedright\arraybackslash}p{0.12\textwidth}
>{\raggedright\arraybackslash}p{0.16\textwidth}
>{\raggedright\arraybackslash}p{0.24\textwidth}
>{\raggedright\arraybackslash}p{0.22\textwidth}
>{\raggedright\arraybackslash}X
@{}}
\toprule
Algorithms
& Hamiltonian assumption
& Oracle assumption
& Query complexity
& Gate overhead \\
\midrule

Truncated Dyson series~\cite{LowWiebe2019,kieferova2019simulating}
&
Differentiable
&
HAM-T or sparse-access oracles of $\widetilde{H}(t)$
&
$\displaystyle
 \widetilde{\mathcal O}\!\left(
 \alpha T \log\left(\frac{1}{\varepsilon}\right)
 \right)$
&
$\displaystyle
 \widetilde{\mathcal O}\!\left(
 \alpha T\,\operatorname{polylog}\left(\frac{1}{\varepsilon}\right)
 \right)$
\\[2pt]

Floquet theory~\cite{mizuta2023optimal,mizuta2023optimalmulti}
&
Periodic or multiperiodic, exponentially decaying Fourier modes
&
Block encodings of the Fourier components or controlled Hamiltonians, control-coefficient-state preparation (for control Hamiltonian)
&
$\displaystyle
 \mathcal O \!\left(
 \alpha T+\log\left(\frac{1}{\varepsilon}\right)
 \right)$
&
$\displaystyle
 \widetilde{\mathcal O}\!\left(
 \alpha T+\log\left(\frac{1}{\varepsilon}\right)
 \right)$
\\[2pt]

Discrete clock~\cite{watkins2024time}
&
Differentiable, given as an LCU
&
Coefficient-state preparation oracle, block-encodings of the controlled Hamiltonians
&
$\displaystyle
 \widetilde{\mathcal O}\!\left(
 \alpha T + \log\left(\frac{1}{\varepsilon}\right) + \frac{\beta T^2}{\varepsilon}
 \right)$
&
unspecified
\\[2pt]

Discrete clock with Gaussian quadrature~\cite{Li2025TimeDependent}
&
Differentiable sparse
&
Sparse-access oracle of $\widetilde{H}(t)$
&
$\displaystyle
 \widetilde{\mathcal O}\!\left(
 \alpha T \log \left(\frac{1}{\varepsilon}\right)
 \right)$
&
$\displaystyle
 \widetilde{\mathcal O}\!\left(
 \alpha T\,\operatorname{polylog}\left(\frac{1}{\varepsilon}\right)
 \right)$
\\[2pt]

Transducer~\cite{chen2026optimal}
&
Lipschitz-continuous
&
HAM-T of $\widetilde{H}(t)$
&
$\displaystyle
 \mathcal O\!\left(
 \alpha T+
 \log\left(\frac{1}{\varepsilon}\right)
 \right)$
&
$\displaystyle
 \widetilde{\mathcal O}\!\left(
 \frac{\alpha T^{\frac{5}{2}}}{\sqrt{\varepsilon}}
 \max\left\{\alpha^{\frac{3}{2}}, \frac{\beta \sqrt{T}}{\sqrt{\varepsilon}}\right\}
 \right)$
\\

\midrule

This work: general Hamiltonian
&
Slow analytic
&
HAM-T of $\widetilde{H}'(t)$, block-encodings of the endpoint derivatives of $\widetilde{H}(t)$
&
$\displaystyle
 \widetilde{\mathcal O}\!\left(
 \alpha T+\log\left(\frac{1}{\varepsilon}\right)
 \right)$
&
$\widetilde{\mathcal{O}}\left(\left(\alpha T + \log\left(\frac{1}{\varepsilon}\right)\right)^2\log\left(\frac{1}{\varepsilon}\right)\right)$
\\[2pt]

This work: control Hamiltonian
&
Slow analytic control Hamiltonian
$\widetilde{H}(t)=\sum_j \widetilde{a}_j(t) M_j$
&
Block-encodings of the controlled Hamiltonians $M_j$
&
$\displaystyle
 \widetilde{\mathcal O}\!\left(
 \alpha T+\log\left(\frac{1}{\varepsilon}\right)
 \right)$
&
$\widetilde{\mathcal{O}}\left(\left(\alpha T + \log\left(\frac{1}{\varepsilon}\right)\right)^2\right)$
\\

\bottomrule
\end{tabularx}

\caption{Comparison of representative time-dependent Hamiltonian simulation algorithms. 
Gate overhead counts additional one- and two-qubit gates and excludes the implementation cost of the stated input oracles. The two rows for this work display the analytic case.
The parameter $T$ is the evolution time, $\varepsilon$ is the error, $\alpha$ is the block-encoding factor for block-encoding oracles and $d \|\widetilde{H}\|_{\max}$ for sparse oracle with sparsity $d$, and $\beta$ is the Lipschitz constant. 
All other parameters are treated as constants. 
}
\label{tab:related_work_comparison}
\end{table}

\subsection{Discussion}\label{subsec:discussion}

In this work, we show that slow analytic time-dependent Hamiltonians can be simulated with nearly additive query $\widetilde{\mathcal O}(\alpha T+\log(1/\varepsilon))$ and ${\mathcal O}(\operatorname{poly}(T,\log(1/\varepsilon)))$ additional gate overhead. 
The main ingredient is to transform the nonperiodic Hamiltonian $H(s)$ into a periodic Gevrey extension $\hat{H}(s)$, after which Floquet theory and time-independent Hamiltonian simulation techniques become applicable. 

An immediate open question is whether these assumptions can be relaxed while achieving the same or better gate efficiency. 
In particular, it would be desirable to replace the derivative oracles by standard HAM-T access alone, or to construct the required periodic extension directly from coherent evaluations of $H(s)$. 
Another question is whether the slow condition $\widetilde{H}(t)=H(t/T)$ can be replaced by a more flexible measure of temporal variation, such as an integral norm involving $\widetilde{H}'(t)$. 
It is also interesting to explore whether we can further improve the gate overhead to $\widetilde{\mathcal O}(\alpha T\log(1/\varepsilon))$, matching the truncated Dyson series method, or even better. 
More fundamentally, it remains open whether general Lipschitz-continuous Hamiltonians admit both optimal query complexity and additional gate complexity that is only polylogarithmic in $1/\varepsilon$.

A second direction is to exploit additional physical structure. 
The present bounds are formulated mainly in a block-encoding model and therefore do not capture possible improvements arising from locality, sparsity, or small commutators, as occurs in product formula methods~\cite{Mizuta2024Explicit}. 
Combining such structure-sensitive techniques with the periodic extension and Floquet framework may improve the system-size dependence and reduce the overhead associated with the number of control terms.

\subsection{Organization}

The rest of this work is organized as follows. 
In \cref{sec:preliminaries}, we review Gevrey functions, Floquet theory, and the simulation of analytic periodic Hamiltonians. 
In \cref{sec:Periodic_Extension}, we construct the periodic Gevrey extension used by our algorithm. 
In \cref{sec:Truncation_order_Four}, we establish the Fourier decay and Floquet space truncation bounds, and in \cref{sec:Error_analysis_5} we discuss the amplification and effective Hamiltonian errors. 
In \cref{sec:Algorithm_implementation}, we implement the algorithm and derive its query and gate complexities in the general and control-Hamiltonian oracle models. 
Finally, \cref{sec:linear_odes} applies these results to slow linear differential equations through LCHS.

\section*{Acknowledgments}

Part of this work was completed while CHZ visited Peking University. 
DA acknowledges funding from Quantum Science and Technology - National Science and Technology Major Project via Project 2024ZD0301900, and the support by The Fundamental Research Funds for the Central Universities, Peking University.

\section{Preliminaries}\label{sec:preliminaries}
\subsection{Symbol conventions}
In this work, we will use $\|\cdot\|$ to denote the vector 2-norm for vectors, and the spectral norm for matrices. 
We will use $H(s)$ to denote the original time-dependent Hamiltonian on $s\in[0,1]$, and periodically extend it to be $\hat{H}(s)$ defined on $[0,2]$. We define $\widetilde{H}(t):=\hat{H}\left(\frac{t}{T}\right)\;\text{for}\,t\in[0,2T]$, where $T>0$ fixed is the adiabatic evolution time. We will use $\widetilde{H}_m$ to denote the Fourier coefficients of the periodic Hamiltonian $\widetilde{H}(t)$.
\subsection{Gevrey class functions}
Let $\sigma \geq 1$. A smooth function $f \in C^{\infty}(\mathbb{R})$ belongs to the Gevrey class of index $\sigma$, denoted by $G^{\sigma}(\mathbb{R})$, if there exists constants $C > 0$, $D > 0$ such that
\begin{equation}\label{eq:def_of_Gevrey_func}
|f^{(n)}(t)| \le C D^n (n!)^{\sigma}, \quad \text{for any}\;t \in \mathbb{R},\;n\geq 0\;.
\end{equation}
When $\sigma=1$ it is also called an analytic function. When $K$ is a closed interval in $\mathbb R$, the space $G^{\sigma}(K)$ is defined by replacing $\mathbb R$ with $K$ in \cref{eq:def_of_Gevrey_func}.\\ In this work, we will use $\sigma$ to denote the Gevrey index of the original Hamiltonian $H(t)$ before periodic extension, and $\tau$ to denote the one for the cut-off function $\chi_{\tau}(t)$. We define $\varrho:=\sigma+\tau-1$ to be the Gevrey index of extended $H(t)$ and $\widetilde{H}(t)$. 
\subsection{Floquet theory}\label{subsec:Floquet_theory}
For the Schrödinger equation under a periodic Hamiltonian in a Hilbert space $\mathscr{H}$,
\begin{equation}\label{eq:Hamitonain_problems}
    i \frac{d}{dt} |\psi(t)\rangle = \widetilde{H}(t) |\psi(t)\rangle, \quad \widetilde{H}(t + 2T) = \widetilde{H}(t),
\end{equation}
and its solution
\begin{equation}\label{eq:origin_problems}
    |\psi(t)\rangle = U(t) |\psi(0)\rangle, \quad U(t) = \mathcal{T} \exp\left(-i \int_{0}^{t} dt' \, \widetilde{H}(t') \right).
\end{equation}
Floquet theorem says that the solution $|\psi(t)\rangle$ can be written in the form of
\begin{equation}
    |\psi(t)\rangle = \sum_{\alpha=1}^{\dim(\mathscr{H})}c_\alpha e^{-i\varepsilon_\alpha t}\,|\phi_\alpha(t) \rangle\,,\quad |\phi_\alpha(t+2T)\rangle = |\phi_\alpha(t)\rangle\,.
\end{equation}
Applying Fourier series expansion to $\widetilde{H}(t)$ and $|\phi_\alpha(t)\rangle$
\begin{equation}
    \widetilde{H}(t) = \sum_{m\in\mathbb{Z}} \widetilde{H}_m e^{im\omega t},\quad |\phi_\alpha(t)\rangle = \sum_{m\in\mathbb{Z}} |\phi_\alpha^m\rangle e^{im\omega t},
\end{equation}
with the frequency $\omega := 2\pi/2T=\pi/T$, the solution of \cref{eq:Hamitonain_problems} can be rewritten as \cite{shirley1965solution,verdeny2016quasi}:
\begin{equation}\label{eq:standard_solution}
    |\psi(t)\rangle=\sum_{l\in\mathbb{Z}}e^{-il\omega t}\langle l| e^{-i\mathscr{H}_{\text{eff}}t}|0\rangle|\psi(0)\rangle. 
\end{equation}
with the effective Hamiltonian defined by
\begin{equation}\label{eq:def_of_effective_ham}
    \mathscr{H}_{\text{eff}}=\sum_{l\in\mathbb{Z}}\left|l\right\rangle\left\langle l \right|\otimes\left(\widetilde{H}_{0}-l\omega\right)+\sum_{l,m\in\mathbb{Z}}\left|l\right \rangle\left\langle l+m\right|\otimes \widetilde{H}_{m}.
\end{equation}
Note that this reduces the time-dependent Hamiltonian simulation problem on the finite-dimensional Hilbert space~$\mathscr{H}$ described by \cref{eq:Hamitonain_problems,eq:origin_problems} to a time-independent problem on the infinite-dimensional Floquet-Hilbert space $\{|l\rangle\}_{l\in \mathbb{Z}}\otimes\mathscr{H}$ given in \cref{eq:standard_solution,eq:def_of_effective_ham}.
\subsection{Optimal Hamiltonian simulation for analytic periodic systems}
\cite{mizuta2023optimal} presented an optimal-complexity algorithm for Hamiltonian simulation that exploits \cref{eq:standard_solution,eq:def_of_effective_ham} in \cref{subsec:Floquet_theory}, under the condition that the Hamiltonian takes the form
\begin{equation}\label{eq:type_of_Hami_control_analytic}
H(t) = \sum_{j=0}^{j_{\max}-1} \alpha_j(t) M_j,\quad \alpha_j(t+T)=\alpha_j(t),\;\;\text{for any}\;0\leq j\leq j_{\max}-1,
\end{equation}
and each $\alpha_j(t)$ is analytic, which means that the Fourier coefficients of $H(t)$ is exponentially decaying:
\begin{equation}\label{eq:exponential_decay_of_FC}
\exists\; h,\zeta >0,\;\text{s.t.}\quad||H_m|| \leq he^{-\frac{|m|}{\zeta}}\quad\forall  m \in \mathbb{Z}.
\end{equation}
In addition, the optimality of the complexity also requires the scenario to be simulated is adiabatic-like. i.e. $t/T\in\mathcal{O}(1)$, where $t$ is the evolution time to be simulated.\\
To achieve this, \cite{mizuta2023optimal} first derived the Lieb-Robinson bound \cite{colmenarez2020lieb,gong2022bounds,harper2020topology} on the transition rate
\begin{equation}\label{eq:Lieb_Robinson_bound_for_analytic}
\left\| \langle l| e^{-i\mathscr{H}_{\text{ eff}}^{l_\text{max}}t} |l'\rangle \right\| \leq \exp\left(-\frac{|l-l'|-2\beta\zeta' ht}{\zeta'}+2/\beta\right), \quad\text{for}\;|l|,|l'| \leq l_{\max}.
\end{equation}
Where $\beta$ and $\zeta'$ are positive constants, $l_{\max}\geq1$ is the truncation order for the Floquet-Hilbert space to be determined, $D^{l_{\max}} = \{-l_{\max} + 1, -l_{\max} + 2, \ldots, l_{\max}\}$ and
\begin{equation}
\mathscr{H}_{\text{ eff}}^{l_{\max}} = \sum_{l \in D^{l_{\max}}} |l\rangle \langle l| \otimes (H_0 - l\omega) + \sum_{l \in D^{l_{\max}}} \sum_{\substack{m \neq 0; \\ l+m \in D^{l_{\max}}}} |l\rangle \langle l+m| \otimes H_m.
\end{equation}
Then use \cref{eq:Lieb_Robinson_bound_for_analytic} to get
\begin{equation}\label{eq:Fouri_truncate_estima_analytic}
\left\| |\psi(t)\rangle - |\psi^{l_{\max}}(t)\rangle \right\| \leq 4\zeta' e^{2\beta h t - (l_{\max}-1)/\zeta'+2/\beta} ,
\end{equation}
where $|\psi(t)\rangle$ is the exact solution as in \cref{eq:standard_solution}, and
\begin{equation}\label{eq:def_of_psi_lmax}
|\psi^{l_{\max}}(t)\rangle = \sum_{l \in D^{l_{\max}}} e^{-i l \omega t} \langle l| e^{-i \mathscr{H}_{\text{ eff}}^{l_{\max}} t} |0\rangle |\psi(0)\rangle,
\end{equation}
Let the right side of \cref{eq:Fouri_truncate_estima_analytic} smaller than $\varepsilon$ to get the truncation order
\begin{equation}
    l_{\text{max}}= \left\lceil 2\beta\zeta' h t + \zeta' \log(1/\varepsilon) + \zeta'  \log(4\zeta')+\frac{2\zeta'}{\beta}+1 \right\rceil \in \Theta\left(  t + \log(1/\varepsilon) \right). 
\end{equation}
The most straightforward way to implement \cref{eq:def_of_psi_lmax} is using
\begin{equation}
\frac{|\psi^{l_{\max}}(t)\rangle}{{\left\||\psi^{l_{\max}}(t)\rangle\right\|}}=\frac{\langle a^{l_{\max}}|\Psi^{l_{\max}}(t)\rangle}{\left\|\langle a^{l_{\max}}|\Psi^{l_{\max}}(t)\rangle\right\|},
\end{equation}
where
\begin{equation}
|a^{l_{\max}}\rangle =\frac{1}{\sqrt{2l_{\max}}}\sum_{l\in D^{l_{\max}}}|l\rangle,\quad\;
|\Psi^{l_{\max}}(t)\rangle=e^{-i\omega t\sum_l|l\rangle\langle l|} e^{-i \mathscr{H}_{\text{ eff}}^{l_{\max}} t} |0\rangle |\psi(0)\rangle.
\end{equation}
But $\left\|\langle a^{l_{\max}}|\Psi^{l_{\max}}(t)\rangle\right\|^2\approx\langle \psi(t)|\psi(t)\rangle/(2l_{\max})=1/(2l_{\max})$. The low success probability would destroy the algorithm’s optimal complexity. To overcome this difficulty, \cite{mizuta2023optimal} used the approximate translation symmetry in a slightly larger truncated Floquet-Hilbert space
\begin{equation}
\left\| \langle l | e^{-i \mathscr{H}_{\text{eff}}^{4l_{\max}}t} | l' \rangle - e^{i l' \omega t} \langle l \ominus l' | e^{-i \mathscr{H}_{\text{eff}}^{4l_{\max}}t} | 0 \rangle \right\| \leq 2e^{2\beta h t-(8l_{\max}-|l|-|l'|)/\zeta'+2/\beta },
\end{equation}
to get
\begin{equation}
\left\| \langle 0| \mathscr{U}_{\text{amp1}}^{l_{\max}}(t)|0\rangle |\psi(0)\rangle - \frac{1}{2} |\psi(t)\rangle \right\|=\mathcal{O}(\varepsilon),
\end{equation}
where
\begin{align}\label{eq:def_of_U_ampl1_analytic}
\mathscr{U}_{\text{amp1}}^{l_{\max}}(t) 
&= (\mathscr{U}_{\text{ini}}^{4l_{\max}})^{\dagger}  e^{-it \sum_{l\in D^{4l_{\max}}}l\omega| l \rangle\langle l |}  e^{-i \mathscr{H}_{\text{eff}}^{4l_{\max}} t}\mathscr{U}_{\text{ini}}^{l_{\max}},
\end{align}
\begin{equation}
\mathscr{U}_{\text{ini}}^{4l_{\max}} = \left(|a^{4l_{\max}}\rangle \langle 0| + \ldots\right) \otimes I.
\end{equation}
Then the $\mathcal{O}(\varepsilon)$ approximation of solution $|\psi(t)\rangle$ is followed by oblivious
amplitude amplification (OAA) \cite{BerryChildsCleveEtAl2014,manenti2023quantum}:
\begin{equation}
\left\| \mathscr{U}_{\text{amp2}}^{l_{\max}}(t) |0\rangle |\psi(0)\rangle - |0\rangle |\psi(t)\rangle \right\| =\mathcal{O}(\varepsilon),
\end{equation}
where
\begin{align}
\mathscr{R} &= (2\,|0\rangle\langle 0|-I)\otimes I, \\
\mathscr{U}_{\text{ amp2}}^{l_{\max}}(t) &= -\mathscr{U}_{\text{ amp1}}^{l_{\max}}(t)\mathscr{R}[\mathscr{U}_{\text{ amp1}}^{l_{\max}}(t)]^{\dagger}\mathscr{R}\mathscr{U}_{\text{ amp1}}^{l_{\max}}(t).\label{eq:OAA_analytic}
\end{align}
The last thing to do is implementing $e^{-it \sum_{l\in D^{4l_{\max}}}l\omega| l \rangle\langle l |}$ and $e^{-i \mathscr{H}_{\text{eff}}^{4l_{\max}}t}$ in \cref{eq:def_of_U_ampl1_analytic}.
To implement $\mathscr{H}_{\text{LP}}^{4l_{\max}}:=\sum_{l\in D^{4l_{\max}}}l\omega|l \rangle\langle l|\,\otimes I,$
\cite{mizuta2023optimal} used the relation
\begin{equation}\label{eq:BE_of_H_LP_analytic}
\langle a^{4l_{\max}} | \langle0|_{\text{single}}(\text{Comp}^\dagger Z_{\text{single}}  \text{Comp}\,\otimes I)\, | a^{4l_{\max}} \rangle|0\rangle_{\text{single}} = \frac{\mathscr{H}_{\text{LP}}^{4l_{\max}}}{4l_{\max}\omega},
\end{equation}
where Comp is defined as
\begin{equation}
\text{Comp}\, |m,l\rangle |0\rangle_{\text{single}}  := 
\begin{cases}
|m,l\rangle |0\rangle_{\text{single}}  & \text{if } l \ge m, \\[4pt]
|m,l\rangle |1\rangle_{\text{single}}  & \text{if } l < m,
\end{cases}  
\end{equation}
and  $Z_{\text{single}}  $ denotes a Pauli $Z$ operator on the single qubit system.
To implement $\mathscr{H}_{\text{ eff}}^{4l_{\max}}$, \cite{mizuta2023optimal} introduced the refined effective Hamiltonian $\mathscr{H}_{\text{eff, pbc}}^{4l_{\max}}$ given by 
\begin{equation}\label{eq:refined_Ham_analytic_case}
\mathscr{H}_{\text{ eff, pbc}}^{4l_{\max}} = \mathscr{H}_{\text{ eff}}^{4l_{\max}} + \sum_{(l,m) \in \partial \tilde{F}^{4l_{\max}}} |l\rangle \langle l \oplus m| \otimes H_{m} + \text{h.c.}=\sum_{m \in D^{4l_{\max}}} \text{Add}_{m}^{4l_{\max}} \otimes H_{-m}- \mathscr{H}_{\text{LP}}^{4l_{\max}}
\end{equation}
with $\partial \tilde{F}^{4l_{\max}} = \{(l,m)\mid l \in D^{4l_{\max}},\; 8l_{\max} - l + 1 \leq m \leq 8l_{\max} - 1\}$ and  $\text{Add}_{m}^{4l_{\max}}=\sum_{l\in D^{4l_{\max}}}|l\oplus m \rangle\langle l|$, which also provides the exact time‑evolved state $|\psi(t)\rangle$ as
\begin{equation}
\left\| \langle 0 | \mathscr{U}_{\text{ amp1, pbc}}^{l_{\max}}(t) | 0 \rangle | \psi(0) \rangle - \frac{1}{2} |\psi(t)\rangle \right\| =\mathcal{O}(\varepsilon),\quad
\left\| \mathscr{U}_{\text{ amp2, pbc}}^{l_{\max}}(t) | 0 \rangle | \psi(0) \rangle - |0\rangle |\psi(t)\rangle \right\|  =\mathcal{O}(\varepsilon),
\end{equation}
where $\mathscr{U}_{\text{amp1,pbc}}^{l_{\max}}(t)$ and $\mathscr{U}_{\text{amp2,pbc}}^{l_{\max}}(t)$ are obtained by replacing $\mathscr{H}_{\text{eff}}^{4l_{\max}}$ in \cref{eq:def_of_U_ampl1_analytic,eq:OAA_analytic} with $\mathscr{H}_{\text{eff,pbc}}^{4l_{\max}}$. For the form of Hamiltonian in \cref{eq:type_of_Hami_control_analytic}, we know
\begin{equation}\label{eq:analytic_H_m_formula}
H_m=\sum_{j=0}^{j_{\max}-1}\alpha_j^mM_j,\quad\alpha_j^m=\frac{1}{T}\int_0^Tdt\alpha_j(t)e^{-im\omega t}.
\end{equation}
Substitute \cref{eq:analytic_H_m_formula} into the leftmost expression of \cref{eq:refined_Ham_analytic_case} and use linear combination of block-encoded matrices (LCU) \cite{GilyenSuLowEtAl2019,Chakraborty2024}, we obtain an approximated block encoding of $\mathscr{H}_{\text{ eff, pbc}}^{4l_{\max}}$ with the scaling factor being $\mathcal{O}(l_{\max}\omega)$ (the leading term $l_{\max}\omega$ arises from the block encoding scaling factor of $\mathscr{H}_{\text{LP}}^{4l_{\max}}$ in \cref{eq:BE_of_H_LP_analytic}). Then using quantum singular value transformation (QSVT) \cite{GilyenSuLowEtAl2019} we can implement $e^{-it \sum_{l\in D^{4l_{\max}}}l\omega| l \rangle\langle l |}$ and $e^{-i \mathscr{H}_{\text{eff}}^{4l_{\max}}t}$ efficiently. It follows immediately that the asymptotic query complexity of $M_j$'s for implementing $e^{-i \mathscr{H}_{\text{eff}}^{4l_{\max}}t}$ is 
\begin{equation}\label{eq:over_all_query_analytic}
\mathcal{O}\left((l_{\max}\omega) t+\log(1/\varepsilon)\right)=\mathcal{O}\left(l_{\max}+\log(1/\varepsilon)\right)=\mathcal{O}\left(t+\log(1/\varepsilon)\right)
\end{equation}
under the adiabatic-like case, i.e., $t/T\in\mathcal{O}(1)$. Since $\mathscr{U}_{\text{amp2,pbc}}^{l_{\max}}(t)$ contains $\mathcal{O}(1)$ instances of $e^{-i\mathscr{H}_{\text{eff}}^{4l_{\max}}t}$, \cref{eq:over_all_query_analytic} also gives the overall asymptotic query complexity of the algorithm.

\section{Periodic extension of Gevrey functions}\label{sec:Periodic_Extension}
In this section, we develop the key technical tool underlying our algorithm: a periodic extension of a nonperiodic Gevrey-class Hamiltonian defined on a finite interval. As discussed in the introduction, Floquet theory applies to time-periodic Hamiltonians and allows us to reduce the time-dependent simulation problem to a time-independent one on the Floquet-Hilbert space. However, the original Hamiltonian $H(s)$ in \cref{eqn:ham_sim} is defined only on $[0,1]$ and is not periodic in general. A naive periodic continuation would introduce discontinuities in the derivatives at the boundaries, which destroys the rapid decay of Fourier coefficients and renders the truncation of the Floquet space inefficient.

To overcome this obstacle, we construct a smooth periodic extension of $H(s)$ to $\hat{H}(s)$ on $[0,2]$ such that $\hat{H}^{(n)}(0)=\hat{H}^{(n)}(2)$ for all $n\ge 0$, while preserving the Gevrey regularity up to a controlled loss. This extension is an explicit ingredient of the quantum algorithm, as the extended Hamiltonian serves as the input to the Floquet-based simulation protocol. The main result of this section is \cref{th:Gevrey_ext}, which establishes the existence of such an extension. Its proof relies on a Gevrey-class cut-off function and a Gevrey version of the Borel lemma, which we develop first.

We begin with a purely combinatorial estimate that will be used in the subsequent constructions.

\begin{lemma}\label[lemma]{lem:esti_of_zuheshu}
Let $l,k, n$ be integers satisfying $1\leq k \leq n$ and $0\leq l \leq n - k$. Then
\begin{equation}
l! \, k^{\,n-k-l} \leq e^{\,n} \, (n-k)!.
\end{equation}
\end{lemma}
\begin{proof}
Set $m = n - k$. Then $m \ge 0$ and $l \le m$. The inequality becomes
\begin{equation}
l! \, k^{\,m-l} \le e^{\,m+k} \, m!.
\end{equation}
For fixed $m, k$, define 
\begin{equation}
\varphi(l) = \frac{l! \, k^{\,m-l}}{m!}, \qquad 0 \le l \le m.
\end{equation}
Consider the ratio
\begin{equation}
\frac{\varphi(l)}{\varphi(l-1)} = \frac{l}{k}, \qquad 1 \le l \le m.
\end{equation}
Hence if $l < k$, then $\varphi(l) < \varphi(l-1)$, so $\varphi$ is decreasing. If $l > k$, then $\varphi(l) > \varphi(l-1)$, so $\varphi$ is increasing. Thus the maximum of $\varphi(l)$ on $0 \le l \le m$ is attained at an endpoint:
\begin{equation}
\max_{0 \le l \le m} \varphi(l) = \max\left\{ \frac{k^{\,m}}{m!},\; 1 \right\}\leq e^k\leq e^{m+k}.
\end{equation}
\end{proof}

A central ingredient in our extension procedure is a smooth cut-off function that transitions from $0$ to $1$ in a Gevrey-regular manner. Unlike analytic functions, Gevrey functions can have compact support while still exhibiting controlled growth of their derivatives. The following lemma, adapted from standard Gevrey theory, provides such a function with explicit bounds on its derivatives. These bounds are crucial for controlling the norms of the extended Hamiltonian and its derivatives in later complexity analyses.

\begin{lemma}[Gevrey-$\tau$ cut-off function]\label[lemma]{lem:cut_off function}
For any $1<\tau<2$ be fixed, there exists a function $\chi(t)\in G^{\tau}(\mathbb{R})$ satisfying 
\begin{equation}
|\chi^{(n)}(t)|\leq C_{\tau,1}8^nC_{\tau,2}^{n-1}(n!)^{\tau} \quad\text{for any }\; t\in\mathbb{R},\;n\ge 1,
\end{equation}
where
\begin{equation}
C_{\tau,1}=\frac{e^2}{2}, \qquad C_{\tau,2}=\frac{4e^{\tau+1}}{\tau-1},
\end{equation}
and in addition
\begin{equation}
0 \leq \chi \leq 1, \quad \chi \equiv 1 \quad \text{on} \quad \left[ -\frac{1}{2}, \frac{1}{2} \right], \quad \text{and} \quad \text{supp}(\chi) \subset [-1, 1].
\end{equation}
\end{lemma}
\begin{proof}
Define
\begin{equation}
\theta(t):=
\begin{cases}
\exp\!\left(-t^{\frac{1}{1-\tau}}\right), & t>0,\\[4mm]
0, & t\le 0.
\end{cases}
\end{equation}
We are going to show that $\theta\in G^{\tau}(\mathbb{R})$. Note that we only need to analyse the behaviour of $\theta(t)$ on $[0,\infty)$. Evidently $\theta(t)\in C^{\infty}(0,\infty)\cap C[0,\infty)$. For every $n\ge 1$ and $t>0$ we have
\begin{equation}
\theta^{(n)}(t)=\exp\!\left(-t^{\frac{1}{1-\tau}}\right)\left(\sum_{k=1}^{n}C_{n,k}t^{\frac{k}{1-\tau}-n}\right),
\end{equation}
where 
\begin{align}
|C_{n,k}|
&\leq\binom{n-1}{k-1}\left(\frac{1}{\tau-1}\right)^k\left(\frac{k}{\tau-1}+k\right)\left(\frac{k}{\tau-1}+k+1\right)\cdots\left(\frac{k}{\tau-1}+n-1\right)\notag\\ 
&\leq\binom{n-1}{k-1}\left(\frac{1}{\tau-1}\right)^k\cdot\sum_{l=0}^{n-k}\left(\binom{n-k}{l}\left(\frac{k}{\tau-1}\right)^{n-k-l}(n-1)(n-2)\cdots(n-l)\right)\notag \\ 
&=\binom{n-1}{k-1}\left(\frac{1}{\tau-1}\right)^k\cdot\sum_{l=0}^{n-k}\left(\binom{n-k}{l}\left(\frac{k}{\tau-1}\right)^{n-k-l}\frac{(n-1)!}{(n-l-1)!l!}l!\right)\notag\\ 
&\leq2^{n}\binom{n-1}{k-1}\left(\frac{1}{\tau-1}\right)^k\cdot\sum_{l=0}^{n-k}\binom{n-k}{l}\left(\frac{k}{\tau-1}\right)^{n-k-l}l!\notag\\ 
&\leq2^{n}\binom{n-1}{k-1}\left(\frac{1}{\tau-1}\right)^n\cdot\sum_{l=0}^{n-k}\binom{n-k}{l}k^{n-k-l}l!\notag\\ 
&\leq2^{n}\binom{n-1}{k-1}\left(\frac{1}{\tau-1}\right)^n\cdot\sum_{l=0}^{n-k}\binom{n-k}{l}k^{n-k-l}l!\notag\\ 
&\leq(2e)^{n}\binom{n-1}{k-1}\left(\frac{1}{\tau-1}\right)^n(n-k)!\cdot\sum_{l=0}^{n-k}\binom{n-k}{l}\notag\\
&=2^{2n-k}e^n\binom{n-1}{k-1}\left(\frac{1}{\tau-1}\right)^n(n-k)!.
\end{align}
In the last inequality we used \cref{lem:esti_of_zuheshu}. Hence for any $t>0$ and $n\geq 1$,
\begin{align}
|\theta^{(n)}(t)|
&\leq \exp\!\left(-t^{\frac{1}{1-\tau}}\right)\left(\sum_{k=1}^{n}|C_{n,k}|t^{\frac{k}{1-\tau}-n}\right)\notag\\
&\leq \exp\!\left(-t^{\frac{1}{1-\tau}}\right)\left(\sum_{k=1}^{n}2^{2n-k}e^n\binom{n-1}{k-1}\left(\frac{1}{\tau-1}\right)^n(n-k)!t^{\frac{k}{1-\tau}-n}\right)\notag\\
&= \left(\frac{4e}{\tau-1}\right)^n\left(\sum_{k=1}^{n}2^{-k}\binom{n-1}{k-1}(n-k)!\exp\!\left(-t^{\frac{1}{1-\tau}}\right)t^{\frac{k}{1-\tau}-n}\right).    
\end{align}
Obviously 
\begin{equation}
\lim_{t \to 0^+}\theta^{(n)}(t)=0.
\end{equation}
So $\theta(t)\in C^{\infty}(\mathbb{R})$. In addition we use the inequality
\begin{equation}
e^{-m}m^d\leq d^{d}e^{-d} \quad \text{for any}\; m,d>0
\end{equation}
to get that for any $t>0,\;1\leq k\leq n$,
\begin{align}
\exp\!\left(-t^{\frac{1}{1-\tau}}\right)t^{\frac{k}{1-\tau}-n}&=\exp\!\left(-(1/t)^{\frac{1}{\tau-1}}\right)(1/t)^{\frac{k+n(\tau-1)}{\tau-1}}\notag\\
&\leq(k+n(\tau-1))^{k+n(\tau-1)}e^{-k-n(\tau-1)}\notag\\
&\leq \frac{1}{2}(k^{k+n(\tau-1)}+(n(\tau-1))^{k+n(\tau-1)})\left(\frac{2}{e}\right)^{k+n(\tau-1)}\notag\\
&= \frac{1}{2}(k^{k}k^{n(\tau-1)}+n^{k}n^{n(\tau-1)}((\tau-1))^{k+n(\tau-1)})\left(\frac{2}{e}\right)^{k+n(\tau-1)}\notag\\
&\leq \frac{1}{2}(e^{\tau k}k!{(n!)}^{\tau-1}+e^{\tau n}k!(n!)^{\tau-1}(\tau-1))\left(\frac{2}{e}\right)\notag\\
&\leq \frac{\tau}{e}e^{\tau n}k!(n!)^{\tau-1}.
\end{align}
Then for any $n\geq 1,t\in\mathbb{R}$ we have
\begin{align}
|\theta^{(n)}(t)|
&\leq \frac{\tau}{e}\left(\frac{4e^{\tau+1}}{\tau-1}\right)^n(n!)^{\tau-1}\left(\sum_{k=1}^{n}2^{-k}\binom{n-1}{k-1}(n-k)!k!\right)  \notag\\  
&\leq \frac{\tau}{2e}\left(\frac{4e^{\tau+1}}{\tau-1}\right)^n(n!)^{\tau}.
\end{align}
In the last inequality we used that for $1\leq k\leq n$,
\begin{equation}
\sum_{k=1}^{n}2^{-k}\binom{n-1}{k-1}(n-k)!k!\leq \frac{n!}{2}.
\end{equation}
This proves that $\theta(t)\in G^{\tau}(\mathbb{R})$. Then $\theta(2+t)\cdot\theta(2-t)\geq0$ is a Gevrey-$\tau$ function compactly supported on $[-2, 2]$ and satisfies
\begin{equation}
|(\theta(2+t)\theta(2-t))^{(n)}|\leq (n+1)\left(\frac{4e^{\tau+1}}{\tau-1}\right)^n(n!)^{\tau}\quad \text{for any}\;n\geq 0.
\end{equation}
Let
\begin{equation}
\chi(t)=\frac{\int_{-2}^{6-8|t|}\theta(2+m)\theta(2-m)dm}{\int_{-2}^2\theta(2+m)\theta(2-m)dm}\leq \frac{e^2}{2}\int_{-2}^{6-8|t|}\theta(2+m)\theta(2-m)dm.
\end{equation}
Then $\chi(t)\in G^{\tau}(\mathbb{R})$ satisfies
\begin{equation}
|\chi^{(n)}(t)|\leq \frac{e^2}{2}\cdot8^n\left(\frac{4e^{\tau+1}}{\tau-1}\right)^{n-1}(n!)^{\tau} \quad\text{for any }\; t\in\mathbb{R},\;n\ge 1,
\end{equation}
and
\begin{equation}
0 \leq \chi \leq 1, \quad \chi \equiv 1 \quad \text{on} \quad \left[ -\frac{1}{2}, \frac{1}{2} \right], \quad \text{and} \quad \text{supp}(\chi) \subset [-1, 1].
\end{equation}
\end{proof}

With the cut-off function in hand, we turn to the Borel's lemma, which is a fundamental result in Gevrey theory. In its classical form, Borel's lemma states that any prescribed sequence of derivatives at a point can be realized by a smooth function. Here we need a Gevrey-regular version, which additionally controls the growth of the derivatives of the resulting function in terms of the growth of the prescribed sequence. The following theorem provides such a result for scalar-valued functions. It is the key to constructing the transition segment of the periodic extension from the boundary derivative data of the original Hamiltonian.

\begin{theorem}[Gevrey version of the Borel lemma]\label{th:Borel_lemma} For any $1<\tau<2$, let $ \chi_\tau(t) \in G_0^\tau (\mathbb{R}) $ be as stated in \cref{lem:cut_off function}. Let $\{a_n\}_{n=0}^{\infty}$ be a sequence of real numbers satisfying
\begin{equation}
|a_n|\le CD^n(n!)^{\sigma}\quad\text{for all }\; n\ge 0,
\end{equation}
where $C>0,D\geq 1$ and $\sigma\ge 1$ are constants. Then there exists a function $f_\tau\in G_{0}^{\sigma-1+\tau}(\mathbb{R})$ such that $f^{(n)}(0)=a_n$ for every $n\ge 0$, and
\begin{equation}
|f_{\tau}(t)|\leq 2C\quad\text{for any }\;t\in\mathbb{R};
\end{equation}
\begin{equation}
|f_{\tau}^{(n)}(t)|\leq4CC_{\tau,1}(32e^{\sigma-1+1/e}  D)^{n}(2C_{\tau,2})^{n-1}(n!)^{\sigma-1+\tau}\quad\text{for any }\; t\in\mathbb{R},\;n\ge 1.
\end{equation}
\end{theorem}
\begin{proof}
For $ C_1 > 0 $ (to be determined later) we set
\begin{equation}
R_0=C_1D\;,\;R_j = C_1\cdot D(j!)^{(\sigma -1)/j}\;\;\text{for}\;j\geq1 .
\end{equation}
Hence, for all $ j\geq 0 $ we have
\begin{equation}
\frac{|a_j|}{j!R_j^j} \leq  C\cdot C_1^{-j}.
\end{equation}
We define the function
\begin{equation}
f_{\tau}(t) = \sum_{j=0}^\infty f_{j,\tau}(t)=\sum_{j=0}^\infty  \frac{a_j}{j!}t^j\chi_\tau(R_j t) .
\end{equation}
For convenience we will omit the subscript $\tau$ on each $f$ in the following proof. Note that $f_j(t)\in C^{\infty}(\mathbb{R})$ for any $j\geq 0$. For any $ n \geq0 $ we have 
\begin{equation}
f_j^{(n)}(t)=\sum_{k=0}^n \binom{n}{k}\left(\chi_\tau(R_j t)\right)^{(n-k)} (t^j)^{(k)} \frac{a_j}{j!}.
\end{equation}
Denote 
\begin{equation}
f_{j,k}^{n}(t)=\binom{n}{k}\left(\chi_\tau(R_j t)\right)^{(n-k)} (t^j)^{(k)}\frac{a_j}{j!}.
\end{equation}
Then for $\text{any}\;j\geq0,\; n\geq 1 ,\; 0\leq k\leq n-1$ and $t\in \mathbb{R}$ we have 
\begin{align}
|f_{j,k}^{n}(t)|
&=\left|\binom{n}{k}\left(\chi_\tau(R_j t)\right)^{(n-k)} (t^j)^{(k)} \frac{a_j}{j!}\right|\notag\\
&=\left|\mathbf{1}_{ \{ j\geq k\}} \binom{n}{k}R_j^{n-k}\left(\chi_\tau^{(n-k)}(R_j t)\right)\frac{j!}{(j-k)!} t^{j-k}\frac{a_j}{j!}\right|\notag\\ 
&\leq \mathbf{1}_{ \{ j\geq k\}} \binom{n}{k}\binom{j}{k}k!R_j^{n-k}\left|\chi_\tau^{(n-k)}(R_j t)\right||t|^{j-k}\frac{|a_j|}{j!}.
\end{align}
Note that for any $0\leq k < n $,
\begin{equation}
\chi_{\tau}^{(n-k)}(R_j t) \equiv 0 \quad \text{for } |t| \leq \frac{1}{2R_j} \;\text{or} \; |t| \geq \frac{1}{R_j}.
\end{equation}
Hence
\begin{align}
|f_{j,k}^{n}(t)|
&\leq2^{k}\mathbf{1}_{ \{ j\geq k\}} \binom{n}{k}\binom{j}{k}k!R_j^{n}\left|\chi_\tau^{(n-k)}(R_j t)\right||t|^{j}\frac{|a_j|}{j!}\notag\\ 
&\leq C_{\tau,1}2^{k}\mathbf{1}_{ \{ j\geq k\}} \binom{n}{k}\binom{j}{k}k!R_j^{n}8^{n-k}C_{\tau,2}^{n-k-1}((n-k)!)^{\tau}\frac{|a_j|}{j!R_j^j}\notag\\ 
&\leq CC_{\tau,1}2^{k}\mathbf{1}_{ \{ j\geq k\}} \binom{n}{k}\binom{j}{k}k!R_j^{n}8^{n-k}C_{\tau,2}^{n-k-1}((n-k)!)^{\tau}C_1^{-j}\notag\\ 
&\leq CC_{\tau,1}8^{n}(2C_{\tau,2})^{n-1}\mathbf{1}_{ \{ j\geq k\}} \binom{n}{k}\binom{j}{k}k!R_j^{n}((n-k)!)^{\tau}C_1^{-j}\notag\\ 
&\leq CC_{\tau,1}(8C_1 D)^{n}(2C_{\tau,2})^{n-1}\mathbf{1}_{ \{ j\geq k\}} \binom{n}{k}\binom{j}{k}k!(j!)^{(\sigma -1)n/j}((n-k)!)^{\tau}C_1^{-j}\notag\\
&\leq CC_{\tau,1}(8C_1 D)^{n}(2C_{\tau,2})^{n-1}\mathbf{1}_{ \{ j\geq k\}} \binom{n}{k}\binom{j}{k}k!(j)^{(\sigma -1)n}((n-k)!)^{\tau}C_1^{-j}\notag\\
&\leq CC_{\tau,1}(8C_1 D)^{n}(2C_{\tau,2})^{n-1}\mathbf{1}_{ \{ j\geq k\}} \binom{n}{k}2^jk!(n!)^{\sigma-1} e^{(\sigma-1) j} ((n-k)!)^{\tau}C_1^{-j}\notag\\ 
&=CC_{\tau,1}(8C_1 D)^{n}(2C_{\tau,2})^{n-1}(n!)^{\sigma-1} \binom{n}{k}k! ((n-k)!)^{\tau}\mathbf{1}_{ \{ j\geq k\}}\left(\frac{2e^{\sigma-1} }{C_1}\right)^{j}.
\end{align}
Similarly for $\text{any}\;j\geq 0,\; n\geq 0 ,\; k= n$ and $t\in \mathbb{R}$ we have
\begin{align}
|f_{j,k}^{n}(t)|
&=|f_{j,n}^{n}(t)|\notag\\
&=\left|\left(\chi_\tau(R_j t)\right) (t^j)^{(n)} \frac{a_j}{j!}\right|\notag\\
&=\left|\mathbf{1}_{ \{ j\geq n\}} \chi_\tau(R_j t)\frac{j!}{(j-n)!} t^{j-n}\frac{a_j}{j!}\right|\notag\\ 
&\leq \mathbf{1}_{ \{ j\geq n\}} \binom{j}{n}n!\left|\chi_\tau(R_j t)\right||R_j|^{n}\frac{|a_j|}{j!R_j^j}\notag\\ 
&\leq C\mathbf{1}_{ \{ j\geq n\}} \binom{j}{n}n!R_j^{n}C_1^{-j}\notag\\ 
&\leq C(C_1D)^n\mathbf{1}_{ \{ j\geq n\}} \binom{j}{n}n!(j!)^{(\sigma -1)n/j}C_1^{-j}\notag\\ 
&\leq C(C_1D)^n\mathbf{1}_{ \{ j\geq n\}} \binom{j}{n}n!(j)^{(\sigma -1)n}C_1^{-j}\notag\\ 
&\leq C(C_1D)^n\mathbf{1}_{ \{ j\geq n\}} \binom{j}{n}(n!)^{\sigma}e^{(\sigma-1)j}C_1^{-j}\notag\\ 
&\leq C(C_1D)^n(n!)^{\sigma}\mathbf{1}_{ \{ j\geq n\}} \left(\frac{2e^{\sigma-1} }{C_1}\right)^{j}.
\end{align}
Hence for any $n\geq 1$ and $t\in \mathbb{R}$,
\begin{align}
&\sum_{j=0}^{\infty}|f_j^{(n)}(t)| \notag \\
&\quad\leq \sum_{j=0}^{\infty}\left(\sum_{k=0}^{n}|f_{j,k}^{n}(t)|\right) \notag\\
&\quad= \sum_{j=0}^{\infty}\left(\sum_{k<n}|f_{j,k}^{n}(t)|+|f_{j,n}^{n}(t)|\right)  \notag\\
&\quad\leq \sum_{j=0}^{\infty}\Bigg( \sum_{k<n}CC_{\tau,1}(8C_1 D)^{n}(2C_{\tau,2})^{n-1}(n!)^{\sigma-1} 
        \binom{n}{k}k! ((n-k)!)^{\tau}\mathbf{1}_{\{j\geq k\}}
        \left(\frac{2e^{\sigma-1}}{C_1}\right)^{\!j} \notag\\
&\qquad\qquad +C(C_1D)^n(n!)^{\sigma}\mathbf{1}_{\{j\geq n\}} 
        \left(\frac{2e^{\sigma-1}}{C_1}\right)^{\!j} \Bigg)  \notag\\
&\quad= CC_{\tau,1}(8C_1 D)^{n}(2C_{\tau,2})^{n-1}(n!)^{\sigma-1} 
        \sum_{j=0}^{\infty}\sum_{k<n}\binom{n}{k}k! ((n-k)!)^{\tau}
        \mathbf{1}_{\{j\geq k\}}\left(\frac{2e^{\sigma-1}}{C_1}\right)^{\!j} \notag\\
&\qquad +C(C_1D)^n(n!)^{\sigma}
        \sum_{j=0}^{\infty}\mathbf{1}_{\{j\geq n\}} 
        \left(\frac{2e^{\sigma-1}}{C_1}\right)^{\!j} \notag\\
&\quad\leq CC_{\tau,1}(8C_1 D)^{n}(2C_{\tau,2})^{n-1}(n!)^{\sigma-1} 
        \sum_{k<n}\binom{n}{k}k! ((n-k)!)^{\tau}
        \sum_{j=k}^{\infty}\left(\frac{2e^{\sigma-1}}{C_1}\right)^{\!j} \notag\\
&\qquad +C(C_1D)^n(n!)^{\sigma}
        \sum_{j=n}^{\infty}\left(\frac{2e^{\sigma-1}}{C_1}\right)^{\!j}.
\end{align}
Pick $ C_1 =4e^{\sigma-1} $ so that
\begin{equation}
\sum_{j=0}^{\infty} \left( \frac{2e^{\sigma-1}}{C_1} \right)^j = 2.
\end{equation}
Then
\begin{align}
\sum_{j=0}^{\infty}|f_j^{(n)}(t)| 
&\leq \sum_{j=0}^{\infty}\sum_{k=0}^{n}|f_{j,k}^{n}(t)| \notag\\
&\leq 2CC_{\tau,1}(8C_1 D)^{n}(2C_{\tau,2})^{n-1}(n!)^{\sigma-1} \sum_{k<n}\binom{n}{k}k! ((n-k)!)^{\tau}+2C(C_1D)^n(n!)^{\sigma}\notag\\
&\leq 2CC_{\tau,1}(8C_1 D)^{n}(2C_{\tau,2})^{n-1}(n!)^{\sigma-1} \sum_{k<n}\binom{n}{k}(k! (n-k)!)^{\tau}+2C(C_1D)^n(n!)^{\sigma}\notag\\
&\leq 2CC_{\tau,1}(8e^{1/e}C_1 D)^{n}(2C_{\tau,2})^{n-1}(n!)^{\sigma-1+\tau} +2C(C_1D)^n(n!)^{\sigma}\notag\\
&\leq 4CC_{\tau,1}(32e^{\sigma-1+1/e}  D)^{n}(2C_{\tau,2})^{n-1}(n!)^{\sigma-1+\tau} <\infty.
\end{align}
For $n= 0$ and $t\in \mathbb{R}$,
\begin{align}
\sum_{j=0}^{\infty}|f_j^{(n)}(t)| 
&\leq \sum_{j=0}^{\infty}\left|f_{j,n}^{n}(t)\right| \notag\\
&\quad\leq \sum_{j=0}^{\infty}C(C_1D)^n(n!)^{\sigma}\mathbf{1}_{ \{ j\geq n\}} \left(\frac{2e^{\sigma-1} }{C_1}\right)^{j}\notag\\
&\quad\leq C(C_1D)^n(n!)^{\sigma}\sum_{j=0}^{\infty}\left(\frac{2e^{\sigma-1} }{C_1}\right)^{j}\notag\\
&\quad= 2C(C_1D)^n(n!)^{\sigma}=2C.
\end{align}
Therefore by Weierstrass M-test, the function series $f(t)=\sum_{j=0}^{\infty}f_j(t) $ converges uniformly on $\mathbb{R}$, and $f(t)\in G_0^{\sigma-1+\tau}(\mathbb{R})$  satisfies
\begin{equation}
f^{(n)}(t)=\sum_{j=0}^{\infty}f^{(n)}_j(t) \quad t\in \mathbb{R},\;n\geq0,
\end{equation}
and specifically
\begin{equation}
f^{(n)}(0)=\sum_{j=0}^{\infty}f^{(n)}_j(0)=a_n \quad n\geq0.
\end{equation}
\end{proof}

Since our Hamiltonian is matrix-valued rather than scalar-valued, we require a matrix analogue of the Gevrey Borel lemma. The following corollary extends \cref{th:Borel_lemma} to matrix-valued functions by applying the scalar result entrywise. This extension is straightforward but essential for our purposes.

\begin{corollary}\label[corollary]{co:matrix_Borel_lemma} For any $1<\tau<2$, let $ \chi_\tau(t) \in G_0^\tau (\mathbb{R}) $ be as stated in \cref{lem:cut_off function}. Let $\{A_n\}_{n=0}^{\infty}$ be a sequence of matrix in $\mathbb{C}^{d\times d} $ satisfying
\begin{equation}
||A_n||\le CD^n(n!)^{\sigma}\quad\text{for all }\; n\ge 0,
\end{equation}
where $||\cdot||$ denotes the matrix 2-norm, $C>0,D\geq 1$ and $\sigma\ge 1$ are constants. Then there exists a marix valued function $f_\tau$ such that $f^{(n)}(0)=A_n$ for any $n\ge 0$, and
\begin{equation}
||f(t)||\leq 2C\quad\text{for any }\;t\in\mathbb{R};
\end{equation}
\begin{equation}
||f^{(n)}(t)||\le 4CC_{\tau,1}(32e^{\sigma-1+1/e}  D)^{n}(2C_{\tau,2})^{n-1}(n!)^{\sigma-1+\tau}\quad\text{for any }\;t\in\mathbb{R},\; n\ge 1.
\end{equation}
\end{corollary}
\begin{proof} 
We define the function
\begin{equation}
f_{\tau}(t) = \sum_{j=0}^\infty f_{j,\tau}(t)=\sum_{j=0}^\infty  \frac{A_j}{j!}t^j\chi_\tau(R_j t) ,
\end{equation}
where
\begin{equation}
R_0=4e^{\sigma-1}D\;,\;R_j = 4e^{\sigma-1}\cdot D(j!)^{(\sigma -1)/j}\;\;\text{for}\;j\geq1 .
\end{equation}
Note that $||A||_{\max}\leq||A||\leq d\cdot||A||_{\max}$ for any $A\in\mathbb{C}^{d\times d}$, where $||\cdot||_{\max}$ denotes the maximum modulus norm. Hence we have
\begin{equation}
(f_{\tau}(t))_{l,m} =\sum_{j=0}^\infty  \frac{(A_j)_{l,m}}{j!}t^j\chi_\tau(R_j t) ,
\end{equation}
and
\begin{equation}
|(A_n)_{l,m}|\le CD^n(n!)^{\sigma}\quad\text{for all }\; n\ge 0,\;1\leq l,m\leq d.
\end{equation}
Then by \cref{th:Borel_lemma} we know that $f_{\tau}(t)= \sum_{j=0}^\infty f_{j,\tau}(t)$ converges uniformly on $\mathbb{R}$, infinitely differentiable and $f^{(n)}(0)=A_n$ for every $n\ge 0$. In addition, for any $ n \geq0 $ and $t\in \mathbb{R}$ we have 
\begin{equation}
||f_j^{(n)}(t)||\leq\sum_{k=0}^n \binom{n}{k}|\left(\chi_\tau(R_j t)\right)^{(n-k)} (t^j)^{(k)}| \frac{||A_j||}{j!}.
\end{equation}
By exactly the same method as in \cref{th:Borel_lemma}, we know that  
\begin{equation}
||f(t)||\leq 2C\quad\text{for any }\;t\in\mathbb{R};
\end{equation}
\begin{equation}
||f^{(n)}(t)||\le 4CC_{\tau,1}(32e^{\sigma-1+1/e}  D)^{n}(2C_{\tau,2})^{n-1}(n!)^{\sigma-1+\tau}\quad\text{for any }\;t\in\mathbb{R},\; n\ge 1.
\end{equation}
\end{proof} 

We are now in a position to state and prove the main result of this section. Given a Gevrey-$\sigma$ Hamiltonian $H(s)$ on $[0,1]$, we construct a 2-periodic extension whose derivatives match at the endpoints to all orders. The construction proceeds by defining an auxiliary function $\varphi(s)$ on $[0,1]$ that smoothly interpolates between the derivative data at $s=1$ and those at $s=0$, using the Borel lemma and the cut-off function. The extension is then obtained by setting $\hat{H}(s)=\varphi(s-1)$ on $[1,2]$, with period $2$. The resulting extension belongs to the Gevrey class $G^{\sigma+\tau-1}$ for any $1<\tau<2$, with explicit bounds on its derivatives that will be essential for the Fourier decay estimates in \cref{sec:Truncation_order_Four}.

\begin{theorem}[Existence of Gevrey Periodic Extension]\label{th:Gevrey_ext}
If a time-dependent Hamiltonian $H(s)\in \mathbb{C}^{d\times d}$ defined on $[0,1]$ belongs to the Gevrey-$\sigma$ class with $1\le \sigma<2$, i.e., there exist constants $C>0,D\geq 1 $ such that
\begin{equation}
\|H^{(n)}(s)\| \le C D^n (n!)^{\sigma}, \quad \text{for any }\; s\in[0,1]\;, n\geq 0,
\end{equation}
then for any $1<\tau<2$, $H(s)$ can be periodically extended to $\hat{H}(s)$ defined on $[0,2]$ so that the extension lies in $G^{\sigma+\tau-1}$ class. More precisely, $\hat{H}(s)$ satisfies
\begin{equation}
||\hat{H}(s)||\le C_{\text{\normalfont ex},1} \quad\text{for any }\; s\in[0,2],
\end{equation}
\begin{equation}
||\hat{H}^{(n)}(s)||\le C_{\text{\normalfont ex},2}A_1^{n}A_2^{n-1}(n!)^{\sigma-1+\tau}\quad\text{for any }\; s\in[0,2],\;n\ge 1,
\end{equation}
where
\begin{equation}\label{eq:def_of_C_and_A}
C_{\text{\normalfont ex},1}=2C,\;\;\;C_{\text{\normalfont ex},2}=2e^2 C,\;\;\;A_1=32e^{\sigma-1+1/e}D,\;\;\;A_2=\frac{8e^{\tau+1}}{\tau-1}; 
\end{equation}
and in addition
\begin{equation}
\hat{H}^{(n)}(0)=\hat{H}^{(n)}(2)\quad\text{for any }\;n\ge 0.
\end{equation}
\end{theorem}
\begin{proof} 
Define $\varphi: [0,1] \to \mathbb{C}^{d\times d}$ by
\begin{equation}
\varphi(s) =F_{\tau}(s) + G_{\tau}(s-1),
\end{equation}
where 
\begin{equation}
F_{\tau}(s) = \sum_{j=0}^\infty  \frac{H^{(j)}(1)}{j!}s^j\chi_\tau(R_j s) ,
\end{equation}
\begin{equation}
G_{\tau}(s) = \sum_{j=0}^\infty  \frac{H^{(j)}(0)}{j!}s^j\chi_\tau(R_j s) ,
\end{equation}
and $ \chi_\tau(s) \in G_0^\tau (\mathbb{R}) $ and $R_j$ are given in \cref{lem:cut_off function,th:Borel_lemma}. Since $\chi_{\tau}(s)$ is compactly supported on $[-1,1]$, we have that for any $n\geq0$,
\begin{equation}
\varphi^{(n)}(0)=F_{\tau}^{(n)}(0)=H^{(n)}(1),
\end{equation}
\begin{equation}
\varphi^{(n)}(1)=G_{\tau}^{(n)}(0)=H^{(n)}(0).
\end{equation}
By \cref{co:matrix_Borel_lemma} we have that for any $n\geq1$, $s\in[0,1]$,
\begin{align}
\|\varphi^{(n)}(s)\|\leq\| F_{\tau}^{(n)}(s) \|\vee\|G_{\tau}^{(n)}(t-1)\| \leq 4CC_{\tau,1}(32e^{\sigma-1+1/e}D)^{n}(2C_{\tau,2})^{n-1}(n!)^{\sigma-1+\tau}.
\end{align}
Recall that 
\begin{equation}
C_{\tau,1}=\frac{e^2}{2},\qquad C_{\tau,2}=\frac{4e^{\tau+1}}{\tau-1}.
\end{equation}
Then for any $n\geq1$, $t\in[0,1]$,
\begin{align}
||\varphi^{(n)}(s)||\leq 2e^2 C(32e^{\sigma-1+1/e}D)^{n}\left(\frac{8e^{\tau+1}}{\tau-1} \right)^{n-1}(n!)^{\sigma-1+\tau}.
\end{align}
For $n=0$ and $s\in[0,1]$, we have
\begin{equation}
||\varphi(s)||\leq 2C.
\end{equation}
Define $\hat{H}(s)=H(s)$ on $ s\in[0,1]$ and $\hat{H}(s)=\varphi(s-1)$ on $ s\in[1,2]$. Since the matrix values and all derivatives match at $s=1$, $\hat{H}(s) \in G^{\sigma-1+\tau}[0,2]$. Moreover,
\begin{equation}
\hat{H}^{(n)}(0)=\hat{H}^{(n)}(2)\quad\text{for any }\;n\ge 0.
\end{equation}
\end{proof}

\section{Truncation order for Fourier coefficients of Gevrey Hamiltonians}\label{sec:Truncation_order_Four}

In the preceding section, we constructed a periodic Gevrey extension of the original Hamiltonian $H(s)$. This extension enables the application of Floquet theory, which reformulates the time-dependent simulation problem as a time-independent one on the infinite-dimensional Floquet-Hilbert space $\mathbb{C}^{\mathbb{Z}}\otimes\mathscr{H}$. However, for practical quantum simulation, we must truncate this infinite space to a finite dimension, and the central question becomes how large the truncation order must be to achieve a desired simulation accuracy. 

The answer heavily depends on the decay rate of the Fourier coefficients of $\widetilde{H}(t)$. The faster these coefficients decay, the more aggressively we can truncate the Floquet space. For analytic periodic Hamiltonians, the Fourier coefficients decay exponentially, leading to a polylogarithmic truncation order. For the Gevrey-class Hamiltonians considered in this work, the decay is only subexponential, which nevertheless suffices to yield a nearly optimal query complexity.

This section is devoted to establishing the subexponential decay of the Fourier coefficients of Gevrey-class periodic Hamiltonians and determining the corresponding truncation order of the Floquet-Hilbert space. The main results are \cref{th:Four_decay_est}, which quantifies the decay rate, and \cref{th:Floquet-Hilbert_space_truncation}, which translates this decay into an explicit bound on the truncation error and hence the required truncation order $l_{\max}$.

\subsection{Sub-exponential decay of Fourier coefficients}\label{sec:Four_decay}
We begin by analyzing the Fourier coefficients of a Gevrey-class periodic Hamiltonian in this section. Firstly we establish the following elementary lemma, which provides a criterion for determining when the function $f(m)=\log m / m^{1/\varrho}$ falls below a given threshold, and will be used in \cref{th:Four_decay_est}.

In the following context we define $\varrho:=\sigma-1+\tau$. 
\begin{lemma}\label[lemma]{lem:m_up_est}
Let $\varrho > 1$ and $C > 0$. Define $f(m) = \dfrac{\text{\normalfont log}\, m}{m^{1/\varrho}}$ and let 
\begin{equation}
C_0(\varrho) = f(2) = \frac{\text{\normalfont log}\, 2}{2^{1/\varrho}}.
\end{equation}
Define 
\begin{equation}
m_{\text{\normalfont up}} = \left( \frac{\varrho}{C} \left( \text{\normalfont log}\frac{\varrho}{C} + \text{\normalfont log}\,\text{\normalfont log}\frac{\varrho}{C} + 1 \right) \right)^\varrho.
\end{equation}
If $C \le C_0(\varrho)$, then $f(m) < C$ for any $m \geq\lceil m_{\text{\normalfont up}} \rceil $.
\end{lemma}
\begin{proof}   
Let $t = \dfrac{\varrho}{C}$. Since $C \le C_0(\varrho) < \dfrac{\varrho}{e}$ (because the maximum of $f(m)$ is $\varrho/e$), we have $t > e$. Now
\begin{equation}
m_{\text{up}}^{1/\varrho} = t \left( \text{log}\, t + \text{log}\,\text{log}\, t + 1 \right).
\end{equation}
Compute $f(m_{\text{up}})$:
\begin{align}
f(m_{\text{up}}) 
&= \frac{\text{log}\, m_{\text{up}}}{m_{\text{up}}^{1/\varrho}} 
= \frac{\varrho \,\text{log}\!\left( t (\text{log}\, t + \text{log}\,\text{log}\, t + 1) \right)}{t (\text{log}\, t + \text{log}\,\text{log}\, t + 1)} \notag \\
&= C \cdot \frac{ \text{log}\!\left( t (\text{log}\, t + \text{log}\,\text{log}\, t + 1) \right) }{ \text{log}\, t + \text{log}\,\text{log}\, t + 1 }.
\end{align}
It suffices to show that
\begin{equation}
\frac{ \text{log}\!\left( t (\text{log}\, t + \text{log}\,\text{log}\, t + 1) \right) }{ \text{log}\, t + \text{log}\,\text{log}\, t + 1 } < 1,
\end{equation}
which is equivalent to
\begin{equation}
\text{log}\!\left( t (\text{log}\, t + \text{log}\,\text{log}\, t + 1) \right) < \text{log}\, t + \text{log}\,\text{log}\, t + 1.
\end{equation}
Expanding the left-hand side,
\begin{equation}
\text{log}\, t + \text{log}\!\left( \text{log}\, t + \text{log}\,\text{log}\, t + 1 \right) < \text{log}\, t + \text{log}\,\text{log}\, t + 1,
\end{equation}
so we need
\begin{equation}
\text{log}\!\left( \text{log}\, t + \text{log}\,\text{log}\, t + 1 \right) < \text{log}\,\text{log}\, t + 1.
\end{equation}
Set $u = \text{log}\, t$; because $t > e$, we have $u > 1$. Then the inequality becomes
\begin{equation}
\text{log}\!\left( u + \text{log}\, u + 1 \right) < \text{log}\, u + 1.
\end{equation}
Exponentiating both sides,
\begin{equation}
u + \text{log}\, u + 1 < u e,
\end{equation}
or equivalently
\begin{equation}
\text{log}\, u + 1 < u(e - 1).
\end{equation}
Define $g(u) = u(e-1) - \text{log}\, u - 1$. Its derivative is
\begin{equation}
g'(u) = (e-1) - \frac{1}{u}.
\end{equation}
For $u > 1$, $g'(u) > 0$; thus $g(u)$ is increasing on $[1,\infty)$. Since $g(1) = e - 2 > 0$, we have $g(u) > 0$ for all $u \ge 1$. Hence the inequality holds for every $u \ge 1$, and therefore
\begin{equation}
f(m_{\text{up}}) < C.
\end{equation}
Next, we verify that $m_{\text{up}} > e^{\varrho}$. Because $t > e$,
\begin{equation}
m_{\text{up}}^{1/\varrho} = t\left( \text{log}\, t + \text{log}\,\text{log}\, t + 1 \right) > e \cdot (1 + 0 + 1) = 2e > e,
\end{equation}
so indeed $m_{\text{up}} > e^{\varrho}$. Consequently, $f(m)$ is strictly decreasing for $m > e^{\varrho}$. Hence for any $m \ge m_{\text{up}}$ we have $f(m) \le f(m_{\text{up}}) < C$. Since $m \geq \lceil m_{\text{up}} \rceil  \geq m_{\text{up}}$, we obtain $f(m) < C$. Thus the chosen $m$ indeed satisfies the desired inequality.
\end{proof}

We now state and prove the main estimate on the Fourier coefficients of a Gevrey-class periodic Hamiltonian. The following theorem shows that if extended $\hat{H}(s)$ belongs to the Gevrey-$\varrho$ class with explicitly bounded derivatives, then the Fourier coefficients of its time-scaled version $\widetilde{H}(t)=\hat{H}(\frac{t}{T})$ decay subexponentially as $he^{-c|m|^{1/\varrho}}$. This decay rate is a direct consequence of the Gevrey regularity and is the key property that enables the efficient truncation of the Floquet-Hilbert space.

\begin{theorem}\label{th:Four_decay_est}
Let $\hat{H}(s)$ be a $2$-periodic Hamiltonian belonging to the Gevrey-$\varrho$ class ($ \varrho>1$). More precisely, there exist constants $C>0, A_1,A_2 > 1$ such that
\begin{equation}
\|\hat{H}(s)\| \le C_{\text{\normalfont ex},1}\quad\;\;\forall s\in[0,2];
\end{equation}
\begin{equation}
\|\hat{H}^{(n)}(s)\| \le C_{\text{\normalfont ex},2} A_1^nA_2^{n-1} (n!)^{\varrho} \quad \forall s\in[0,2]\;, n\geq 1.
\end{equation}
Define $A=A_1A_2$, and $\widetilde{H}(t)=\hat{H}(\frac{t}{T})\;(t\in[0,2T]\;,T>0)$ , hence $\widetilde{H}(t)$ is  $2T$-periodic , its Fourier coefficients are
\begin{equation}
\widetilde{H}_m = \frac{1}{2T} \int_0^{2T} \widetilde{H}(t) e^{-i  m\omega t } \, dt \quad \forall m \in \mathbb{Z},
\end{equation}
where $\omega=\frac{\pi}{T}$. Then there exists $h=h(\varrho,A_1,A_2,C_{\text{\normalfont ex},1},C_{\text{\normalfont ex},2})>0$ only depends on $\varrho,A_1,A_2,C_{\text{\normalfont ex},1},C_{\text{\normalfont ex},2}$ and $\zeta=\zeta(\varrho,A)>0 $ only depends on $\varrho,A$ such that
\begin{equation}
||\widetilde{H}_m|| \leq he^{-\frac{|m|^{\frac{1}{\varrho}}}{\zeta}}\quad \forall m \in \mathbb{Z}.
\end{equation}
\end{theorem}
\begin{proof}  
Given $m \neq0$, for any integer $n\geq1$, integration by parts $n$ times (boundary terms cancel due to periodicity) gives
\begin{equation}
\widetilde{H}_m  =  \left( \frac{1}{i m\omega} \right)^n \frac{1}{2T}\int_0^{2T} \widetilde{H}^{(n)}(t) e^{-i  m\omega t } \, dt =  \left( \frac{1}{im\omega T} \right)^n\frac{1}{2T} \int_0^{2T} \hat{H}^{(n)}\left(\frac{t}{T}\right) e^{-i  m\omega t } \, dt.
\end{equation}
Hence
\begin{equation}
||\widetilde{H}_m|| \leq \left( \frac{1}{|m|\omega T} \right)^n\frac{1}{2T} \int_0^{2T} ||\hat{H}^{(n)}\left(\frac{t}{T}\right)||  \, dt\leq \frac{C_{\text{ex},2}A^n(n!)^{\varrho}}{A_2(|m|\pi)^n}.
\end{equation}
Stirling's inequality states that for $n\geq1$,
\begin{equation}
n! < \sqrt{2\pi n} \left( \frac{n}{e} \right)^n e^{\frac{1}{12n}}< 2\sqrt{2\pi n} \left( \frac{n}{e} \right)^n.
\end{equation}
For any $|m|\geq\left\lceil  \frac{A}{\pi} \left( \text{log}(\frac{A}{\pi})^{\frac{1}{\varrho}} + \text{log}\,\text{log}(\frac{A}{\pi})^{\frac{1}{\varrho}} + 1 \right)^\varrho \right\rceil$, let $n=\left\lfloor(\frac{\pi|m|}{A})^{\frac{1}{\varrho}}\right\rfloor\geq 1$. We have $\frac{A n^{\varrho}}{\pi |m|} \leq 1$ and then
\begin{align}
||\widetilde{H}_m||
&\leq 2(2\pi)^{\frac{\varrho}{2}}A_2^{-1}C_{\text{ex},2}n^{\frac{\varrho}{2}}\left( \frac{An^{\varrho}}{\pi e^{\varrho}|m|} \right)^n \notag \\
&\leq 
2(2\pi)^{\frac{\varrho}{2}}A_2^{-1}C_{\text{ex},2}n^{\frac{\varrho}{2}}e^{-\varrho n} \notag \\
&= 2(2\pi)^{\frac{\varrho}{2}}A_2^{-1}C_{\text{ex},2}\cdot\text{exp}(-\varrho(n-\frac{1}{2}\text{log}n)) \notag \\
&\leq 2(2\pi)^{\frac{\varrho}{2}}A_2^{-1}C_{\text{ex},2}\cdot\text{exp}\left(-\varrho\left(\left(\frac{\pi}{A}\right)^{\frac{1}{\varrho}}|m|^{\frac{1}{\varrho}}-\frac{1}{2\varrho}\text{log}|m|-\frac{1}{2\varrho}\text{log}\left(\frac{\pi}{A}\right)-1\right)\right) \notag \\
&\leq 2(2\pi)^{\frac{\varrho}{2}}A_2^{-1}C_{\text{ex},2}\cdot\text{exp}\left(-\varrho\left(\frac{1}{2}\left(\frac{\pi}{A}\right)^{\frac{1}{\varrho}}|m|^{\frac{1}{\varrho}}-\frac{1}{2\varrho}\text{log}\left(\frac{\pi}{A}\right)-1\right)\right) \notag \\
&= 2^{1+\frac{\varrho}{2}}\pi^{\frac{1+\varrho}{2}}e^{\varrho}A_2^{-\frac{3}{2}}A_1^{-\frac{1}{2}}C_{\text{ex},2}\cdot\text{exp}\left(-\frac{\varrho}{2}\left(\frac{\pi}{A}\right)^{\frac{1}{\varrho}}|m|^{\frac{1}{\varrho}}\right).
\end{align}
In the last inequality we used the fact
\begin{equation}
\frac{\text{log}|m|}{|m|^{\frac{1}{\varrho}}}\leq\varrho\left(\frac{\pi}{A}\right)^{\frac{1}{\varrho}},
\end{equation}
which followed by \cref{lem:m_up_est}. For $\left\lceil\frac{A}{\pi}\right\rceil\leq |m|<\left\lceil  \frac{A}{\pi} \left( \text{log}(\frac{A}{\pi})^{\frac{1}{\varrho}} + \text{log}\,\text{log}(\frac{A}{\pi})^{\frac{1}{\varrho}} + 1 \right)^\varrho \right\rceil$, still let $n=\left\lfloor(\frac{\pi|m|}{A})^{\frac{1}{\varrho}}\right\rfloor\geq 1$, we have
\begin{align}
||\widetilde{H}_m|| 
&\leq 2(2\pi)^{\frac{\varrho}{2}}A_2^{-1}C_{\text{ex},2}\cdot\text{exp}\left(-\varrho\left(\left(\frac{\pi}{A}\right)^{\frac{1}{\varrho}}|m|^{\frac{1}{\varrho}}-\frac{1}{2\varrho}\text{log}|m|-\frac{1}{2\varrho}\text{log}\left(\frac{\pi}{A}\right)-1\right)\right) \notag \\
&\leq 2^{1+\frac{\varrho}{2}}\pi^{\frac{1+\varrho}{2}}e^{\varrho}A_2^{-\frac{3}{2}}A_1^{-\frac{1}{2}}C_{\text{ex},2}|m|^{\frac{1}{2}}\cdot\text{exp}\left(-\frac{\varrho}{2}\left(\frac{\pi}{A}\right)^{\frac{1}{\varrho}}|m|^{\frac{1}{\varrho}}\right) \notag \\
&\leq 2^{1+\frac{\varrho}{2}}\pi^{\frac{1+\varrho}{2}}e^{\varrho}A_2^{-\frac{3}{2}}A_1^{-\frac{1}{2}}C_{\text{ex},2}\left(\left\lceil  \frac{A}{\pi} \left( \text{log}(\frac{A}{\pi})^{\frac{1}{\varrho}} + \text{log}\,\text{log}(\frac{A}{\pi})^{\frac{1}{\varrho}} + 1 \right)^\varrho \right\rceil\right)^{\frac{1}{2}}\cdot\text{exp}\left(-\frac{\varrho}{2}\left(\frac{\pi}{A}\right)^{\frac{1}{\varrho}}|m|^{\frac{1}{\varrho}}\right).
\end{align}
For $ 0<|m|<\left\lceil \frac{A}{\pi} \right\rceil $, we have
\begin{align}
||\widetilde{H}_m|| 
&\leq \frac{C_{\text{ex},2}A_1}{|m|\pi}\cdot 1 \notag \\
&= \frac{C_{\text{ex},2}A_1}{|m|\pi}\cdot\text{exp}\left(\frac{\varrho}{2}\left(\frac{\pi}{A}\right)^{\frac{1}{\varrho}}|m|^{\frac{1}{\varrho}}\right)\text{exp}\left(-\frac{\varrho}{2}\left(\frac{\pi}{A}\right)^{\frac{1}{\varrho}}|m|^{\frac{1}{\varrho}}\right) \notag \\
&\leq \frac{C_{\text{ex},2}A_1}{|m|\pi}\cdot\text{exp}\left(\frac{\varrho}{2}\left(\frac{\pi}{A}\right)^{\frac{1}{\varrho}}\left( \frac{A}{\pi} \right)^{\frac{1}{\varrho}}\right)\text{exp}\left(-\frac{\varrho}{2}\left(\frac{\pi}{A}\right)^{\frac{1}{\varrho}}|m|^{\frac{1}{\varrho}}\right) \notag \\
&\leq \frac{C_{\text{ex},2}A_1}{|m|\pi}\cdot\text{exp}\left(\frac{\varrho}{2}\right)\text{exp}\left(-\frac{\varrho}{2}\left(\frac{\pi}{A}\right)^{\frac{1}{\varrho}}|m|^{\frac{1}{\varrho}}\right).
\end{align}
For $m=0$ we have $||\widetilde{H}_m||\leq C_{\text{ex},1}$. Consequently, let 
\begin{equation}
h_1=h_1(\varrho,A_1,A_2,C_{\text{ex},2})=2^{1+\frac{\varrho}{2}}\pi^{\frac{1+\varrho}{2}}e^{\varrho}A_2^{-\frac{3}{2}}A_1^{-\frac{1}{2}}C_{\text{ex},2},
\end{equation}
\begin{equation}
h_2=h_2(\varrho,A_1,A_2,C_{\text{ex},2})=2^{1+\frac{\varrho}{2}}\pi^{\frac{1+\varrho}{2}}e^{\varrho}A_2^{-\frac{3}{2}}A_1^{-\frac{1}{2}}C_{\text{ex},2}\left(\left\lceil  \frac{A}{\pi} \left( \text{log}(\frac{A}{\pi})^{\frac{1}{\varrho}} + \text{log}\,\text{log}(\frac{A}{\pi})^{\frac{1}{\varrho}} + 1 \right)^\varrho \right\rceil\right)^{\frac{1}{2}},
\end{equation}
\begin{equation}
h_3=h_3(\varrho,A_1,C_{\text{ex},2})=\frac{C_{\text{ex},2}}{\pi}A_1\cdot\text{exp}\left(\frac{\varrho}{2}\right),
\end{equation}
\begin{equation}
h_4=h_4(C_{\text{ex},1})=C_{\text{ex},1},
\end{equation}
and
\begin{equation}
\zeta=\zeta(\varrho,A)=\frac{2}{\varrho}\left(\frac{A}{\pi}\right)^{\frac{1}{\varrho}}.
\end{equation}
Then
\begin{equation}\label{eq:decay_of_H_m}
||\widetilde{H}_m||\leq
\begin{cases}
h_1e^{-\frac{|m|^{\frac{1}{\varrho}}}{\zeta}}, & |m|\geq\left\lceil  \frac{A}{\pi} \left( \text{log}(\frac{A}{\pi})^{\frac{1}{\varrho}} + \text{log}\,\text{log}(\frac{A}{\pi})^{\frac{1}{\varrho}} + 1 \right)^\varrho \right\rceil,\\[4mm]
h_2e^{-\frac{|m|^{\frac{1}{\varrho}}}{\zeta}}, & \left\lceil\frac{A}{\pi}\right\rceil\leq |m|<\left\lceil  \frac{A}{\pi} \left( \text{log}(\frac{A}{\pi})^{\frac{1}{\varrho}} + \text{log}\,\text{log}(\frac{A}{\pi})^{\frac{1}{\varrho}} + 1 \right)^\varrho \right\rceil,\\[4mm]
h_3\frac{1}{|m|}e^{-\frac{|m|^{\frac{1}{\varrho}}}{\zeta}}\leq h_3e^{-\frac{|m|^{\frac{1}{\varrho}}}{\zeta}}, & 0<|m|<\left\lceil \frac{A}{\pi} \right\rceil,\\[4mm]
h_4e^{-\frac{|m|^{\frac{1}{\varrho}}}{\zeta}}, & m=0.
\end{cases}
\end{equation}
Let 
\begin{equation}
h=h(\varrho,A_1,A_2,C_{\text{ex},1},C_{\text{ex},2})=h_1\vee h_2\vee h_3\vee h_4,
\end{equation}
then we have
\begin{equation}
||\widetilde{H}_m|| \leq he^{-\frac{|m|^{\frac{1}{\varrho}}}{\zeta}}\quad\forall  m \in \mathbb{Z}.
\end{equation}
\end{proof}
Even though we could use $||\widetilde{H}_m|| \leq he^{-\frac{|m|^{\frac{1}{\varrho}}}{\zeta}}$ as an estimator, we want to use piece-wise estimation of $||\widetilde{H}_m||$ in \cref{eq:decay_of_H_m} to get a better complexity order latter.

\subsection{Truncation order analysis}\label{sec:Error_analysis}
We want to determine the proper truncation order $ l_{\max} $ for the Floquet-Hilbert space. To this aim, we begin with deriving the Lieb-Robinson bound on the transition rate. First we will need the following lemma in the subsequent proof.
\begin{lemma}\label[lemma]{lam:est_of_Szeta}
Fix $\varrho > 1$, for the sum
\begin{equation}
S(\zeta) = \sum_{m=0}^{\infty} e^{-m^{1/\varrho}/\zeta}, \quad \zeta > 0,
\end{equation}
we have
\begin{equation}
\Gamma(\varrho+1)\,\zeta^{\varrho} \;\le\; S(\zeta) \;\le\; 1 + \Gamma(\varrho+1)\,\zeta^{\varrho}  \quad \text{for any}\;\zeta>0.
\end{equation}
\end{lemma}
\begin{proof}  
Define the function $f(x) = e^{-x^{1/\varrho}/\zeta}$ for $x \ge 0$.
Since $f$ is decreasing, for any integer $N \ge 1$ we have
\begin{equation}
\int_{0}^{N} f(x)\,dx \le \sum_{m=0}^{N-1} f(m) \le f(0) + \int_{0}^{N-1} f(x)\,dx.
\end{equation}
Letting $N \to \infty$ yields
\begin{equation}\label{eq:medium_est_of_Szeta}
\int_{0}^{\infty} f(x)\,dx \le S(\zeta) \le 1 + \int_{0}^{\infty} f(x)\,dx.
\end{equation}
Now evaluate the integral
\begin{equation}
I(\zeta) := \int_{0}^{\infty} e^{-x^{1/\varrho}/\zeta}\,dx.
\end{equation}
Perform the substitution $u = x^{1/\varrho}$, so $x = u^{\varrho}$ and $dx = \varrho u^{\varrho-1}\,du$:
\begin{equation}
I(\zeta) = \int_{0}^{\infty} e^{-u/\zeta}\,\varrho u^{\varrho-1}\,du
        = \varrho \int_{0}^{\infty} u^{\varrho-1} e^{-u/\zeta}\,du.
\end{equation}
Next, set $v = u/\zeta$, hence $u = \zeta v$, $du = \zeta\,dv$:
\begin{align}
I(\zeta) &= \varrho \int_{0}^{\infty} (\zeta v)^{\varrho-1} e^{-v}\,\zeta\,dv \notag \\
        &= \varrho \zeta^{\varrho} \int_{0}^{\infty} v^{\varrho-1} e^{-v}\,dv \notag \\
        &= \varrho \zeta^{\varrho}\,\Gamma(\varrho) \notag \\
        &= \Gamma(\varrho+1)\,\zeta^{\varrho},
\end{align}
where we used the Gamma function identity $\varrho\Gamma(\varrho)=\Gamma(\varrho+1)$.
Insert this result into \cref{eq:medium_est_of_Szeta}:
\begin{equation}
\Gamma(\varrho+1)\,\zeta^{\varrho} \;\le\; S(\zeta) \;\le\; 1 + \Gamma(\varrho+1)\,\zeta^{\varrho}.
\end{equation}
\end{proof}
Consequently, as $\zeta \to \infty$,
\begin{equation}
S(\zeta) \;\sim\; \Gamma(\varrho+1)\,\zeta^{\varrho},
\end{equation}
which means that $S(\zeta)$ grows as $\zeta^{\varrho}$.

\begin{theorem}[Bound on transition rate]\label{th:transition_rate}
We assume $\|\widetilde{H}_m\| $ can be bounded as in \cref{sec:Periodic_Extension} and \cref{sec:Four_decay}. Define
\begin{equation}
D^{l_{\max}} = \{-l_{\max} + 1, -l_{\max} + 2, \ldots, l_{\max}\} \subset \mathbb{Z}.
\end{equation}
Then, for $ l, l' $ such that $ |l|, |l'| \leq l_{\max} $, the transition rate is bounded from above by
\begin{equation}
\left\| \langle l| e^{-i\mathscr{H}_{\text{\normalfont  eff}}^{l_{\max}}t} |l'\rangle \right\| \leq \exp\left(2\beta  t -\frac{|l-l'|^{\frac{1}{\varrho}}}{2\zeta}\right).
\end{equation}
Here, $\beta$ is a positive constant defined by 
\begin{equation}\label{eq:def_of_bata_in_l_max}
\beta = h_2\sum_{m\geq 1}e^{-m^{\frac{1}{\varrho}}/(2\zeta)}+h_3(\text{\normalfont log}(\left\lceil  A/\pi \right\rceil-1)+1)+h_4.
\end{equation}
and
\begin{equation}
\mathscr{H}_{\text{\normalfont eff}}^{l_{\max}} = \sum_{l \in D^{l_{\max}}} |l\rangle \langle l| \otimes (\widetilde{H}_0 - l\omega) + \sum_{l \in D^{l_{\max}}} \sum_{\substack{m \neq 0; \\ l+m \in D^{l_{\max}}}} |l\rangle \langle l+m| \otimes \widetilde{H}_m.
\end{equation}
\end{theorem}
\begin{proof}  
We first employ the interaction picture. With the unitary operation defined by  
\begin{equation}
\mathscr{U}_0^{l_{\max}}(t) = e^{-i \mathscr{H}_0^{l_{\max}} t}, \quad \mathscr{H}_0^{l_{\max}} = \sum_{l \in D^{l_{\max}}} |l\rangle \langle l| \otimes (\widetilde{H}_0 - l\omega),
\end{equation}
the time evolution operator $ e^{-i \mathscr{H}_{\text{eff}}^{l_{\max}} t} $ is represented as  
\begin{equation}
e^{-i \mathscr{H}_{\text{eff}}^{l_{\max}} t} = \mathscr{U}_0^{l_{\max}}(t) \mathscr{U}_I^{l_{\max}}(t),
\end{equation}
\begin{equation}
\mathscr{U}^{l_{\max}}_I(t) = 
\mathcal{T} \exp \left( -i \int_0^t dt' \mathscr{H}^{l_{\max}}_I(t') \right).
\end{equation}
Here, the Hamiltonian in the interaction picture is defined by  
\begin{align}
\mathscr{H}_I^{l_{\max}}(t) &= \mathscr{U}_0^{l_{\max}}(t)^\dagger (\mathscr{H}_{\text{eff}}^{l_{\max}} - \mathscr{H}_0^{l_{\max}}) \mathscr{U}_0^{l_{\max}}(t) \notag \\
&= \sum_l \sum_{m>0} (e^{im\omega t} |l\rangle \langle l + m| \otimes H_m^I(t) + \text{h.c.}),    
\end{align}
\begin{equation}
H_{m}^I(t) = e^{i\widetilde{H}_0 t} \widetilde{H}_m e^{-i\widetilde{H}_0 t},
\end{equation}
where the summation in the second line is taken over $ l, m $ such that $ l, l+m \in D^{l_{\max}} $. Using the Dyson series expansion in the interaction picture we have
\begin{align}
\left\| \langle l| e^{-i\mathscr{H}_{\text{eff}}^{l_{\max} }t} |l'\rangle \right\| 
&=\left\| \langle l| \mathscr{U}_0^{l_{\max}}(t) \mathscr{U}_I^{l_{\max}}(t) 
 |l'\rangle \right\| \notag \\
 &=\left\|  \sum_{k=0}^{\infty}\frac{(-it)^k}{k!}\langle l|\left(\sum_{l \in D^{l_{\max}}} |l\rangle \langle l| \otimes (\widetilde{H}_0 - l\omega)\right)^k \mathscr{U}_I^{l_{\max}}(t) 
 |l'\rangle \right\| \notag \\
&=\left\|  \sum_{k=0}^{\infty}\frac{(-it)^k}{k!}\left( I \otimes (\widetilde{H}_0 - l\omega)\right)^k \langle l|\mathscr{U}_I^{l_{\max}}(t) 
 |l'\rangle \right\| \notag \\
&=\left\|  e^{-i\left( I \otimes (\widetilde{H}_0 - l\omega)\right)t} \langle l|\mathscr{U}_I^{l_{\max}}(t) 
|l'\rangle \right\| \notag \\
&=  \left\| \langle l | \mathscr{U}^{l_{\max}}_I(t) | l' \rangle \right\|    \notag \\
&\leq \sum_{n=0}^{\infty} \int_0^t dt_n \cdots \int_0^{t_2} dt_1 \left\| \sum_{\{l_i\}} \prod_{i=1}^n \langle l_i | \mathscr{H}_I^{l_{\max}}(t_i) | l_{i-1} \rangle \right\| \notag \\
&\leq \sum_{n=0}^{\infty} \int_0^t dt_n \cdots \int_0^{t_2} dt_1 \sum_{\{l_i\}} \prod_{i=1}^n \left\| \langle l_i | \mathscr{H}^{l_{\max}}_I(t_i) | l_{i-1} \rangle \right\|,
\end{align}
where we insert the identity $\sum_{l_i \in D^{l_{\max}}} |l_i\rangle \langle l_i| = I$ for $n - 1$ times. The summation $\sum_{\{l_i\}}$ is taken over $l_i \in D^{l_{\max}}$ for $i = 1, 2, \ldots, n - 1$ under fixed $l_0 = l'$ and $l_n = l$. We introduce a new variable $m_i = l_i - l_{i-1}$ instead of using $\{l_i\}$. Then the above integrand is bounded by
\begin{align}
F_n &:= \sum_{\{l_i\}} \prod_{i=1}^n \|\langle l_i| \mathscr{H}^{l_{\max}}_I(t_i)|l_{i-1}\rangle\| \notag \\
&\leq \sum_{\{m_i\}_{i=1}^n \in \mathbb{Z}^n} \delta_{l' + m_1 + \ldots + m_n, l} \prod_{i=1}^n \|\widetilde{H}_{m_i}\|.
\end{align}
Since $ l' + m_1 + \ldots + m_n = l $ implies $ |m_1| + \ldots + |m_n| \geq |l - l'| $, it can be further bounded as follows,
\begin{align}
F_n &\leq \sum_{\{m_i\}_{i=1}^n \in \mathbb{Z}^n} \mathbf{1}_{ \{ \sum_{i=1}^n |m_i| \geq |l - l'| \}}\prod_{i=1}^n \|\widetilde{H}_{m_i}\|\label{eq:medium_of_LR_bound} .
\end{align}
Let us evaluate the upper bound on the right side of \cref{eq:medium_of_LR_bound}. We denote $L_1=L_1(\varrho,A)=\left\lceil  \frac{A}{\pi} \left( \text{log}(\frac{A}{\pi})^{\frac{1}{\varrho}} + \text{log}\,\text{log}(\frac{A}{\pi})^{\frac{1}{\varrho}} + 1 \right)^\varrho \right\rceil,L_2=L_2(A)=\left\lceil  \frac{A}{\pi} \right\rceil$. Recall that $h_1\leq h_2$, $A/\pi>1$ under the choice of $A_1,A_2$ in \cref{th:Gevrey_ext}. Then 
\begin{align}
S_n(M) &:=\sum_{m_1, \ldots, m_n = 0}^{\infty} \mathbf{1}_{\{ \sum_{i=1}^n  m_i \ge M \}}\prod_{i=1}^n \|\widetilde{H}_{m_i}\| \notag \\
&\le \sum_{m_1, \ldots, m_n = 0}^{\infty} \mathbf{1}_{\{ (\sum_{i=1}^n  m_i)^{\frac1\varrho} \ge M^{\frac1\varrho} \}} \notag \\
&\qquad\qquad\cdot \prod_{i=1}^n \Bigl( h_1\mathbf{1}_{\{ m_i\ge L_1\}}
          + h_2\mathbf{1}_{\{L_2\le m_i<L_1\}}
          + h_3\frac{1}{ m_i}\mathbf{1}_{\{0< m_i<L_2\}}
          + h_4\mathbf{1}_{\{ m_i=0\}} \Bigr) 
          \cdot e^{-\sum_i m_i^{\frac1\varrho}/\zeta} \notag \\
&\le \sum_{m_1, \ldots, m_n = 0}^{\infty} \mathbf{1}_{\{ \sum_{i=1}^n  m_i^{\frac1\varrho} \ge M^{\frac1\varrho} \}} \notag \\
&\qquad\qquad\cdot \prod_{i=1}^n \Bigl( h_1\mathbf{1}_{\{ m_i\ge L_1\}}
          + h_2\mathbf{1}_{\{L_2\le m_i<L_1\}}
          + h_3\frac{1}{ m_i}\mathbf{1}_{\{0< m_i<L_2\}}
          + h_4\mathbf{1}_{\{ m_i=0\}} \Bigr) 
          \cdot e^{-\sum_i m_i^{\frac1\varrho}/\zeta} \notag \\
&\le \sum_{\substack{k_j \ge 0 \\ k_1+k_2+k_3+k_4=n}} 
        \binom{n}{k_1, k_2, k_3, k_4}
        \sum_{m_1, \ldots, m_n = 0}^{\infty} 
        \mathbf{1}_{\{ \sum_{i=1}^n  m_i^{\frac1\varrho} \ge M^{\frac1\varrho} \}} \notag \\
&\qquad\qquad\cdot \mathbf{1}_{\left\{\substack{ m_i\ge L_1 \\ 0< i\le k_1}\right\}}
        \mathbf{1}_{\left\{\substack{L_2\le m_i< L_1 \\ k_1< i\le k_1+k_2}\right\}}
        \mathbf{1}_{\left\{\substack{0< m_i< L_2 \\ k_1+k_2< i\le k_1+k_2+k_3}\right\}}
        \mathbf{1}_{\left\{\substack{ m_i=0 \\ k_1+k_2+k_3< i\le n}\right\}} \notag \\
&\qquad\qquad\cdot \bigl(h_1^{k_1}h_2^{k_2}h_3^{k_3}h_4^{k_4}\bigr)
        \Bigl(\prod_{i:\,k_1+k_2<i\le k_1+k_2+k_3}\frac{1}{ m_i}\Bigr)
        \cdot e^{-\sum_i m_i^{\frac1\varrho}/\zeta} \notag \\
&\le \sum_{\substack{k_j \ge 0 \\ k_1+k_2+k_3+k_4=n}} 
        \binom{n}{k_1, k_2, k_3, k_4}
        \sum_{m_1, \ldots, m_n = 0}^{\infty} 
        e^{\frac{1}{2\zeta}\bigl( \sum_{i=1}^n  m_i^{\frac1\varrho} - M^{\frac1\varrho} \bigr)} \notag \\
&\qquad\qquad\cdot \mathbf{1}_{\left\{\substack{ m_i\ge L_1 \\ 0< i\le k_1}\right\}}
        \mathbf{1}_{\left\{\substack{L_2\le m_i< L_1 \\ k_1< i\le k_1+k_2}\right\}}
        \mathbf{1}_{\left\{\substack{0< m_i< L_2 \\ k_1+k_2< i\le k_1+k_2+k_3}\right\}}
        \mathbf{1}_{\left\{\substack{ m_i=0 \\ k_1+k_2+k_3< i\le n}\right\}} \notag \\
&\qquad\qquad\cdot \bigl(h_1^{k_1}h_2^{k_2}h_3^{k_3}h_4^{k_4}\bigr)
        \Bigl(\prod_{i:\,k_1+k_2<i\le k_1+k_2+k_3}\frac{1}{ m_i}\Bigr)
        \cdot e^{-\sum_i m_i^{\frac1\varrho}/\zeta} \notag \\
&\le \Biggl( \sum_{\substack{k_j \ge 0 \\ k_1+k_2+k_3+k_4=n}} 
          \binom{n}{k_1, k_2, k_3, k_4} \notag \\
&\qquad\qquad\cdot \Bigl( h_1\sum_{ m\ge L_1} e^{-m^{\frac1\varrho}/(2\zeta)} \Bigr)^{\!k_1}
          \Bigl( h_2\sum_{L_2\le m<L_1} e^{-m^{\frac1\varrho}/(2\zeta)} \Bigr)^{\!k_2} \notag \\
&\qquad\qquad\cdot \Bigl( h_3\sum_{0< m<L_2} \frac{1}{ m} e^{-m^{\frac1\varrho}/(2\zeta)} \Bigr)^{\!k_3}
          \Bigl( h_4\sum_{ m=0} e^{-m^{\frac1\varrho}/(2\zeta)} \Bigr)^{\!k_4} \Biggr)
          e^{- M^{\frac1\varrho}/(2\zeta)} \notag \\
&= \Bigl( h_1\sum_{ m\ge L_1} e^{-m^{\frac1\varrho}/(2\zeta)}
          + h_2\sum_{L_2\le m<L_1} e^{-m^{\frac1\varrho}/(2\zeta)} \notag \\
&\qquad\qquad + h_3\sum_{0< m<L_2} \frac{1}{ m} e^{-m^{\frac1\varrho}/(2\zeta)}
          + h_4\sum_{ m=0} e^{-m^{\frac1\varrho}/(2\zeta)} \Bigr)^{\!n}
          e^{- M^{\frac1\varrho}/(2\zeta)} \notag \\
&\le \Bigl( h_2\sum_{ m\ge 1} e^{-m^{\frac1\varrho}/(2\zeta)}
          + h_3\sum_{0< m<L_2} \frac{1}{ m}
          + h_4 \Bigr)^{\!n}
          e^{- M^{\frac1\varrho}/(2\zeta)} \notag \\
&\le \Bigl( h_2\sum_{ m\ge 1} e^{-m^{\frac1\varrho}/(2\zeta)}
          + h_3\bigl( \log(\lceil A/\pi\rceil-1)+1 \bigr)
          + h_4 \Bigr)^{\!n}
          e^{- M^{\frac1\varrho}/(2\zeta)} \notag \\
&= \beta^{\,n} \, e^{-\frac{M^{\frac1\varrho}}{2\zeta}}.
\end{align}
$S_n(|l-l'|)$ accounts for the case where all $m_i$'s in the right side of \cref{eq:medium_of_LR_bound} are non-negative, The same upper bound holds for the remaining cases. Hence we can evaluate the integrand $F_n$ as
\begin{equation}
F_n\leq (2\beta )^n e^{-\frac{|l-l'|^{\frac{1}{\varrho}}}{2\zeta} }.
\end{equation}
Finally we arrive at the bound on the transition rate
\begin{equation}
\left\| \langle l| e^{-i\mathscr{H}_{\text{eff}}^{l_\text{max}}t} |l'\rangle \right\| \leq \sum_{n=0}^{\infty} \frac{t^n}{n!} (2\beta )^n e^{-\frac{|l-l'|^{\frac{1}{\varrho}}}{2\zeta} } \leq e^{2\beta  t -\frac{|l-l'|^{\frac{1}{\varrho}}}{2\zeta}}.
\end{equation}
\end{proof}

\begin{theorem}[Floquet-Hilbert space truncation]\label{th:Floquet-Hilbert_space_truncation}
We assume $\|\widetilde{H}_{m}\| $ can be bounded as in \cref{sec:Periodic_Extension} and \cref{sec:Four_decay}. Let $|\psi^{l_{\max}}(t)\rangle$ be
\begin{equation}
|\psi^{l_{\max}}(t)\rangle = \sum_{l \in D^{l_{\max}}} e^{-i l \omega t} \langle l| e^{-i \mathscr{H}_{\text{\normalfont eff}}^{l_{\max}} t} |0\rangle |\psi(0)\rangle.
\end{equation}
Then the exact time-evolved state $|\psi(t)\rangle$ is approximated by the truncated state $|\psi^{l_{\max}}(t)\rangle$ as
\begin{equation}
\left\| |\psi(t)\rangle - |\psi^{l_{\max}}(t)\rangle \right\| \leq 4\left(\sum_{m=0}^{\infty} e^{- \frac{m^{{\frac{1}{\varrho}}}}{4\zeta} }\right)e^{2\beta  t - \frac{l^{{\frac{1}{\varrho}}}_{\max}}{4\zeta} }.
\end{equation}
\end{theorem}
\begin{proof}   
We evaluate the convergence of $ |\psi^{l_{\max}}(t)\rangle $. For different orders $ l_{\max}, l'_{\max} $ satisfying $ l'_{\max} > l_{\max} $, we compute the difference
\begin{equation}
\left\| |\psi^{l'_{\max}}(t)\rangle - |\psi^{l_{\max}}(t)\rangle \right\| \leq \varepsilon_1 + \varepsilon_2,
\end{equation}
\begin{equation}
\varepsilon_1 = \sum_{l \in (D^{l'_{\max}} \setminus D^{l_{\max}})} \left\| \langle l| e^{-i \mathscr{H}_{\text{eff}}^{l'_{\max}} t} |0 \rangle \right\|,
\end{equation}
\begin{equation}
\varepsilon_2 = \sum_{l \in D^{l_{\max}}} \left\| \langle l| \left( e^{-i \mathscr{H}_{\text{eff}}^{l'_{\max}} t} - e^{-i \mathscr{H}_{\text{eff}}^{l_{\max}} t} \right) |0 \rangle \right\|.
\end{equation}
Hence $\left\| |\psi(t)\rangle - |\psi^{l_{\max}}(t)\rangle \right\| $ is bounded by $\varepsilon_1 + \varepsilon_2$ under $l_{\max}' \to \infty$. \cref{th:transition_rate} soon concludes the upper bound of $\varepsilon_1$:
\begin{equation}
\varepsilon_1 \leq \sum_{l \in (D^{l_{\max}'} \setminus D^{l_{\max}})} e^{2\beta  t - \frac{|l|^{{\frac{1}{\varrho}}}}{2\zeta} }\leq 2\sum_{l=l_{\max}}^{\infty} e^{2\beta  t - \frac{l^{{\frac{1}{\varrho}}}}{2\zeta} }\leq 2\left(\sum_{m=0}^{\infty} e^{- \frac{m^{{\frac{1}{\varrho}}}}{4\zeta} }\right)e^{2\beta  t - \frac{l^{{\frac{1}{\varrho}}}_{\max}}{4\zeta} }.
\end{equation}
The second error $\varepsilon_2$ comes from truncating the order of the effective Hamiltonian. In a similar way to the proof of \cref{th:transition_rate}, each term is bounded by
\begin{equation}
\left\| \langle l| \left( e^{-i \mathscr{H}_{\text{eff}}^{l'_{\max}} t} - e^{-i \mathscr{H}_{\text{eff}}^{l_{\max}} t} \right) |0 \rangle \right\|\leq\sum_{n=0}^{\infty} \int_{0}^{t} dt_n \cdots \int_{0}^{t_2} dt_1 \left\| \langle l| \left( \prod_{i=1}^{n} \mathscr{H}_{I}^{l_{\max}'}(t_i) - \prod_{i=1}^{n} \mathscr{H}_{I}^{l_{\max}}(t_i) \right) |0 \rangle \right\|.
\end{equation}
As stated in the proof of \cref{th:transition_rate}, each of $\langle l| \prod_{i=1}^{n} \mathscr{H}_{I}^{l_{\max}'}(t_i) |0 \rangle$ and $\langle l| \prod_{i=1}^{n} \mathscr{H}_{I}^{l_{\max}}(t_i) |0 \rangle$ is decomposed into a product of transition amplitudes via a path $|0\rangle \to |l_1\rangle \to \ldots \to |l_{n-1}\rangle \to |l\rangle$. Their difference appears only when the path goes across $D^{l_{\max}'} \setminus D^{l_{\max}}$. In other words, denoting the summation over the set of $\{l_i\}$ taking such nontrivial paths by $\sum_{\{l_i\}}'$, we reach
\begin{equation}
\left\| \langle l| \left( e^{-i \mathscr{H}_{\text{eff}}^{l'_{\max}} t} - e^{-i \mathscr{H}_{\text{eff}}^{l_{\max}} t} \right) |0 \rangle \right\|\leq \sum_{n=0}^{\infty} \int_{0}^{t} dt_n \cdots \int_{0}^{t_2} dt_1 \sum_{\{l_i\}}\!' \prod_{i=1}^{n} \left\| \langle l_i| \mathscr{H}_{I}^{l_{\max}'}(t_i) |l_{i-1} \rangle \right\|.
\end{equation}
When we define the hopping distance $ m_i $ by $ m_i = l_i - l_{i-1} $, it should satisfy $ \sum_{i=1}^n |m_i| \geq (l_{\max} - |l|) + (l_{\max} - 0) $ for $\{l_i\}$ such that a nontrivial path is organized. This results in the relation described by
\begin{align}
\varepsilon_2 &\leq \sum_{l \in D^{l_{\max}}} \sum_{n=0}^{\infty} \frac{t^n}{n!} \sum_{\{m_i\}_{i=1}^n \in \mathbb{Z}^n} \left( \prod_{i=1}^n \|\widetilde{H}_{m_i}\| \right) \times \mathbf{1}_{\{ \sum_{i=1}^n |m_i| - 2l_{\max} + |l| \geq 0\}} \notag \\
&\leq \sum_{l \in D^{l_{\max}}} \sum_{n=0}^{\infty} \frac{(2 t)^n}{n!} S_n (2l_{\max} - |l|) \notag \\
&\leq \sum_{l \in D^{l_{\max}}} \sum_{n=0}^{\infty} \frac{(2\beta  t)^n}{n!} e^{-(2l_{\max} - |l|)^{\frac{1}{\varrho}}/{2\zeta}} \notag \\
&\leq 2\left(\sum_{m=0}^{\infty} e^{- \frac{m^{{\frac{1}{\varrho}}}}{4\zeta} }\right) e^{2\beta  t - l_{\max}^{\frac{1}{\varrho}}/4\zeta }.
\end{align}
Combining these inequalities with $\||\psi(t)\rangle - |\psi^{l_{\max}}(t)\rangle\| \leq \varepsilon_1 + \varepsilon_2$ we complete the proof.
\end{proof}

\cref{th:Floquet-Hilbert_space_truncation} determines the proper truncation order of the Floquet-Hilbert space for the algorithm. When we aim at the desirable accuracy as $\||\psi(t)\rangle - |\psi^{l_{\max}}(t)\rangle\| \leq \varepsilon$, it is sufficient to choose the truncation order $l_{\max}$ by
\begin{equation}\label{eq:def_of_lmax}
l_{\max} \geq \left\lceil 4^{\varrho}\zeta^{\varrho}\left(2\beta   t +  \log(1/\varepsilon) +  \log(S(4\zeta))+\log(4)   \right)^{\varrho}\right\rceil \in \Theta \left( ( Ct + \log(1/\varepsilon))^{\varrho} \right).
\end{equation}
Note that the block-encoding factor $\alpha$ of the Hamiltonian $H(s)$ satisfies $\alpha \geq C$, where $C$ is an upper bound on its true norm. In what follows, when analyzing the asymptotic complexity, we rescale $C$ to $\alpha$. We pick
\begin{equation}
\tau = 1 + \frac{1}{\text{log}\bigl(\alpha t +e+ \log(1/\varepsilon)\bigr)}.
\end{equation}
Note that by \cref{th:Gevrey_ext,th:Four_decay_est,lam:est_of_Szeta},
\begin{equation}
A_2 = \mathcal{O}\!\left(\frac{1}{\tau-1}\right),\qquad 
A    = \mathcal{O}\!\left(\frac{1}{\tau-1}\right),\qquad 
\sum_{m\ge 1} e^{-m^{1/\varrho}/(2\zeta)} = \mathcal{O}(\zeta^{\varrho}) = \mathcal{O}\!\left(\frac{1}{\tau-1}\right),
\end{equation}
\begin{equation}
h_2\sum_{m\geq 1}e^{-m^{\frac{1}{\varrho}}/(2\zeta)} = \mathcal{O}\left(\alpha \text{log}\left(\frac{1}{\tau-1}\right)\right),\qquad 
h_3= \mathcal{O}(\alpha ),\qquad h_4 = \mathcal{O}(\alpha ).
\end{equation}
Then
\begin{equation}
\beta = \mathcal{O}\!\left(\alpha \log\!\left(\frac{1}{\tau-1}\right)\right)
\end{equation}
and
\begin{align}
l_{\max} &= \mathcal{O}\left( \frac{\left(\log\left(\frac{1}{\tau-1}\right)\alpha t + \log(1/\varepsilon)\right)^{\tau-1}}{\tau-1} 
           \cdot \left(\log\left(\frac{1}{\tau-1}\right)\alpha t + \log(1/\varepsilon)\right)^{\sigma} \right) \notag \\
        &= \mathcal{O}\left( \frac{(\alpha t + \log(1/\varepsilon))^{\tau-1}}{\tau-1} 
           \cdot (\alpha t + \log(1/\varepsilon))^{\sigma} 
           \cdot \left( \log\left(\frac{1}{\tau-1}\right) \right)^{\sigma+\tau-1} \right) \notag \\
        &= \mathcal{O}\left( \frac{(\alpha t + \log(1/\varepsilon))^{\tau-1}}{\tau-1} 
           \cdot (\alpha t + \log(1/\varepsilon))^{\sigma} 
           \cdot \left( \log\left(\frac{1}{\tau-1}\right) \right)^{\sigma} \right) \notag \\
        &= \mathcal{O}\left( (\alpha t + \log(1/\varepsilon))^{\sigma} 
           \cdot \text{log}\bigl(\alpha t + \log(1/\varepsilon)\bigr) 
           \cdot \left( \text{log}\text{log}\bigl(\alpha t + \log(1/\varepsilon)\bigr) \right)^{\sigma} \right)\notag \\
        &=\widetilde{\mathcal{O}}\left(\left(\alpha t+  \text{log}(1/\varepsilon\right))^{\sigma}\right).
\end{align}
In the following context we always assume $0<\varepsilon<1$, so that $\varrho=\sigma+\tau-1\in(1,3)$.

\section{Error analysis of simulating periodic Gevrey-class Hamiltonians}\label{sec:Error_analysis_5}

In this section, we discuss additional quantum subroutines needed to simulate the truncated effective Hamiltonian. 
These additional subroutines, including the symmetric amplitude amplification technique, the oblivious amplitude amplification (OAA) technique, and the refined effective Hamiltonian introduced for block encoding, are all generalizations of the corresponding techniques in \cite{mizuta2023optimal} to the Gevrey-class Hamiltonian setting. 
We briefly sketch the main results here, and provide all proofs of the theorems in this section in Appendix~\ref{app:main}.

\subsection{Amplitude amplification}

We verify the validity of the amplitude amplification protocols for time-periodic Hamiltonians with subexponentially-decaying Fourier components $\|\widetilde{H}_{m}\|\leq h e^{-|m|^{\frac{1}{\varrho}}/\zeta}$. 
The first amplification relies on the translation symmetry of the effective Hamiltonian. We begin with discussing the approximate translation symmetry in the truncated Floquet-Hilbert space.
\begin{lemma}[Approximate translation symmetry]\label[lemma]{lem:translation_symmetry}
We assume $\|\widetilde{H}_{m}\| $ can be bounded as in \cref{sec:Periodic_Extension} and \cref{sec:Four_decay}, and consider the truncated Floquet-Hilbert space $\mathbb{C}^{8l_{\max}}\otimes\mathscr{H}$. Then the transition rate has an approximate translation symmetry in that it satisfies
\begin{equation}
\left\| \langle l | e^{-i \mathscr{H}_{\text{\normalfont eff}}^{4l_{\max}}t} | l' \rangle - e^{i l' \omega t} \langle l \ominus l' | e^{-i \mathscr{H}_{\text{\normalfont eff}}^{4l_{\max}}t} | 0 \rangle \right\| \leq 2e^{2\beta  t-(8l_{\max}-|l|-|l'|)^{\frac{1}{\varrho}}/2\zeta },
\end{equation}
for any $ l\in D^{4l_{\max}}, l'\in D^{l_{\max}} $, where the integer $l\ominus l' \in D^{4l_{\max}}$ is defined modulo $8l_{\max}$.
\end{lemma}

With the usage of this approximate translation symmetry, we can organize the amplitude amplification by symmetry. We prepare the truncated Floquet-Hilbert space $\mathbb{C}^{8l_{\max}} \otimes \mathscr{H}$, and make the initial state uniform in $|l\rangle$ with $l_{\max} \in \Theta(( t+\log(1/\varepsilon))^{\varrho})$. Let us summarize the amplification protocol relying on the symmetry by $\mathscr{U}_{\text{amp}}^{l_{\max}}(t)$, whose explicit formula is given by 
\begin{align}
\mathscr{U}_{\text{amp1}}^{l_{\max}}(t) 
&= (\mathscr{U}_{\text{ini}}^{4l_{\max}})^{\dagger}  e^{-it \sum_{l\in D^{4l_{\max}}}l\omega| l \rangle\langle l |}  e^{-i \mathscr{H}_{\text{eff}}^{4l_{\max}} t}\mathscr{U}_{\text{ini}}^{l_{\max}},
\end{align}
where
\begin{equation}
\mathscr{U}_{\text{ini}}^{l_{\max}} = \left(|a^{l_{\max}}\rangle \langle 0| + \ldots\right) \otimes I,
\end{equation}
\begin{equation}
|a^{l_{\max}}\rangle =\frac{1}{\sqrt{2l_{\max}}}\sum_{l\in D^{l_{\max}}}|l\rangle.
\end{equation}
The unitary operator $\mathscr{U}_{\text{ini}}^{l_{\max}}$ non-trivially acts only on the ancillary space, and it can be implemented with $\mathcal{O}(\log l_{\max})$ gates (e.g., for $l_{\max}$ such that $\log_2 l_{\max} \in \mathbb{N}$, we have $\mathscr{U}_{\text{ini}}^{l_{\max}} = \text{Had}^{\otimes \log_2(2l_{\max})} \otimes I$ with the Hadamard gate $\text{Had}$). This process plays a role of making the initial state approximately translation-invariant with the width of Fourier indices, $2l_{\max}$. The state resulting from the time evolution $\exp(-i\mathscr{H}_{\text{eff}}^{4l_{\max}}t)$ outputs the target time-evolved state $|\psi(t)\rangle$ with amplitude $1/2$ as follows.

\begin{theorem}[Amplification by symmetry]\label{th:Amplification_by_symmetry}
We assume $\|\widetilde{H}_{m}\| $ can be bounded as in \cref{sec:Periodic_Extension} and \cref{sec:Four_decay}, and choose the truncation order $l_{\max}$ as in \cref{eq:def_of_lmax}. Then, $\mathscr{U}_{\text{\normalfont amp1}}^{l_{\max}}(t)$  generates the time-evolved state $|\psi(t)\rangle$ with $\mathcal{O}(1)$ amplitude as
\begin{equation}
\left\| \langle 0| \mathscr{U}_{\text{\normalfont amp1}}^{l_{\max}}(t)|0\rangle |\psi(0)\rangle - \frac{1}{2} |\psi(t)\rangle \right\| \leq \frac{10\varepsilon^{3^{1/\varrho}}}{4^{ 3^{1/\varrho}}(S(4\zeta))^{3^{1/\varrho}-1}}\in\mathcal{O}(\varepsilon)
\end{equation}
for $0<\varepsilon<1$.
\end{theorem}

The above amplification based on the translation symmetry of $\mathscr{H}_{\text{eff}}$ enhances the amplitude of $|\psi(t)\rangle$ from $1/\sqrt{2l_{\max}}$ to $1/2-\mathcal{O}(\varepsilon)$. Here, we introduce another amplification, called the oblivious amplitude amplification~\cite{BerryChildsKothari2015}, to achieve the amplitude (or the success probability) $1-\mathcal{O}(\varepsilon)$.
The starting point of this amplification is the result of \cref{th:Amplification_by_symmetry}, indicating that the consequence of the amplification by symmetry is written as
\begin{align}
\mathscr{U}_{\text{amp1}}^{l_{\max}}(t)\,|0\rangle\,|\psi(0)\rangle &=\left(\sum_{l\in D^{4l_{\max}}} |l\rangle\langle l|\right)\mathscr{U}_{\text{amp1}}^{l_{\max}}(t)\,|0\rangle\,|\psi(0)\rangle \notag \\
&=|0\rangle\left(\langle 0|\mathscr{U}_{\text{amp1}}^{l_{\max}}(t)\,|0\rangle\,|\psi(0)\rangle\right)+\sum_{l\in D^{4l_{\max}}\backslash\{0\}} |l\rangle\left(\langle l|\mathscr{U}_{\text{amp1}}^{l_{\max}}(t)\,|0\rangle\,|\psi(0)\rangle\right) \notag \\
&=:\frac{1}{2}\,|0\rangle\,|\psi(t)\rangle+|0\rangle\Delta+|\Psi^{\perp}\rangle,
\end{align}
with 
\begin{equation}
\Delta=\langle 0| \mathscr{U}_{\text{amp1}}^{l_{\max}}(t)|0\rangle |\psi(0)\rangle - \frac{1}{2} |\psi(t)\rangle,\quad \quad \left\||0\rangle\Delta\right\|\leq\frac{10\varepsilon^{3^{1/\varrho}}}{4^{ 3^{1/\varrho}}(S(4\zeta))^{3^{1/\varrho}-1}},
\end{equation}
and additional term $|\Psi^{\perp}\rangle\in\mathbb{C}^{8l_{\max}}\otimes\mathscr{H}$ satisfying
\begin{equation}
(|0\rangle\langle 0|\otimes I)\,|\Psi^{\perp}\rangle=0.
\end{equation}
The state $|\Psi^{\perp}\rangle$ generally depends on $|\psi(0)\rangle$ and $t$. The oblivious amplitude amplification takes a similar strategy to that of the Grover's search algorithm. 
We summarize the results in the following theorem.
\begin{theorem}[Amplification by oblivious amplitude amplification]\label{th:OAA}
We define two unitary operators
\begin{align}
\mathscr{R} &= (2\,|0\rangle\langle 0|-I)\otimes I, \\
\mathscr{U}_{\text{\normalfont amp2}}^{l_{\max}}(t) &= -\mathscr{U}_{\text{\normalfont amp1}}^{l_{\max}}(t)\mathscr{R}[\mathscr{U}_{\text{\normalfont amp1}}^{l_{\max}}(t)]^{\dagger}\mathscr{R}\mathscr{U}_{\text{\normalfont amp1}}^{l_{\max}}(t).
\end{align}
Then $\mathscr{U}_{\text{\normalfont amp2}}^{l_{\max}}(t)$  generates the time-evolved state $|\psi(t)\rangle$ with $1-\mathcal{O}(\varepsilon)$ amplitude as
\begin{equation}
\left\| \mathscr{U}_{\text{\normalfont amp2}}^{l_{\max}}(t) |0\rangle |\psi(0)\rangle - |0\rangle |\psi(t)\rangle \right\| \leq \frac{40\varepsilon^{3^{1/\varrho}}}{4^{ 3^{1/\varrho}}(S(4\zeta))^{3^{1/\varrho}-1}}\in\mathcal{O}(\varepsilon)
\end{equation}
for $0<\varepsilon<1$.
\end{theorem}

\subsection{Refined effective Hamiltonian}
We introduce a refined effective Hamiltonian $\mathscr{H}_{\text{eff, pbc}}^{4l_{\max}}$ which acts on the truncated Floquet-Hilbert space $ \mathbb{C}^{8l_{\max}} \otimes \mathscr{H} $, by
\begin{equation}
\mathscr{H}_{\text{eff, pbc}}^{4l_{\max}} = \mathscr{H}_{\text{eff}}^{4l_{\max}} + \tilde{\mathscr{H}}_b, \qquad 
\tilde{\mathscr{H}}_b = \sum_{(l,m) \in \partial \tilde{F}^{4l_{\max}}} |l\rangle \langle l \oplus m| \otimes H_{m} + (|l\rangle \langle l \oplus m| \otimes H_{m})^{\dagger} ,
\end{equation}
where $\partial \tilde{F}^{4l_{\max}} = \{(l,m)\mid l \in D^{4l_{\max}},\; 8l_{\max} - l + 1 \leq m \leq 8l_{\max} - 1\}$ shares the same definition with the one in \cref{lem:translation_symmetry}. Owing to the additional term $\tilde{\mathscr{H}}_b$, the hopping terms induced by $H_{m}$ become translation symmetric in $|l\rangle$ as
\begin{equation}
\mathscr{H}_{\text{eff, pbc}}^{4l_{\max}} = \sum_{m \in D^{4l_{\max}}} \text{Add}_{m}^{4l_{\max}} \otimes \widetilde{H}_{-m}- \mathscr{H}_{\text{LP}}^{4l_{\max}}.
\end{equation}
Here, the linear potential Hamiltonian defined as 
\begin{equation}
\mathscr{H}_{\text{LP}}^{4l_{\max}}=\sum_{l\in D^{4l_{\max}}}l\omega|l \rangle\langle l|\,\otimes I,
\end{equation}
and a full quantum adder defined by 
\begin{equation}\label{eq:def_of_Add}
\text{Add}_{m}^{4l_{\max}}=\sum_{l\in D^{4l_{\max}}}|l\oplus m \rangle\langle l|,
\end{equation}
where $\text{Add}_m^{4l_{\max}}$ in \cref{eq:def_of_Add} and $\sum_{m\in D^{4l_{\max}}}\ket{m}\bra{m}\otimes\text{ Add}_m^{4l_{\max}}$ both can be implemented by $\mathcal{O}(\text{\normalfont log}\,l_{\max})$ elementary gates \cite{cuccaro2004new}. Now we prove the validity of the refined effective Hamiltonian.

\begin{theorem}[Refined effective Hamiltonian]\label{th:Refined_effective_Hamiltonian}
We assume $\|\widetilde{H}_{m}\| $ can be bounded as in \cref{sec:Periodic_Extension} and \cref{sec:Four_decay}, and choose the truncation order $l_{\max}$ as in \cref{eq:def_of_lmax}. We organize the two amplification protocols $\mathscr{U}_{\text{\normalfont amp1, pbc}}^{l_{\max}}(t)$ and $\mathscr{U}_{\text{\normalfont amp2, pbc}}^{l_{\max}}(t)$ respectively based on \cref{th:Amplification_by_symmetry} and \cref{th:OAA} with using the refined effective Hamiltonian $\mathscr{H}_{\text{\normalfont eff, pbc}}^{4l_{\max}}$ given by 
\begin{equation}
\mathscr{H}_{\text{\normalfont eff, pbc}}^{4l_{\max}} = \mathscr{H}_{\text{\normalfont eff}}^{4l_{\max}} + \tilde{\mathscr{H}}_b, \qquad 
\tilde{\mathscr{H}}_b = \sum_{(l,m) \in \partial \tilde{F}^{4l_{\max}}} |l\rangle \langle l \oplus m| \otimes \widetilde{H}_{m} + (|l\rangle \langle l \oplus m| \otimes \widetilde{H}_{m})^{\dagger} .
\end{equation}
Then they also provide the exact time‑evolved state $|\psi(t)\rangle$ as
\begin{equation}
\left\| \langle 0 | \mathscr{U}_{\text{\normalfont amp1, pbc}}^{l_{\max}}(t) | 0 \rangle | \psi(0) \rangle - \frac{1}{2} |\psi(t)\rangle \right\| \leq  \frac{11\varepsilon^{3^{1/\varrho}}}{4^{ 3^{1/\varrho}}(S(4\zeta))^{3^{1/\varrho}-1}},
\end{equation}
\begin{equation}
\left\| \mathscr{U}_{\text{\normalfont amp2, pbc}}^{l_{\max}}(t) | 0 \rangle | \psi(0) \rangle - |0\rangle |\psi(t)\rangle \right\| \leq  \frac{44\varepsilon^{3^{1/\varrho}}}{4^{3^{1/\varrho}}(S(4\zeta))^{3^{1/\varrho}-1}},
\end{equation}
for arbitrary initial states $|\psi(0)\rangle \in \mathscr{H}$. For an allowable error $0<\varepsilon<1$, both of the left hand sides are smaller than $\frac{\varepsilon}{2}$.
\end{theorem}

\section{Algorithm implementation and complexity analysis}\label{sec:Algorithm_implementation}
In this section, we present the detailed implementation of the quantum algorithm for simulating Gevrey-class time-dependent Hamiltonians and analyze its query and gate complexity. The exposition is divided into two parts, corresponding to the two oracle models discussed in \cref{sec:intro}.
\subsection{The general Hamiltonian case}
\subsubsection{Oracle definitions for the general Hamiltonian case}\label{subsubsec:oracle_general_case}
Before presenting the detailed implementation, we first specify the oracles used in this section. We assume access to block encodings of the derivative $H'(s)$ and of the high-order derivatives $H^{(j)}(0)$, $H^{(j)}(1)$ at the endpoints. In the construction below, these oracles are organized into three controlled oracle families:

\begin{itemize}
\item \textbf{HAM'-\!T}: 
\begin{equation}
\mathrm{HAM'\!-\!T} := \sum_{j=0}^{L(\epsilon_2,K)} |j\rangle\langle j| \otimes O_j^{(0)},
\end{equation}
which provides block encodings of $H'(t_j)/(CD)$ at the Clenshaw--Curtis quadrature nodes $t_j$.
\item \textbf{HOHAM-1}: 
\begin{equation}
\mathrm{HOHAM\!-\!1} := \sum_{j=0}^{C(\epsilon_2,K)} |j\rangle\langle j| \otimes O_j^{(1)},
\end{equation}
which provides block encodings of $H^{(j)}(1)/(CD^j(j!)^\sigma)$.
\item \textbf{HOHAM-0}: 
\begin{equation}
\mathrm{HOHAM\!-\!0} := \sum_{j=0}^{C(\epsilon_2,K)} |j\rangle\langle j| \otimes O_j^{(2)},
\end{equation}
which provides block encodings of $H^{(j)}(0)/(CD^j(j!)^\sigma)$.
\end{itemize}

Here $L(\epsilon_2,K)$ is the truncation order of the Clenshaw--Curtis quadrature (see Step~2 below), and $C(\epsilon_2,K)$ is the truncation order of the Gevrey series expansion (also defined in Step~2). The constants $C>0$ and $D\ge 1$ are those appearing in the Gevrey bound on $H(s)$ from \cref{sec:Periodic_Extension}.

The block encoding of $\mathscr{H}_{\mathrm{eff,pbc}}^{4l_{\max}}$ is constructed via the standard linear combination of unitaries (LCU) technique. For this purpose, we define a prepare oracle $G_{\mathrm{prep},L}$ and $G_{\mathrm{prep},R}$ as
\begin{align}
G_{\mathrm{prep},L} &= \left(\sum_{k=0}^{2}(|k\rangle\langle k|)_5\otimes I_4\otimes G_{\mathrm{coef},3\sim 2,L}^{(k)}\otimes I_{1,b,a} + (|3\rangle\langle 3|)_5\otimes I_4\otimes \mathcal{U}_{\mathrm{ini}}^{4l_{\max}}\otimes I_{2,1,b,a}\right) \notag \\
&\qquad \cdot \left(G_{\mathrm{coef},5}\otimes I_{4\sim 1,b,a}\right), 
\end{align}
\begin{align}
G_{\mathrm{prep},R} &= \left(\sum_{k=0}^{2}(|k\rangle\langle k|)_5\otimes I_4\otimes G_{\mathrm{coef},3\sim 2,R}^{(k)}\otimes I_{1,b,a} + (|3\rangle\langle 3|)_5\otimes I_4\otimes \mathcal{U}_{\mathrm{ini}}^{4l_{\max}}\otimes I_{2,1,b,a}\right) \notag \\
&\qquad \cdot \left(G_{\mathrm{coef},5}\otimes I_{4\sim 1,b,a}\right), 
\end{align}
where the auxiliary system 5 is a $2$-qubit system that selects among the four terms in $\mathscr{H}_{\mathrm{eff,pbc}}^{4l_{\max}}$: the three derivative-related contributions ($k=0,1,2$) and the linear potential term ($k=3$). The coefficient preparation unitaries $G_{\mathrm{coef},3\sim 2,L}^{(k)}$ and $G_{\mathrm{coef},3\sim 2,R}^{(k)}$ encode the coefficients $B_{j,m}^{(k)}$ arising from the Fourier expansion, as defined in Step~2 below. The unitary $G_{\mathrm{coef},5}$ prepares the normalization factors for the four terms.

The select oracle is defined by
\begin{equation}
O_{\mathrm{select}} = \sum_{k=0}^{2}(|k\rangle\langle k|)_5\otimes O^{(k)} + (|3\rangle\langle 3|)_5\otimes O_{\mathrm{LP}}, 
\end{equation}
where $O^{(0)}$, $O^{(1)}$, $O^{(2)}$ apply the controlled unitaries $\mathrm{Add}_m^{4l_{\max}}\otimes O_{j,a}^{(k)}$ conditioned on the indices $m$ and $j$, and $O_{\mathrm{LP}}$ implements the linear potential $\mathscr{H}_{\mathrm{LP}}^{4l_{\max}}$ via the comparison gate Comp. The explicit forms of $O^{(0)}$, $O^{(1)}$, $O^{(2)}$, and $O_{\mathrm{LP}}$ are given in Step~4 below.

With these oracles, the block encoding of $\mathscr{H}_{\mathrm{eff,pbc}}^{4l_{\max}}$ takes the form
\begin{equation}
\left\| \langle 0|_{5\sim 1}(G_{\mathrm{prep},L})^{\dagger}O_{\mathrm{select}}G_{\mathrm{prep},R}|0\rangle_{5\sim 1} - \frac{\mathscr{H}_{\mathrm{eff,pbc}}^{4l_{\max}}}{\sum_{k=0}^{3}D^{(k)}}\right\| \le \epsilon_1, 
\end{equation}
where $D^{(k)}$ are the normalization constants defined in \cref{eq:def_of_D0,eq:def_of_D2,eq:def_of_D3}. The construction of these oracles and the proof of the above block-encoding error are detailed in Steps~2--4 below.

\subsubsection{Algorithm implementation for the general Hamiltonian case}\label{subsec:general_case}
\textbf{Step 1:} Firstly, by further exploiting the sub-exponential decay of $||\widetilde{H}_{-m}||$, we can impose a more ambitious truncation on 
\begin{equation}
\sum_{m \in D^{4l_{\max}}} \text{Add}_{m}^{4l_{\max}} \otimes \widetilde{H}_{-m},
\end{equation}
which appears in the expression for $\mathscr{H}_{\text{eff, pbc}}^{4l_{\max}}$. Specifically, for any $\epsilon_1>0$, let $0<K\leq l_{\max}$ be a integer such that
\begin{equation}
\left\| \sum_{m \in D^{K}} \text{Add}_{m}^{4l_{\max}} \otimes \widetilde{H}_{-m} -\sum_{m \in D^{4l_{\max}}} \text{Add}_{m}^{4l_{\max}} \otimes \widetilde{H}_{-m}\right\|=\left\| \sum_{m \in D^{4l_{\max}}\setminus D^K} \text{Add}_{m}^{4l_{\max}} \otimes \widetilde{H}_{-m}\right\|  \leq  \epsilon_1.
\end{equation}
Note that
\begin{align}
\left\| \sum_{m \in D^{4l_{\max}}\setminus D^K} \text{Add}_{m}^{4l_{\max}} \otimes \widetilde{H}_{-m}\right\| 
&\leq\sum_{m \in D^{4l_{\max}}\setminus D^K}\left\|\text{Add}_{m}^{4l_{\max}} \otimes \widetilde{H}_{-m}\right\| \notag \\
&=\sum_{m \in D^{4l_{\max}}\setminus D^K}\left\| \widetilde{H}_{-m}\right\| \notag \\
&\leq2h\sum_{m \geq K}e^{- \frac{m^{\frac{1}{\varrho}}}{\zeta}} \notag \\
&\leq2S(2\zeta)h e^{- \frac{K^{\frac{1}{\varrho}}}{2\zeta}},
\end{align}
recall that $S(2\zeta)=\sum_{m=0}^{\infty}e^{- \frac{m^{\frac{1}{\varrho}}}{2\zeta}}$. We only need to ensure
\begin{equation}
2\left(\sum_{m=0}^{\infty}e^{- \frac{m^{\frac{1}{\varrho}}}{2\zeta}}\right)h e^{- \frac{K^{\frac{1}{\varrho}}}{2\zeta}}\leq\epsilon_1.
\end{equation}
Hence we can take
\begin{equation}
K=\left\lceil\left(2\zeta\text{log}\left(\frac{2S(2\zeta) h}{\epsilon_1}\right)\right)  ^{\varrho}\right\rceil\bigwedge \;l_{\text{max}}\;,
\end{equation}
which could exhibit a better dependence on $t$ than $l_{\max}$, under the choice of $\epsilon_1$ we will determine later.\\
\textbf{Step 2:} Now we aim to further rewrite the expression 
\begin{equation}
\sum_{m \in D^{K}} \text{Add}_{m}^{4l_{\max}} \otimes \widetilde{H}_{-m}.
\end{equation}
Note that for $m\in \mathbb{Z}\setminus\{0\}$,
\begin{align}
\widetilde{H}_{-m}
&=\frac{1}{2T}\int_{0}^{2T} \widetilde{H}(t)e^{im\frac{2\pi}{2T}t}dt \notag \\
&=\frac{1}{2T}\int_{0}^{2T} \hat{H}(t/T)e^{im\pi\frac{t}{T}}dt \notag \\
&=\frac{1}{2}\int_{0}^{2} \hat{H}(t)e^{im\pi t}dt \notag \\
&=\frac{1}{-2m\pi i}\int_{0}^{2} \hat{H}'(t)e^{im\pi t}dt \notag \\
&=\frac{1}{-2m\pi i}\int_{0}^{1} H'(t)e^{im\pi t}dt+\frac{(-1)^m}{-2m\pi i}\int_{0}^{1} \hat{H}'(t+1)e^{im\pi t}dt.
\end{align}
Then we have 
\begin{align}
\sum_{m\in D^{K}}\text{Add}_m^{4l_{\text{max}}}\otimes \widetilde{H}_{-m}
&=\sum_{m\in D^{K}\setminus\{0\}}\frac{\text{Add}_m^{4l_{\text{max}}}}{-2m\pi i}\otimes \int_{0}^{1} H'(t)e^{im\pi t}\,dt \notag \\
&\quad+\sum_{m\in D^{K}\setminus\{0\}}\frac{(-1)^m\text{Add}_m^{4l_{\text{max}}}}{-2m\pi i}\otimes \int_{0}^{1} \hat{H}'(t+1)e^{im\pi t}\,dt \notag \\
&\quad+\frac{1}{2}\int_{0}^{2} \hat{H}(t)\,dt =:I_1+I_2+I_3.
\end{align}
For the first term $I_1$, we can use Clenshaw–Curtis quadrature in \cite{riess1971error}, which provides an integral approximation that converges exponentially for analytic functions $(\sigma=1)$ and sub-exponentially for Gevrey functions $(\sigma>1)$. For any $\epsilon_2>0$, we choose $L(\epsilon_2,K)=\mathcal{O}((\text{log}(1/\epsilon_2)^{\sigma})$ such that
\begin{equation}\label{eq:cond_of_L(eps)_1}
\left\| \int_{0}^{1}  tH'(t)\,dt-\sum_{j=0}^{L(\epsilon_2,K)}t_jH'(t_j)w_j\right\| \leq {\epsilon_2},
\end{equation}
and
\begin{equation}\label{eq:cond_of_L(eps)_2}
\left\| \int_{0}^{1}  H'(t)e^{im\pi t}\,dt-\sum_{j=0}^{L(\epsilon_2,K)}H'(t_j)e^{im\pi t_j}w_j\right\| \leq \frac{\pi\epsilon_2}{\text{log}K+1},   \quad\text{for any}\;m\in D^K\setminus\{0\}. 
\end{equation}
Then
\begin{align}
I_1&=\sum_{m\in D^{K}\setminus\{0\}}\frac{\text{Add}_m^{4l_{\text{max}}}}{-2m\pi i}\otimes \int_{0}^{1} H'(t)e^{im\pi t}\,dt \notag \\
&\quad= \sum_{m\in D^{K}\setminus\{0\}}\sum_{j=0}^{L(\epsilon_2,K)}\frac{e^{im\pi t_j}w_j}{-2m\pi i}  \left(\text{Add}_m^{4l_{\text{max}}}\otimes H'(t_j)\right) \notag \\
&\qquad+\sum_{m\in D^{K}\setminus\{0\}}\frac{\text{Add}_m^{4l_{\text{max}}}}{-2m\pi i}\otimes\left(\int_{0}^{1}  H'(t)e^{im\pi t}\,dt-\sum_{j=0}^{L(\epsilon_2,K)}H'(t_j)e^{im\pi t_j}w_j\right),
\end{align}
and 
\begin{equation}
\left\|\sum_{m\in D^{K}\setminus\{0\}}\frac{\text{Add}_m^{4l_{\text{max}}}}{-2m\pi i}\otimes\left(\int_{0}^{1}  H'(t)e^{im\pi t}\,dt-\sum_{j=0}^{L(\epsilon_2,K)}H'(t_j)e^{im\pi t_j}w_j\right)\right\| \leq\sum_{m=1}^K\frac{1}{m}\frac{\epsilon_2}{\text{log}K+1}\leq \epsilon_2.
\end{equation}
Denote the coefficients 
\begin{equation}\label{eq:def_of_A(0)}
\frac{e^{im\pi t_j}w_j}{-2m\pi i}=A^{(0)}_{j,m},
\end{equation}
we can see that 
\begin{equation}
\sum_{m\in D^{K}\setminus\{0\}}\sum_{j=0}^{L(\epsilon_2,K)}|A^{(0)}_{j,m}|\leq\sum_{m\in D^{K}\setminus\{0\}}\sum_{j=0}^{L(\epsilon_2,K)}\frac{w_j}{2|m|\pi}\leq \frac{\text{log}K+1}{\pi}.
\end{equation}
For the second term $I_2$, recall that from \cref{th:Gevrey_ext}
\begin{equation}
\hat{H}(t+1) = \left(\sum_{j=0}^\infty  \frac{H^{(j)}(1)}{j!}t^j\chi_\tau(R_j t)\right) + \left(\sum_{j=0}^\infty  \frac{H^{(j)}(0)}{j!}(t-1)^j\chi_\tau(R_j (t-1))\right).
\end{equation}
Hence we have
\begin{align}
\int_{0}^{1}\hat{H}'(t+1)e^{im\pi t}\,dt &=\int_{0}^{1} \left(\sum_{j=0}^\infty  \frac{H^{(j)}(1)}{j!}(t^j)'\chi_\tau(R_j t)\right) e^{im\pi t}\,dt \notag \\
&\quad+ \int_{0}^{1}\left(\sum_{j=0}^\infty  \frac{H^{(j)}(1)}{j!}t^j(\chi_\tau(R_j t))'\right)e^{im\pi t}\,dt \notag \\
&\quad+ \int_{0}^{1} \left(\sum_{j=0}^\infty  \frac{H^{(j)}(0)}{j!}((t-1)^j)'\chi_\tau(R_j (t-1))\right)e^{im\pi t}\,dt \notag \\
&\quad+ \int_{0}^{1} \left(\sum_{j=0}^\infty  \frac{H^{(j)}(0)}{j!}(t-1)^j(\chi_\tau(R_j (t-1)))'\right)e^{im\pi t}\,dt=:\sum_{l=1}^4I_{2,l}.
\end{align}
We first focus on $I_{2,1}$:
\begin{align}
I_{2,1}
&=\int_{0}^{1}\left(\sum_{j=0}^\infty  \frac{H^{(j)}(1)}{j!}(t^j)'\chi_\tau(R_j t)\right) e^{im\pi t}\,dt \notag \\
&=\sum_{j=0}^\infty \left(C\int_{0}^{1}  \frac{D^j(j!)^{\sigma}}{j!}(t^j)'\chi_\tau(R_j t) e^{im\pi t}\,dt\right)\frac{H^{(j)}(1)}{CD^j(j!)^{\sigma}},
\end{align}
Note that $\left\|\frac{H^{(j)}(1)}{CD^j(j!)^{\sigma}}\right\|\leq1$ and 
\begin{align}
\sum_{j=0}^\infty C\left|\int_{0}^{1}  \frac{D^j(j!)^{\sigma}}{j!}(t^j)'\chi_\tau(R_j t) e^{im\pi t}\,dt\right|
&\leq C\sum_{j=0}^\infty \int_{0}^{1}  D^j(j!)^{\sigma-1}R_j^{-j}jR_j\left|\chi_\tau(R_j t)\right|\,dt \notag \\
&\leq C\sum_{j=0}^\infty D^j(j!)^{\sigma-1}(4e^{\sigma-1}D(j!)^{(\sigma-1)/j})^{-j}jR_j \notag \\
&\leq CD\sum_{j=0}^\infty \frac{j^{\sigma}}{(4e^{\sigma-1})^{j-1}}.
\end{align}
For $I_{2,2}$, we have 
\begin{align}
I_{2,2}
&=\int_{0}^{1} \left(\sum_{j=0}^\infty  \frac{H^{(j)}(1)}{j!}t^j(\chi_\tau(R_j t))'\right)e^{im\pi t}\,dt \notag \\
&=\sum_{j=0}^\infty \left(C\int_{0}^{1} \frac{D^j(j!)^{\sigma}}{j!}t^j(\chi_\tau(R_j t))' e^{im\pi t}\,dt\right)\frac{H^{(j)}(1)}{CD^j(j!)^{\sigma}},
\end{align}
and
\begin{align}
\sum_{j=0}^\infty C\left|\int_{0}^{1} \frac{D^j(j!)^{\sigma}}{j!}t^j(\chi_\tau(R_j t))' e^{im\pi t}\,dt\right|
&\leq 4e^2CD\sum_{j=0}^\infty \frac{j^{\sigma-1}}{(4e^{\sigma-1})^{j-1}}.
\end{align}
The other terms $I_{2,l}\,(l=3,4)$ are handled in exactly the same way. We can choose $C(\epsilon_2,K)=\widetilde{\mathcal{O}}(\text{log}(C/\epsilon_2)+\text{log}\text{log}K)$ such that
\begin{equation}\label{eq:cond_of_C(eps)}
4e^2CD\sum_{j=C(\epsilon_2,K)+1}^\infty \frac{j^{\sigma}}{(4e^{\sigma-1})^{j-1}}\leq \frac{\pi\epsilon_2}{4(\text{log}K+1)}.
\end{equation}
Denote 
\begin{align}
A^{(1)}_{j,m} &=  C\int_{0}^{1} \frac{D^j(j!)^{\sigma}}{j!} (t^j)' \chi_\tau(R_j t) e^{im\pi t}\, dt \notag\\
&\quad + C\int_{0}^{1} \frac{D^j(j!)^{\sigma}}{j!} t^j (\chi_\tau(R_j t))' e^{im\pi t}\, dt, \label{eq:def_of_A(1)}
\end{align}
\begin{align}
A^{(2)}_{j,m} &= C\int_{0}^{1}\frac{D^j(j!)^{\sigma}}{j!} ((t-1)^j)' \chi_\tau(R_j (t-1)) e^{im\pi t}\, dt \notag\\
&\quad + C\int_{0}^{1} \frac{D^j(j!)^{\sigma}}{j!} (t-1)^j (\chi_\tau(R_j (t-1)))' e^{im\pi t}\, dt.\label{eq:def_of_A(2)}
\end{align}
Then we have 
\begin{align}
&\Bigg\|\sum_{m\in D^{K}\setminus\{0\}}\frac{(-1)^m\text{Add}_m^{4l_{\text{max}}}}{-2m\pi i}\otimes \int_{0}^{1} \hat{H}'(t+1)e^{im\pi t}\,dt \notag \\
&\qquad - \sum_{m\in D^{K}\setminus\{0\}} \frac{(-1)^m\text{Add}_m^{4l_{\text{max}}}}{-2m\pi i}\otimes\Bigg(\sum_{j=0}^{C(\epsilon_2,K)}A^{(1)}_{j,m}\frac{H^{(j)}(1)}{CD^j(j!)^{\sigma}}+\sum_{j=0}^{C(\epsilon_2,K)}A^{(2)}_{j,m}\frac{H^{(j)}(0)}{CD^j(j!)^{\sigma}}\Bigg)\Bigg\| \notag \\
&\quad\leq \sum_{m=1}^K\frac{1}{m }\frac{\epsilon_2}{(\log K+1)}\leq \epsilon_2 .
\end{align}
In addition, we can see that
\begin{align}
&\sum_{m\in D^{K}\setminus\{0\}}\sum_{j=0}^{C(\epsilon_2,K)}\frac{|A^{(1)}_{j,m}|+|A^{(2)}_{j,m}|}{2|m|\pi} \notag \\
&\quad\leq\sum_{m\in D^{K}\setminus\{0\}}\frac{1}{2|m|\pi}
       \left((8e^2+2)CD\sum_{j=0}^\infty \frac{j^{\sigma}}{(4e^{\sigma-1})^{j-1}}\right) \notag \\
&\quad\leq \frac{(8e^2+2)CD(\log K+1)}{\pi}
       \left(\sum_{j=0}^\infty \frac{j^{\sigma}}{(4e^{\sigma-1})^{j-1}}\right).
\end{align}
For the second term $I_3$,
\begin{align}
I_3 &= \frac{1}{2}\int_{0}^{2} \hat{H}(t)\, dt \notag \\
    &= -\frac{1}{2}\int_{0}^{1} t H'(t)\, dt + \frac{1}{2}H(1)+\frac{1}{2}\int_{1}^{2} \hat{H}(t)\, dt  \notag \\
    &= -\frac{1}{2}\int_{0}^{1} t H'(t)\, dt + \frac{1}{2}\int_{0}^{1} \hat{H}(t+1)\, dt +  \frac{1}{2}H(1) \notag \\
    &= -\frac{1}{2}\int_{0}^{1} t H'(t)\, dt \notag \\
      &\quad + \Bigg( \frac{1}{2}\int_{0}^{1}  \left(\sum_{j=0}^\infty \frac{H^{(j)}(1)}{j!} t^j \chi_\tau(R_j t)\right) dt \notag \\
      &\qquad\quad +\frac{1}{2}\int_{0}^{1} \left(\sum_{j=0}^\infty \frac{H^{(j)}(0)}{j!} (t-1)^j \chi_\tau(R_j (t-1))\right) dt +  \frac{1}{2}H(1)\Bigg) \notag \\
    &=: I_{3,1} + I_{3,2} .
\end{align}
For $I_{3,1}$ we have
\begin{equation}
I_{3,1}=-\frac{1}{2}\int_{0}^{1} tH'(t)dt= -\frac{1}{2}\sum_{j=0}^{L(\epsilon_2,K)}t_jH'(t_j)w_j-\frac{1}{2}\left(\int_{0}^{1}  tH'(t)\,dt-\sum_{j=0}^{L(\epsilon_2,K)}t_jH'(t_j)w_j\right),
\end{equation}
For $I_{3,2}$ we have
\begin{align}
I_{3,2} =& \frac{1}{2}\int_{0}^{1} \left(\sum_{j=0}^\infty  \frac{H^{(j)}(1)}{j!}t^j\chi_\tau(R_j t)\right)dt \notag \\
         &+\frac{1}{2}\int_{0}^{1} \left(\sum_{j=0}^\infty  \frac{H^{(j)}(0)}{j!}(t-1)^j\chi_\tau(R_j (t-1))\right)dt+\frac{1}{2}H(1) \notag \\
        =&\ \frac{1}{2}\left(\sum_{j=0}^{C(\epsilon_2,K)}+\sum_{j=C(\epsilon_2,K)+1}^\infty\right) 
           \left(\int_{0}^{1}\frac{CD^j(j!)^{\sigma}}{j!}t^j\chi_\tau(R_j t)\,dt+C\delta_{j,\,0}\right) \frac{H^{(j)}(1)}{CD^j(j!)^{\sigma}} \notag \\
         & +\frac{1}{2} \left(\sum_{j=0}^{C(\epsilon_2,K)}+\sum_{j=C(\epsilon_2,K)+1}^\infty\right)   
           \left(\int_{0}^{1} \frac{CD^j(j!)^{\sigma}}{j!}(t-1)^j\chi_\tau(R_j (t-1))\,dt\right)\frac{H^{(j)}(0)}{CD^j(j!)^{\sigma}}.
\end{align}
Denote
\begin{equation}\label{eq:def_of_A(3)}
A^{(3)}_{j}=\frac{1}{2}\left(\int_{0}^{1} \frac{CD^j(j!)^{\sigma}}{j!}t^j\chi_\tau(R_j t)\,dt+C\delta_{j,\,0}\right),
\end{equation}
\begin{equation}\label{eq:def_of_A(4)}
A^{(4)}_{j}=\frac{1}{2}\int_{0}^{1} \frac{CD^j(j!)^{\sigma}}{j!}(t-1)^j\chi_\tau(R_j (t-1))\,dt,
\end{equation}
we know that
\begin{equation}
\left\|\sum_{j=C(\epsilon_2,K)+1}^\infty A^{(3)}_{j} \frac{H^{(j)}(1)}{CD^j(j!)^{\sigma}}+\sum_{j=C(\epsilon_2,K)+1}^\infty A^{(4)}_{j}\frac{H^{(j)}(0)}{CD^j(j!)^{\sigma}}\right\|\leq \frac{\epsilon_2}{2}
\end{equation}
and
\begin{equation}
\sum_{j=0}^{C(\epsilon_2,K)}|A^{(3)}_{j}|+|A^{(4)}_{j}|\leq \sum_{j=0}^{C(\epsilon_2,K)}C\left(\frac{1}{4e^{\sigma-1}}\right)^j+C\leq \sum_{j=0}^{\infty}C\left(\frac{1}{4e^{\sigma-1}}\right)^j+C.
\end{equation}
Combining all the estimates in Step 2, we obtain that
\begin{align}
&\Bigg\|\sum_{m\in D^{K}}\sum_{j=0}^{L(\epsilon_2,K)}B_{j,m}^{(0)}
   \left(\text{Add}_m^{4l_{\text{max}}}\otimes \frac{H'(t_j)}{CD}\right) \notag \\
&\quad+ \sum_{m\in D^{K}} \sum_{j=0}^{C(\epsilon_2,K)}
   B_{j,m}^{(1)}
   \left(\text{Add}_m^{4l_{\text{max}}}\otimes \frac{H^{(j)}(1)}{CD^j(j!)^{\sigma}}\right) \notag \\
&\quad+ \sum_{m\in D^{K}} \sum_{j=0}^{C(\epsilon_2,K)}
   B_{j,m}^{(2)}
   \left(\text{Add}_m^{4l_{\text{max}}}\otimes \frac{H^{(j)}(0)}{CD^j(j!)^{\sigma}}\right) \notag \\
&\quad- \sum_{m\in D^{K}}\text{Add}_m^{4l_{\text{max}}}\otimes \widetilde{H}_{-m}\Bigg\|
\le 3\epsilon_2 .
\end{align}
Where 
\begin{equation}\label{eq:def_of_B0}
B_{j,m}^{(0)}=CD(A^{(0)}_{j,m}\mathbf{1}_{\{m\neq0\}}-\frac{1}{2}w_jt_j\mathbf{1}_{\{m=0\}}),
\end{equation}
\begin{equation}\label{eq:def_of_B1}
B_{j,m}^{(1)}=\frac{(-1)^m A^{(1)}_{j,m}}{-2m\pi i}\mathbf{1}_{\{m\neq 0\}}
   + A_j^{(3)}\mathbf{1}_{\{m=0\}},
\end{equation}
\begin{equation}\label{eq:def_of_B2}
B_{j,m}^{(2)}=\frac{(-1)^m A^{(2)}_{j,m}}{-2m\pi i}\mathbf{1}_{\{m\neq 0\}}
   + A_j^{(4)}\mathbf{1}_{\{m=0\}},
\end{equation}
satisfy
\begin{equation}
\sum_{m\in D^{K}}\sum_{j=0}^{L(\epsilon_2,K)}|B_{j,m}^{(0)}|\leq CD\left(\sum_{m\in D^{K}\setminus\{0\}}\sum_{j=0}^{L(\epsilon_2,K)}|A_{j,m}^{(0)}|+\frac{1}{2}\right)\leq \frac{CD(\text{log}K+3)}{\pi}\in\mathcal{O}(\text{log}K)
\end{equation}
and
\begin{align}
&\sum_{m\in D^{K}}\sum_{j=0}^{C(\epsilon_2,K)}|B_{j,m}^{(1)}|+|B_{j,m}^{(2)}| \notag \\
&\quad\leq \sum_{m\in D^{K}\setminus\{0\}}\sum_{j=0}^{C(\epsilon_2,K)}
       \frac{|A^{(1)}_{j,m}|+|A^{(2)}_{j,m}|}{2|m|\pi}
       +\sum_{j=0}^{C(\epsilon_2,K)}|A^{(3)}_{j}|+|A^{(4)}_{j}| \notag \\
&\quad\leq \frac{(8e^2+2)CD(\log K+1)}{\pi}
       \left(\sum_{j=0}^\infty \frac{j^{\sigma}}{(4e^{\sigma-1})^{j-1}}\right)
       +\sum_{j=0}^{\infty}C\left(\frac{1}{4e^{\sigma-1}}\right)^{\!j}
       +C \notag \\
&\quad\in\mathcal{O}(\log K).
\end{align}
\textbf{Step 3} (see also \cite{mizuta2023optimal}):\;To implement
\begin{equation}
\mathscr{H}_{\text{LP}}^{4l_{\max}}=\sum_{l\in D^{4l_{\max}}}l\omega|l \rangle\langle l|\,\otimes I,
\end{equation}
we define
\begin{equation}
\mathscr{O}_{\text{LP}}^{4l_{\max}} := \sum_{m \in D^{4l_{\max}}} |m\rangle\langle m| \otimes V_{m}^{4l_{\max}} =\left( \sum_{m,l;\,l-m \geq 0} |m,l\rangle\langle m,l| - \sum_{m,l;\,l-m < 0} |m,l\rangle\langle m,l|. \right)
\end{equation}
where
\begin{equation}
V_{m}^{4l_{\max}} =\sum_{l=m}^{4l_{\max}} |l\rangle\langle l| - \sum_{l=-4l_{\max}+1}^{m-1} |l\rangle\langle l|.
\end{equation}
Then we have
\begin{equation}
\langle a^{4l_{\max}} | \mathscr{O}_{\text{LP}}^{4l_{\max}} | a^{4l_{\max}} \rangle = \frac{\sum_{l\in D^{4l_{\max}}}l\omega|l \rangle\langle l|}{4l_{\max}\omega}.
\end{equation}
To build $\mathscr{O}_{\text{LP}}^{4l_{\max}} $, we introduce another single qubit system labeled by subscript $\ket{\cdot}_{\text{single}} $. Note that
\begin{equation}
\langle0|_{\text{single}}\text{Comp}^\dagger Z_{\text{single}}  \text{Comp}\,|0\rangle_{\text{single}} 
= \mathscr{O}_{\text{LP}}^{4l_{\max}} ,
\end{equation}
where
\begin{equation}\label{eq:def_of_Comp}
\text{Comp}\, |m,l\rangle |0\rangle_{\text{single}}  = 
\begin{cases}
|m,l\rangle |0\rangle_{\text{single}}  & \text{if } l \ge m, \\[4pt]
|m,l\rangle |1\rangle_{\text{single}}  & \text{if } l < m,
\end{cases}  
\end{equation}
and  $Z_{\text{single}}  $ denotes a Pauli $Z$ operator on the single qubit system. Comp can be composed of $\mathcal{O}(\text{log}\,l_{\max})$ elementary gates \cite{cuccaro2004new}. \\
\textbf{Step 4:}
For simplicity we consider the case $\text{log}_2\,(8l_{\max})$, $\text{log}_2\,(L(\epsilon_2,K)+1)$ and $\text{log}_2\,(C(\epsilon_2,K)+1)$ are all integers. We now proceed to construct a block-encoding of 
\begin{align}
&\widetilde{\mathscr{H}}_{\text{eff, pbc}}^{4l_{\max}}:=\sum_{m\in D^{K}}\sum_{j=0}^{L(\epsilon_2,K)}B_{j,m}^{(0)}
   \left(\text{Add}_m^{4l_{\text{max}}}\otimes\frac{H'(t_j)}{CD}\right) \notag \\
&\quad+ \sum_{m\in D^{K}} \sum_{j=0}^{C(\epsilon_2,K)}
   B_{j,m}^{(1)}
   \left(\text{Add}_m^{4l_{\text{max}}}\otimes \frac{H^{(j)}(1)}{CD^j(j!)^{\sigma}}\right) \notag \\
&\quad+ \sum_{m\in D^{K}} \sum_{j=0}^{C(\epsilon_2,K)}
   B_{j,m}^{(2)}
   \left(\text{Add}_m^{4l_{\text{max}}}\otimes \frac{H^{(j)}(0)}{CD^j(j!)^{\sigma}}\right)-\mathscr{H}_{\text{LP}}^{4l_{\max}}.
\end{align}
We prepare a system consisting of two working systems labeled by $a, b$ and five kinds of auxiliary systems labeled by 1$\sim$5. The working system $a$ is a $\lceil\text{log}_2\,d\rceil$-qubit system used to accommodate the matrix  $\frac{H'(t_j)}{CD}$ , $ \frac{H^{(j)}(1)}{CD^j(j!)^{\sigma}}$and $\frac{H^{(j)}(0)}{CD^j(j!)^{\sigma}}$.  The working system $b$ is a $\text{log}_2\,(8l_{\max})$-qubit system used to accommodate the unitary matrix $\text{Add}_m^{4l_{\text{max}}}$ and $\text{Comp}_{\text{working}}$, where $\text{Comp}_{\text{working}}$ represents the working-qubit part of Comp. The auxiliary system 1 is an $n_1$-qubit system prepared for the block-encoding
of $\frac{H'(t_j)}{CD}$ , $ \frac{H^{(j)}(1)}{CD^j(j!)^{\sigma}}$and $\frac{H^{(j)}(0)}{CD^j(j!)^{\sigma}}$, i.e.
\begin{equation}\label{eq:cond_of_O0}
\left\|\bra{0^{n_1}}O_{j}^{(0)}\ket{0^{n_1}}-\frac{H'(t_j)}{CD}\right\|\leq \epsilon_4,
\end{equation}
\begin{equation}\label{eq:cond_of_O1}
\left\|\bra{0^{n_1}}O_{j}^{(1)}\ket{0^{n_1}}-\frac{H^{(j)}(1)}{CD^j(j!)^{\sigma}}\right\|\leq \epsilon_4,
\end{equation}
\begin{equation}\label{eq:cond_of_O2}
\left\|\bra{0^{n_1}}O_{j}^{(2)}\ket{0^{n_1}}-\frac{H^{(j)}(0)}{CD^j(j!)^{\sigma}}\right\|\leq \epsilon_4.
\end{equation}
The auxiliary system 2 is a $\text{log}_2\,(L(\epsilon_2,K)+1)\vee\text{log}_2\,(C(\epsilon_2,K)+1)=:n_2$-qubit system characterized by $\{\ket{j}\}_{j=0}^{2^{n_2}-1 }$, which represent the truncation of the integrally-generated discrete series and the derivative truncation of the Gevrey expansion.
The auxiliary system 3 is a $\text{log}_2\,(8l_{\max})=:n_3$-qubit system characterized by $\{\ket{m}\}_{m\in D^{n_3}}:=\{\ket{-2^{n_3-1}+1},\ket{-2^{n_3-1}+2},...,\ket{2^{n_3-1}-1},\ket{2^{n_3-1}}\}$, which is used to form the linear combination of $\text{Add}_m^{4l_{\text{max}}}$ and accommodate $\text{Comp}_{\text{ancilla}}$, where $\text{Comp}_{\text{ancilla}}$ represents the ancilla-qubit part of Comp. The auxiliary system 4 is a 1-qubit system used to accommodate $\text{Comp}_{\text{single}}$ and $Z_{\text{single}} $, where $\text{Comp}_{\text{single}}$ represents the single-qubit part of Comp. The auxiliary system 5 is a $\lceil\text{log}_2\,(4)\rceil=2$ qubit system to combine the four terms in $\widetilde{\mathscr{H}}_{\text{eff, pbc}}^{4l_{\max}}$. The prepare oracle is defined by
\begin{align}
G_{\text{prep},L}
&= \Bigl( \sum_{k=0}^{2} (\ket{k}\bra{k})_5 \otimes I_4 \otimes G_{\text{coef},3\sim 2,L}^{(k)}
         \otimes I_{1,b,a}  + (\ket{3}\bra{3})_5 \otimes I_4 \otimes \mathscr{U}_{\text{ini}}^{4l_{\max}}
         \otimes I_{2,1,b,a} \Bigr) \notag \\
&\quad\quad\quad\cdot \Bigl( G_{\text{coef,5}} \otimes I_{4\sim1,b,a} \Bigr),\label{eq:def_of_prepL}
\end{align}
\begin{align}
G_{\text{prep},R}
&= \Bigl( \sum_{k=0}^{2} (\ket{k}\bra{k})_5 \otimes I_4 \otimes G_{\text{coef},3\sim2,R}^{(k)}
         \otimes I_{1,b,a} + (\ket{3}\bra{3})_5 \otimes I_4 \otimes \mathscr{U}_{\text{ini}}^{4l_{\max}}
         \otimes I_{2,1,b,a} \Bigr) \notag\\
&\quad\quad\quad\cdot \Bigl( G_{\text{coef,5}}  \otimes I_{4\sim1,b,a} \Bigr). \label{eq:def_of_prepR}
\end{align}
where
\begin{equation}\label{eq:cond_of_G_(012)1}
G_{\text{coef},3\sim 2,L}^{(k)}\ket{0}_3\ket{0}_2=\sum_{m\in D^{n_3}}\sum_{j=0}^{2^{n_2}-1}B_{j,m,L}^{(k)}\ket{m}_3\ket{j}_2\Big/\sqrt{\sum_{m\in D^{n_3}}\sum_{j=0}^{2^{n_2}-1}\left|B_{j,m,L}^{(k)}\right|^2},
\end{equation}
\begin{equation}\label{eq:cond_of_G_(012)2}
G_{\text{coef},3\sim 2,R}^{(k)}\ket{0}_3\ket{0}_2=\sum_{m\in D^{n_3}}\sum_{j=0}^{2^{n_2}-1}B_{j,m,R}^{(k)}\ket{m}_3\ket{j}_2\Big/\sqrt{\sum_{m\in D^{n_3}}\sum_{j=0}^{2^{n_2}-1}\left|B_{j,m,R}^{(k)}\right|^2},
\end{equation}
satisfy
\begin{equation}\label{eq:cond_of_G_(012)3}
\sum_{m\in D^{K}}\sum_{j=0}^{L(\epsilon_2,K)}\left|B_{j,m}^{(0)}- \frac{(\sum_{m\in D^{K}}\sum_{j=0}^{L(\epsilon_2,K)}|B_{j,m}^{(0)}|)\cdot(B_{j,m,L}^{(0)})^*B_{j,m,R}^{(0)}}{\sqrt{\sum_{m\in D^{n_3}}\sum_{j=0}^{2^{n_2}-1}\left|B_{j,m,L}^{(0)}\right|^2}\sqrt{\sum_{m\in D^{n_3}}\sum_{j=0}^{2^{n_2}-1}\left|B_{j,m,R}^{(0)}\right|^2}} \right|\leq \epsilon_3,
\end{equation}
\begin{equation}\label{eq:cond_of_G_(012)4}
\sum_{m\in D^{K}}\sum_{j=0}^{C(\epsilon_2,K)}\left|B_{j,m}^{(1)}- \frac{(\sum_{m\in D^{K}}\sum_{j=0}^{C(\epsilon_2,K)}|B_{j,m}^{(1)}|)\cdot(B_{j,m,L}^{(1)})^*B_{j,m,R}^{(1)}}{\sqrt{\sum_{m\in D^{n_3}}\sum_{j=0}^{2^{n_2}-1}\left|B_{j,m,L}^{(1)}\right|^2}\sqrt{\sum_{m\in D^{n_3}}\sum_{j=0}^{2^{n_2}-1}\left|B_{j,m,R}^{(1)}\right|^2}} \right|\leq \epsilon_3,
\end{equation}
\begin{equation}\label{eq:cond_of_G_(012)5}
\sum_{m\in D^{K}}\sum_{j=0}^{C(\epsilon_2,K)}\left|B_{j,m}^{(2)}- \frac{(\sum_{m\in D^{K}}\sum_{j=0}^{C(\epsilon_2,K)}|B_{j,m}^{(2)}|)\cdot(B_{j,m,L}^{(2)})^*B_{j,m,R}^{(2)}}{\sqrt{\sum_{m\in D^{n_3}}\sum_{j=0}^{2^{n_2}-1}\left|B_{j,m,L}^{(2)}\right|^2}\sqrt{\sum_{m\in D^{n_3}}\sum_{j=0}^{2^{n_2}-1}\left|B_{j,m,R}^{(2)}\right|^2}} \right|\leq \epsilon_3,
\end{equation}
and for all $(m,j)\notin D^K\times\{0,1,2,...,L(\epsilon_2,K)\}$,
\begin{equation}\label{eq:cond_of_G_(012)6}
(B_{j,m,L}^{(0)})^*B_{j,m,R}^{(0)}=0.
\end{equation}
For all $(m,j)\notin D^K\times\{0,1,2,...,C(\epsilon_2,K)\}$,
\begin{equation}\label{eq:cond_of_G_(012)7}
(B_{j,m,L}^{(1)})^*B_{j,m,R}^{(1)}=0,
\end{equation}
\begin{equation}\label{eq:cond_of_G_(012)8}
(B_{j,m,L}^{(2)})^*B_{j,m,R}^{(2)}=0.
\end{equation}
In addition, let
\begin{equation}\label{eq:cond_of_G5_1}
G_{\text{coef,5}}\ket{0}_5=\frac{\sum_{k=0}^3\sqrt{\tilde{D}^{(k)}}\ket{k}_5}{\sqrt{\sum_{k=0}^3\tilde{D}^{(k)}}},
\end{equation}
satisfy
\begin{equation}\label{eq:cond_of_G5_2}
\sum_{k=0}^3\left|D^{(k)}-\frac{(\sum_{k=0}^3D^{(k)})\tilde{D}^{(k)}}{(\sum_{k=0}^3\tilde{D}^{(k)})}\right|\leq \epsilon_3,
\end{equation}
where 
\begin{equation}\label{eq:def_of_D0}
D^{(0)}=\sum_{m\in D^{K}}\sum_{j=0}^{L(\epsilon_2,K)}\left|B_{j,m}^{(0)}\right|,
\end{equation}
\begin{equation}\label{eq:def_of_D2}
D^{(1)}=\sum_{m\in D^{K}}\sum_{j=0}^{C(\epsilon_2,K)}\left|B_{j,m}^{(1)}\right|,\quad\quad\quad D^{(2)}=\sum_{m\in D^{K}}\sum_{j=0}^{C(\epsilon_2,K)}\left|B_{j,m}^{(2)}\right|,
\end{equation}
\begin{equation}\label{eq:def_of_D3}
D^{(3)}=4l_{\max}\omega.
\end{equation}
The select oracle is defined by
\begin{equation}\label{eq:def_of_selec}
O_{\text{select}}=\sum_{k=0}^{2} (\ket{k}\bra{k})_5 \otimes O^{(k)}  + (\ket{3}\bra{3})_5 \otimes O_{\text{LP}},
\end{equation}
where
\begin{align}
O^{(0)}&=I_4\otimes\sum_{m\in D^K}(\ket{m}\bra{m})_3\otimes\sum_{j=0}^{L(\epsilon_2,K)}(\ket{j}\bra{j})_2\otimes \tikzmarknode{O_{j,1}^{(0)}}{O_{j,1}^{(0)}} \otimes(\text{Add}_m^{4l_{\text{max}}})_b\otimes\tikzmarknode{O_{j,a}^{(0)}}{O_{j,a}^{(0)}} \notag \\
&\quad+I_4\otimes\sum_{(m,j)\notin D^K\times\{0,1,2,...,L(\epsilon_2,K)\}}(\ket{m,j}\bra{m,j})_{3,2}\otimes I_{1,b,a},
\end{align}
\begin{align}
O^{(1)}&=I_4\otimes\sum_{m\in D^K}(\ket{m}\bra{m})_3\otimes\sum_{j=0}^{C(\epsilon_2,K)}(\ket{j}\bra{j})_2\otimes \tikzmarknode{O_{j,1}^{(1)}}{O_{j,1}^{(1)}} \otimes(\text{Add}_m^{4l_{\text{max}}})_b\otimes\tikzmarknode{O_{j,a}^{(1)}}{O_{j,a}^{(1)}} \notag \\
&\quad+I_4\otimes\sum_{(m,j)\notin D^K\times\{0,1,2,...,C(\epsilon_2,K)\}}(\ket{m,j}\bra{m,j})_{3,2}\otimes I_{1,b,a},
\end{align}
\begin{align}
O^{(2)}&=I_4\otimes\sum_{m\in D^K}(\ket{m}\bra{m})_3\otimes\sum_{j=0}^{C(\epsilon_2,K)}(\ket{j}\bra{j})_2\otimes \tikzmarknode{O_{j,1}^{(2)}}{O_{j,1}^{(2)}} \otimes(\text{Add}_m^{4l_{\text{max}}})_b\otimes\tikzmarknode{O_{j,a}^{(2)}}{O_{j,a}^{(2)}} \notag \\
&\quad+I_4\otimes\sum_{(m,j)\notin D^K\times\{0,1,2,...,C(\epsilon_2,K)\}}(\ket{m,j}\bra{m,j})_{3,2}\otimes I_{1,b,a},
\end{align}
\begin{tikzpicture}[overlay, remember picture]
  \draw[thick] ([yshift=-3pt]O_{j,1}^{(0)}.south) to[out=-30,in=-150] ([yshift=-3pt]O_{j,a}^{(0)}.south);
  \draw[thick] ([yshift=-3pt]O_{j,1}^{(1)}.south) to[out=-30,in=-150] ([yshift=-3pt]O_{j,a}^{(1)}.south);
  \draw[thick] ([yshift=-3pt]O_{j,1}^{(2)}.south) to[out=-30,in=-150] ([yshift=-3pt]O_{j,a}^{(2)}.south);
\end{tikzpicture}
\\
\bigskip
\begin{align}
\tilde{O}_{\text{LP}} &= 
\tikzmarknode{A}{\text{Comp}_{\text{single,ancilla}}} 
\otimes I_{2\sim 1} \otimes 
\tikzmarknode{C}{\text{Comp}_{\text{working}}} 
\otimes I_a,
\end{align}
\begin{tikzpicture}[overlay, remember picture]
  \draw[thick] ([yshift=-5pt]A.south) to[out=-30, in=-150] ([yshift=-5pt]C.south);
\end{tikzpicture}
\begin{equation}\label{eq:def_of_O_LP}
O_{\text{LP}}=\tilde{O}_{\text{LP}}^{\dagger}((e^{i\pi}Z)\otimes I_{3\sim 1,b,a})\tilde{O}_{\text{LP}}.
\end{equation}
In the above expression, the arc connecting the two symbols below represents the coupling between them, which together form a unitary gate. We now compute
\begin{equation}
\left\|\bra{0}_{5\sim 1}(G_{\text{prep},L})^{\dagger}O_{\text{select}}G_{\text{prep},R}\ket{0}_{5\sim 1}-\frac{\widetilde{\mathscr{H}}_{\text{eff, pbc}}^{4l_{\max}}}{\sum_{k=0}^3D^{(k)}}\right\|.
\end{equation}
Recall that
\begin{align}
(G_{\text{coef,5},L} \otimes I_{4\sim1,b,a} )\ket{0}_{5\sim 1}
&=\frac{\sum_{k=0}^3\sqrt{\tilde{D}^{(k)}}\ket{k}_5}{\sqrt{\sum_{k=0}^3\tilde{D}^{(k)}}}\ket{0}_{4\sim 1},
\end{align}
and
\begin{align}
&\Bigl( \sum_{k=0}^{2} (\ket{k}\bra{k})_5 \otimes I_4 \otimes G_{\text{coef},3\sim 2,L}^{(k)}
         \otimes I_{1,b,a}  + (\ket{3}\bra{3})_5 \otimes I_4 \otimes \mathscr{U}_{\text{ini}}^{4l_{\max}}
         \otimes I_{2,1,b,a} \Bigr)^{\dagger} \notag \\
&\qquad \cdot O_{\text{select}} \Bigl( \sum_{k=0}^{2} (\ket{k}\bra{k})_5 \otimes I_4 \otimes G_{\text{coef},3\sim2,R}^{(k)}
         \otimes I_{1,b,a} + (\ket{3}\bra{3})_5 \otimes I_4 \otimes \mathscr{U}_{\text{ini}}^{4l_{\max}}
         \otimes I_{2,1,b,a} \Bigr)   \notag \\
&\quad= \sum_{k=0}^{2} (\ket{k}\bra{k})_5 \otimes 
         \bigl( I_4 \otimes G_{\text{coef},3\sim 2,L}^{(k)} \otimes I_{1,b,a} \bigr)^{\dagger}
         O^{(k)}
         \bigl( I_4 \otimes G_{\text{coef},3\sim2,R}^{(k)} \otimes I_{1,b,a} \bigr) \notag \\
&\qquad + (\ket{3}\bra{3})_5 \otimes 
         \bigl( I_4 \otimes \mathscr{U}_{\text{ini}}^{4l_{\max}} \otimes I_{2,1,b,a} \bigr)^{\dagger}
         O_{\text{LP}}
         \bigl( I_4 \otimes \mathscr{U}_{\text{ini}}^{4l_{\max}} \otimes I_{2,1,b,a} \bigr).
\end{align}
Then
\begin{align}
& \Bigl\|\bra{0}_{5\sim 1}(G_{\text{prep},L})^{\dagger}O_{\text{select}}G_{\text{prep},R}\ket{0}_{5\sim 1} \notag \\
&\qquad -\frac{1}{\sum_{k=0}^{3}D^{(k)}}\bra{0}_{4\sim 1}
           \Bigl( \sum_{k=0}^{2} D^{(k)}\bigl( I_4 \otimes G_{\text{coef},3\sim 2,L}^{(k)}
                 \otimes I_{1,b,a} \bigr)^{\dagger}
                 O^{(k)}
                 \bigl( I_4 \otimes G_{\text{coef},3\sim2,R}^{(k)}
                 \otimes I_{1,b,a} \bigr) \notag \\
&\qquad\qquad + D^{(3)}\bigl( I_4 \otimes \mathscr{U}_{\text{ini}}^{4l_{\max}}
                 \otimes I_{2,1,b,a} \bigr)^{\dagger}
                 O_{\text{LP}}
                 \bigl( I_4 \otimes \mathscr{U}_{\text{ini}}^{4l_{\max}}
                 \otimes I_{2,1,b,a} \bigr) \Bigr)\ket{0}_{4\sim 1}\Bigr\| \notag \\
&= \Bigl\| \frac{1}{\sum_{k=0}^{3}\tilde{D}^{(k)}}\bra{0}_{4\sim 1}
          \Bigl( \sum_{k=0}^{2} \tilde{D}^{(k)}\bigl( I_4 \otimes G_{\text{coef},3\sim 2,L}^{(k)}
                \otimes I_{1,b,a} \bigr)^{\dagger}
                O^{(k)}
                \bigl( I_4 \otimes G_{\text{coef},3\sim2,R}^{(k)}
                \otimes I_{1,b,a} \bigr) \notag \\
&\qquad\qquad + \tilde{D}^{(3)}\bigl( I_4 \otimes \mathscr{U}_{\text{ini}}^{4l_{\max}}
                \otimes I_{2,1,b,a} \bigr)^{\dagger}
                O_{\text{LP}}
                \bigl( I_4 \otimes \mathscr{U}_{\text{ini}}^{4l_{\max}}
                \otimes I_{2,1,b,a} \bigr) \Bigr)\ket{0}_{4\sim 1} \notag \\
&\qquad -\frac{1}{\sum_{k=0}^{3}D^{(k)}}\bra{0}_{4\sim 1}
          \Bigl( \sum_{k=0}^{2} D^{(k)}\bigl( I_4 \otimes G_{\text{coef},3\sim 2,L}^{(k)}
                \otimes I_{1,b,a} \bigr)^{\dagger}
                O^{(k)}
                \bigl( I_4 \otimes G_{\text{coef},3\sim2,R}^{(k)}
                \otimes I_{1,b,a} \bigr) \notag \\
&\qquad\qquad + D^{(3)}\bigl( I_4 \otimes \mathscr{U}_{\text{ini}}^{4l_{\max}}
                \otimes I_{2,1,b,a} \bigr)^{\dagger}
                O_{\text{LP}}
                \bigl( I_4 \otimes \mathscr{U}_{\text{ini}}^{4l_{\max}}
                \otimes I_{2,1,b,a} \bigr) \Bigr)\ket{0}_{4\sim 1}\Bigr\| \notag \\
&\le \frac{\epsilon_3}{\sum_{k=0}^{3}D^{(k)}} .
\end{align}
Note that
\begin{align}
& \Bigl\| \bra{0}_{4\sim 1} \bigl( I_4 \otimes G_{\text{coef},3\sim 2,L}^{(0)}
         \otimes I_{1,b,a} \bigr)^{\dagger}
         O^{(0)}
         \bigl( I_4 \otimes G_{\text{coef},3\sim2,R}^{(0)}
         \otimes I_{1,b,a} \bigr) \ket{0}_{4\sim 1} \notag \\
&\qquad - \sum_{m\in D^{K}} \sum_{j=0}^{L(\epsilon_2,K)} 
          \frac{B_{j,m}^{(0)}}{D^{(0)}} \,
          \text{Add}_m^{4l_{\text{max}}} \otimes \frac{H'(t_j)}{CD} \Bigr\| \notag \\
&= \Bigl\| \sum_{m\in D^{K}} \sum_{j=0}^{L(\epsilon_2,K)}
          \frac{(B_{j,m,L}^{(0)})^* B_{j,m,R}^{(0)}}
               {\sqrt{\sum_{m\in D^{n_3}}\sum_{j=0}^{2^{n_2}-1} |B_{j,m,L}^{(0)}|^2}
                \sqrt{\sum_{m\in D^{n_3}}\sum_{j=0}^{2^{n_2}-1} |B_{j,m,R}^{(0)}|^2}}
          \text{Add}_m^{4l_{\text{max}}} \otimes \bra{0}_1 O_{j}^{(0)} \ket{0}_1 \notag \\
&\qquad - \sum_{m\in D^{K}} \sum_{j=0}^{L(\epsilon_2,K)}
          \frac{B_{j,m}^{(0)}}{\sum_{m\in D^{K}}\sum_{j=0}^{L(\epsilon_2,K)} |B_{j,m}^{(0)}|}
          \text{Add}_m^{4l_{\text{max}}} \otimes \bra{0}_1 O_{j}^{(0)} \ket{0}_1 \notag \\
&\qquad + \sum_{m\in D^{K}} \sum_{j=0}^{L(\epsilon_2,K)}
          \frac{B_{j,m}^{(0)}}{\sum_{m\in D^{K}}\sum_{j=0}^{L(\epsilon_2,K)} |B_{j,m}^{(0)}|}
          \text{Add}_m^{4l_{\text{max}}} \otimes \bra{0}_1 O_{j}^{(0)} \ket{0}_1 \notag \\
&\qquad - \sum_{m\in D^{K}} \sum_{j=0}^{L(\epsilon_2,K)}
          \frac{B_{j,m}^{(0)}}{\sum_{m\in D^{K}}\sum_{j=0}^{L(\epsilon_2,K)} |B_{j,m}^{(0)}|}
          \text{Add}_m^{4l_{\text{max}}} \otimes \frac{H'(t_j)}{CD} \Bigr\| \notag \\
&\le \frac{\epsilon_3}{D^{(0)}} + \epsilon_4 .
\end{align}
In a similar way,
\begin{align}
& \Bigl\| \bra{0}_{4\sim 1} \bigl( I_4 \otimes G_{\text{coef},3\sim 2,L}^{(1)}
         \otimes I_{1,b,a} \bigr)^{\dagger}
         O^{(1)}
         \bigl( I_4 \otimes G_{\text{coef},3\sim2,R}^{(1)}
         \otimes I_{1,b,a} \bigr) \ket{0}_{4\sim 1} \notag \\
&\qquad - \sum_{m\in D^{K}} \sum_{j=0}^{C(\epsilon_2,K)} 
          \frac{B_{j,m}^{(1)}}{D^{(1)}} \,
          \text{Add}_m^{4l_{\text{max}}} \otimes \frac{H^{(j)}(1)}{CD^j(j!)^{\sigma}} \Bigr\| \notag \\
&\le \frac{\epsilon_3}{D^{(1)}} + \epsilon_4 .
\end{align}
\begin{align}
& \Bigl\| \bra{0}_{4\sim 1} \bigl( I_4 \otimes G_{\text{coef},3\sim 2,L}^{(2)}
         \otimes I_{1,b,a} \bigr)^{\dagger}
         O^{(2)}
         \bigl( I_4 \otimes G_{\text{coef},3\sim2,R}^{(2)}
         \otimes I_{1,b,a} \bigr) \ket{0}_{4\sim 1} \notag \\
&\qquad - \sum_{m\in D^{K}} \sum_{j=0}^{C(\epsilon_2,K)} 
          \frac{B_{j,m}^{(2)}}{D^{(2)}} \,
          \text{Add}_m^{4l_{\text{max}}} \otimes \frac{H^{(j)}(0)}{CD^j(j!)^{\sigma}} \Bigr\| \notag \\
&\le \frac{\epsilon_3}{D^{(2)}} + \epsilon_4 .
\end{align}
In addition,
\begin{equation}
\bra{0}_{4\sim 1} (I_4 \otimes \mathscr{U}_{\text{ini}}^{4l_{\max}}
         \otimes I_{2,1,b,a} )^{\dagger} O_{\text{LP}} (I_4 \otimes \mathscr{U}_{\text{ini}}^{4l_{\max}}
         \otimes I_{2,1,b,a} )\ket{0}_{4\sim 1} =-\langle a^{4l_{\max}} |_3(\mathscr{O}_{\text{LP}}^{4l_{\max}} \otimes I_a)| a^{4l_{\max}} \rangle_3=-\frac{\mathscr{H}_{\text{LP}}^{4l_{\max}}}{4l_{\max}\omega}.
\end{equation}
Combining the estimations above we have
\begin{align}
& \Bigl\| \bra{0}_{5\sim 1} (G_{\text{prep},L})^{\dagger} O_{\text{select}} G_{\text{prep},R} \ket{0}_{5\sim 1}
        - \frac{\widetilde{\mathscr{H}}_{\text{eff, pbc}}^{4l_{\max}}}{\sum_{k=0}^{3}D^{(k)}} \Bigr\| \notag \\
&\le \frac{\epsilon_3}{\sum_{k=0}^{3}D^{(k)}}
     + \frac{1}{\sum_{k=0}^{3}D^{(k)}} \sum_{k=0}^{2} D^{(k)}
       \Bigl\| \bra{0}_{4\sim 1} \bigl( I_4 \otimes G_{\text{coef},3\sim 2,L}^{(k)}
             \otimes I_{1,b,a} \bigr)^{\dagger}
             O^{(k)}
             \bigl( I_4 \otimes G_{\text{coef},3\sim2,R}^{(k)}
             \otimes I_{1,b,a} \bigr) \ket{0}_{4\sim 1} \notag \\
&\qquad - \frac{\text{the }k\text{-th term in }\widetilde{\mathscr{H}}_{\text{eff, pbc}}^{4l_{\max}}}{D^{(k)}} \Bigr\| \notag \\
&\le \frac{4\epsilon_3}{\sum_{k=0}^{3}D^{(k)}}
     + \frac{\sum_{k=0}^{2} D^{(k)} \epsilon_4}{\sum_{k=0}^{3}D^{(k)}} .
\end{align}
Recall step 1$\sim$2, for any $\varepsilon_1>0$, let
\begin{equation}\label{eq:def_of_epsilon}
\epsilon_1=\frac{\sum_{k=0}^{3}D^{(k)}}{4}\varepsilon_1,\quad\quad\epsilon_2=\frac{\sum_{k=0}^{3}D^{(k)}}{12}\varepsilon_1,\quad\quad\epsilon_3=\frac{\sum_{k=0}^{3}D^{(k)}}{16}\varepsilon_1,\quad\quad\epsilon_4=\frac{\sum_{k=0}^{3}D^{(k)}}{4\sum_{k=0}^{2}D^{(k)}}\varepsilon_1,
\end{equation}
then
\begin{equation}
\Bigl\| \bra{0}_{5\sim 1} (G_{\text{prep},L})^{\dagger} O_{\text{select}} G_{\text{prep},R} \ket{0}_{5\sim 1}
        - \frac{\mathscr{H}_{\text{eff, pbc}}^{4l_{\max}}}{\sum_{k=0}^{3}D^{(k)}} \Bigr\| \leq \varepsilon_1.
\end{equation}
To summarize, we get
\begin{theorem}[BE of refined effective Hamiltonian in the general case]\label{th:BE_of_refined_effective_Hamiltonian_gene}
Consider $H(s)$ on $[0,1]$ and its Gevrey extension in \cref{sec:Periodic_Extension}, the decay of Fourier coefficients in \cref{sec:Four_decay}, $\varepsilon>0$ and the definition of $l_{\max}$ and $\mathscr{H}_{\text{\normalfont eff, pbc}}^{4l_{\max}}$ in \cref{sec:Error_analysis,sec:Error_analysis_5}. For any $\varepsilon_1>0$, $L(\epsilon_2,K),C(\epsilon_2,K)$ defined in \cref{eq:cond_of_L(eps)_1,eq:cond_of_L(eps)_2,eq:cond_of_C(eps)}; $A_{j,m}^{(k)}\;(k=0,1,2),\,A_{j}^{(3)},A_{j}^{(4)}$ defined in \cref{eq:def_of_A(0),eq:def_of_A(1),eq:def_of_A(2),eq:def_of_A(3),eq:def_of_A(4)}; $B_{j,m}^{(k)}\;(k=0,1,2)$ defined in \cref{eq:def_of_B0,eq:def_of_B1,eq:def_of_B2}; $D^{(k)}\;(k=0,1,2,3)$ defined in \cref{eq:def_of_D0,eq:def_of_D2,eq:def_of_D3}; $\epsilon_k\;(k=1,2,3,4)$ defined in \cref{eq:def_of_epsilon}. Given $G_{\text{\normalfont coef},3\sim 2,L}^{(k)}\;(k=0,1,2)$ and $G_{\text{\normalfont coef},3\sim 2,R}^{(k)}\;(k=0,1,2)$ satisfying \cref{eq:cond_of_G_(012)1,eq:cond_of_G_(012)2,eq:cond_of_G_(012)3,eq:cond_of_G_(012)4,eq:cond_of_G_(012)5,eq:cond_of_G_(012)6,eq:cond_of_G_(012)7,eq:cond_of_G_(012)8}; $G_{\text{\normalfont coef},5}$ satisfying \cref{eq:cond_of_G5_1,eq:cond_of_G5_2}; $O_j^{(k)}\;(k=0,1,2)$ satisfying \cref{eq:cond_of_O0,eq:cond_of_O1,eq:cond_of_O2}; $\text{\normalfont Add}_m^{4l_{\max}}$ in \cref{eq:def_of_Add} and $\sum_{m\in D^K}\ket{m}\bra{m}\otimes\text{\normalfont Add}_m^{4l_{\max}}$ which both can be implemented by $\mathcal{O}(\text{\normalfont log}\,l_{\max})$ elementary gates; $\text{\normalfont Comp}$ in \cref{eq:def_of_Comp} which can be implemented by $\mathcal{O}(\text{\normalfont log}\,l_{\max})$ elementary gates. We can construct prepare oracles $G_{\text{\normalfont prep},L}$ and $G_{\text{\normalfont prep},R}$ as in \cref{eq:def_of_prepL,eq:def_of_prepR} and a select oracle $O_{\text{\normalfont select}}$ as in \cref{eq:def_of_selec}, with five ancillary-qubit systems together comprising 
\begin{equation}
n_1+\text{\normalfont log}_2\,(L(\epsilon_2,K)+1)\vee\text{\normalfont log}_2\,(C(\epsilon_2,K)+1)+\text{\normalfont log}_2(8l_{\max})+3
\end{equation}
ancilla qubits, such that
\begin{equation}
\Bigl\| \bra{0}_{5\sim 1} (G_{\text{\normalfont prep},L})^{\dagger} O_{\text{\normalfont select}} G_{\text{\normalfont prep},R} \ket{0}_{5\sim 1}
        - \frac{\mathscr{H}_{\text{\normalfont eff, pbc}}^{4l_{\max}}}{\sum_{k=0}^{3}D^{(k)}} \Bigr\| \leq \varepsilon_1,
\end{equation}
using 1 query of
\begin{equation}
\text{\normalfont HAM}'\text{\normalfont -T}:=\sum_{j=0}^{L(\epsilon_2,K)}\ket{j}\bra{j}\otimes O^{(0)}_j,
\end{equation}
\begin{equation}
\text{\normalfont HOHAM}\text{\normalfont -1}:=\sum_{j=0}^{C(\epsilon_2,K)}\ket{j}\bra{j}\otimes O^{(1)}_j\;\;\text{and}\;\;\text{\normalfont HOHAM}\text{\normalfont -0}:=\sum_{j=0}^{C(\epsilon_2,K)}\ket{j}\bra{j}\otimes O^{(2)}_j,
\end{equation}
respectively.
\end{theorem}
\qed\\
\textbf{Step 5:} Under \cref{th:BE_of_refined_effective_Hamiltonian_gene} and step 3 above, denote
\begin{equation}
U_1=(G_{\text{prep},L})^{\dagger} O_{\text{select}} G_{\text{prep},R},
\end{equation}
\begin{equation}\label{eq:def_of_U_2}
U_2=(I_4\otimes \mathscr{U}_{\text{ini}}^{4l_{\max}}\otimes I_{b,a})^{\dagger}(\text{Comp}\otimes I_a)^\dagger (Z_{4}\otimes I_{3,b,a})  (\text{Comp}\otimes I_a)(I_4\otimes \mathscr{U}_{\text{ini}}^{4l_{\max}}\otimes I_{b,a}).
\end{equation}
We can then construct a block encoding of $e^{-i\mathscr{H}_{\text{eff, pbc}}^{4l_{\max}}t}$ and $e^{-i\mathscr{H}_{\text{LP}}^{4l_{\max}}t}$ using QSVT such that
\begin{equation}
\left\|V_1(t)-e^{-i\mathscr{H}_{\text{eff, pbc}}^{4l_{\max}}t}\right\|\leq 2t(\sum_{k=0}^{3}D^{(k)})\varepsilon_1,
\end{equation}
\begin{equation}
\left\|V_2(t)-e^{-i\mathscr{H}_{\text{LP}}^{4l_{\max}}t}\right\|\leq 2tD^{(3)}\varepsilon_2.
\end{equation}
Where $V_1$ is comprised of 
\begin{equation}
6(\sum_{k=0}^{3}D^{(k)})t+9\text{log}\left(\frac{12}{2t(\sum_{k=0}^{3}D^{(k)})\varepsilon_1}\right)
\end{equation}
uses of $U_1$ or its inverse, 3 uses of controlled-$U_1$ or its inverse, with
\begin{equation}
n_1+\text{log}_2\,(L(\epsilon_2,K)+1)\vee\text{log}_2\,(C(\epsilon_2,K)+1)+\text{log}_2(8l_{\max})+5
\end{equation}
ancilla qubits. $V_2$ is comprised of 
\begin{equation}
6D^{(3)}t+9\text{log}\left(\frac{6}{2tD^{(3)}\varepsilon_2}\right)
\end{equation}
uses of $U_2$ or its inverse, 3 uses of controlled-$U_2$ or its inverse, with
\begin{equation}
\text{log}_2(8l_{\max})+3
\end{equation}
ancilla qubits. Let
\begin{align}
\widetilde{\mathscr{U}}_{\text{amp1,pbc}}^{l_{\max}}(t) 
&= (\mathscr{U}_{\text{ini}}^{4l_{\max}})^{\dagger} V_2(t) V_1(t)\mathscr{U}_{\text{ini}}^{l_{\max}} ,\\
\widetilde{\mathscr{U}}_{\text{amp2,pbc}}^{l_{\max}}(t)&= -\widetilde{\mathscr{U}}_{\text{amp1,pbc}}^{l_{\max}}(t) \mathscr{R}[\widetilde{\mathscr{U}}_{\text{amp1,pbc}}^{l_{\max}}(t) ]^{\dagger}\mathscr{R}\widetilde{\mathscr{U}}_{\text{amp1,pbc}}^{l_{\max}}(t) .
\end{align}
Then
\begin{equation}
\left\|\widetilde{\mathscr{U}}_{\text{amp2,pbc}}^{l_{\max}}(t)-\mathscr{U}_{\text{amp2,pbc}}^{l_{\max}}(t)\right\|\leq 6t(\sum_{k=0}^{3}D^{(k)})\varepsilon_1+6tD^{(3)}\varepsilon_2.
\end{equation}
Denote
\begin{equation}\label{eq:def_of_E0}
E^{(0)}:=\frac{CD(\text{log}l_{\max}+3)}{\pi}\geq D^{(0)},
\end{equation}
\begin{equation}\label{eq:def_of_E1}
E^{(1)}:=\frac{(8e^2+2)CD(\log l_{\max}+1)}{\pi}
       \left(\sum_{j=0}^\infty \frac{j^{\sigma}}{(4e^{\sigma-1})^{j-1}}\right)
       +\sum_{j=0}^{\infty}C\left(\frac{1}{4e^{\sigma-1}}\right)^{\!j}
       +C\geq D^{(1)}+D^{(2)}.
\end{equation}
Pick
\begin{equation}
\varepsilon_1=\frac{\varepsilon}{24t(\sum_{k=0}^{1}E^{(k)}+D^{(3)})},\quad\varepsilon_2=\frac{\varepsilon}{24tD^{(3)}}.
\end{equation}
Combining \cref{th:Refined_effective_Hamiltonian}, we finally get that
\begin{equation}
\left\| \widetilde{\mathscr{U}}_{\text{amp2, pbc}}^{l_{\max}}(t) | 0 \rangle | \psi(0) \rangle - |0\rangle |\psi(t)\rangle \right\| \leq  \varepsilon.
\end{equation}\\
\textbf{Step 6:} We now focus on the query complexity of $U_1$ given by
\begin{equation}
18(\sum_{k=0}^{3}D^{(k)})t+27\text{log}\left(\frac{144(\sum_{k=0}^{1}E^{(k)}+D^{(3)})}{(\sum_{k=0}^{3}D^{(k)})\varepsilon}\right).
\end{equation}
Although substituting the definition of $D^{(k)}$ in \cref{eq:def_of_D0,eq:def_of_D2,eq:def_of_D3} and $E^{(k)}$ in \cref{eq:def_of_E0,eq:def_of_E1} into the formula above yields an upper bound on the query complexity, we here derive the asymptotic dependence. Note that For $k=0,1$,
\begin{equation}
E^{(k)}=\mathcal{O}(\alpha \text{log}l_{\max})=\mathcal{O}\left(\alpha \text{log}\left(\alpha t+\text{log}\left(\frac{1}{\varepsilon}\right)\right)\right).
\end{equation}
Hence
\begin{equation}
K=\left\lceil\left(2\zeta\text{log}\left(\frac{2S(2\zeta) h}{\epsilon_1}\right)\right)  ^{\varrho}\right\rceil\bigwedge \;l_{\text{max}}\;=\mathcal{O}\left(\left(\text{log}\,\alpha t+\text{log}\left(\frac{1}{\varepsilon}\right)\right)^{\sigma}\left(\log\left(\alpha t+\log\left(\frac{1}{\varepsilon}\right)\right)\right)\right).
\end{equation}
Then for $k=0,1,2$,
\begin{equation}
D^{(k)}=\mathcal{O}(\alpha \text{log}K)=\mathcal{O}\left(\alpha \text{log}\text{log}\,\alpha t+\alpha \text{log}\text{log}\left(\frac{1}{\varepsilon}\right)\right),
\end{equation}
and
\begin{equation}
D^{(3)}=\mathcal{O}\left(\omega\left(\alpha t+  \log\left(\frac{1}{\varepsilon}\right)\right)^{\sigma}\text{log}\left(\alpha t+ \log\left(\frac{1}{\varepsilon}\right)\right)\left(\text{log}\text{log}\left(\alpha t+ \log\left(\frac{1}{\varepsilon}\right)\right)\right)^{\sigma}\right).
\end{equation}
Hence
\begin{align}
&18(\sum_{k=0}^{3}D^{(k)})t+27\text{log}\left(\frac{144(\sum_{k=0}^{1}E^{(k)}+D^{(3)})}{(\sum_{k=0}^{3}D^{(k)})\varepsilon}\right)\notag\\
&\quad\quad=\mathcal{O}\left((\sum_{k=0}^{3}D^{(k)})t+\text{log}\left(\frac{144(\sum_{k=0}^{1}E^{(k)}+D^{(3)})}{(\sum_{k=0}^{3}D^{(k)})\varepsilon}\right)\right)\notag \\
&\quad\quad=\mathcal{O}\left(\left(\alpha t+  \log\left(\frac{1}{\varepsilon}\right)\right)^{\sigma}\text{log}\left(\alpha t+ \log\left(\frac{1}{\varepsilon}\right)\right)\left(\text{log}\text{log}\left(\alpha t+ \log\left(\frac{1}{\varepsilon}\right)\right)\right)^{\sigma}\right)\notag \\
&\quad\quad=\widetilde{\mathcal{O}}\left(\left(\alpha t+  \log\left(\frac{1}{\varepsilon}\right)\right)^{\sigma}\right).
\end{align}
By exactly the same derivation, we find that the query complexity of $U_2$ is also 
$\widetilde{\mathcal{O}}\left(\left(\alpha t + \log(1/\varepsilon)\right)^{\sigma}\right)$.

We next consider the complexity of the additional gates used in the algorithm. This complexity firstly arises from prepare oracles for the linear combination of unitary matrices in \cref{eq:def_of_prepL,eq:def_of_prepR}, including 
$G_{\text{\normalfont coef},3\sim 2,L}^{(k)}\ (k=0,1,2)$, 
$G_{\text{\normalfont coef},3\sim 2,R}^{(k)}\ (k=0,1,2)$, and 
$G_{\text{\normalfont coef},5}$. 
Note that the expressions for these gates contain at most 
$\mathcal{O}(l_{\max}\cdot (L(\epsilon_2,K)\vee C(\epsilon_2,K)))$ terms, where $L(\epsilon_2,K)\vee C(\epsilon_2,K)=\mathcal{O}\bigl((\log\alpha t + \log(1/\varepsilon))^{\sigma}\bigr)$. We make the worst-case assumption that a prepare oracle consisting of $\mathcal{O}(l_{\max}\cdot (L(\epsilon_2,K)\vee C(\epsilon_2,K)))$ terms requires $\mathcal{O}(l_{\max}\cdot (L(\epsilon_2,K)\vee C(\epsilon_2,K)))$ elementary gates. Therefore, $G_{\text{prep},L}$ and $G_{\text{prep},R}$ can be implemented with at most 
$\widetilde{\mathcal{O}}\bigl((\alpha t + \log(1/\varepsilon))^{\sigma}\cdot(\log(1/\varepsilon))^{\sigma}\bigr)$ 
elementary gates. Since $U_1$ is queried 
$\widetilde{\mathcal{O}}\left(\left(\alpha t + \log(1/\varepsilon)\right)^{\sigma}\right)$ 
times, the elementary gate complexity of the prepare oracle part is 
$\widetilde{\mathcal{O}}\bigl((\alpha t + \log(1/\varepsilon))^{2\sigma}\cdot(\log(1/\varepsilon))^{\sigma}\bigr)$.

The gate complexity of the second part comes from the adder gates and compare gates used in the select oracle in \cref{eq:def_of_selec}. Since the gate $O_{\text{\normalfont select}} $ contains $\mathcal{O}(1)$ copies of 
$\sum_{m\in D^K}\ket{m}\bra{m}\otimes \text{\normalfont Add}_m^{4l_{\max}}$ and $\text{Comp}$, 
which both can be implemented with $\mathcal{O}(\log l_{\max})$ elementary gates, the parts of the gate $O_{\text{\normalfont select}} $ other than HAM-T part can be implemented with 
$\mathcal{O}(\log l_{\max}) = \mathcal{O}\bigl(\log(\alpha t + \log(1/\varepsilon))\bigr)$ 
elementary gates. Since the gate $O_{\text{\normalfont select}}$ is queried 
$\widetilde{\mathcal{O}}\left(\left(\alpha t + \log(1/\varepsilon)\right)^{\sigma}\right)$ 
times in total, the elementary gate complexity of the select oracle part is 
$\widetilde{\mathcal{O}}\left(\left(\alpha t + \log(1/\varepsilon)\right)^{\sigma}\right)$.

The third part of the gate complexity comes from the compare gate $\mathrm{Comp}$ and the Hadamard gate $\mathscr{U}_{\text{ini}}^{4l_{\max}}$ in $U_2$ in \cref{eq:def_of_U_2}, both of which can be implemented with $\mathcal{O}(\log l_{\max})$ elementary gates. Since $U_2$ is queried 
$\widetilde{\mathcal{O}}\left(\left(\alpha t + \log(1/\varepsilon)\right)^{\sigma}\right)$ 
times in total, the elementary gate complexity of the $U_2$ part is 
$\widetilde{\mathcal{O}}\left(\left(\alpha t + \log(1/\varepsilon)\right)^{\sigma}\right)$.

According to \cite{GilyenSuLowEtAl2019}, the final additional gate complexity should be multiplied by the number of ancilla qubits in \cref{th:BE_of_refined_effective_Hamiltonian_gene}. Since the number of ancilla qubits in the circuit is $\widetilde{\mathcal{O}}(\log l_{\max})$, this does not affect the asymptotic complexity of the additional elementary gates. Combining the above derivations, we find that the additional gate complexity required by the algorithm is 
$\widetilde{\mathcal{O}}\bigl((\alpha t + \log(1/\varepsilon))^{2\sigma}\cdot(\log(1/\varepsilon))^{\sigma}\bigr)$. 
This leads to the following theorem.

\begin{theorem}[Query and gate complexity in the general Hamiltonian case]\label{th:query_gate_complex_gene}
Let $1 \leq \sigma < 2$ and $H(s)\in G^{\sigma}([0,1])$ satisfy $||H(s)||\leq \alpha$. Define $\widetilde{H}(t) = H(t/T)$.
Suppose that we are given access to the $\text{\normalfont HAM}'\text{\normalfont -T}$, $\text{\normalfont HOHAM}\text{\normalfont -1}$ and $\text{\normalfont HOHAM}\text{\normalfont -0}$ oracles, where $\text{\normalfont HAM}'\text{\normalfont -T}$ is the block encoding of $\sum_{j=1}^{L(\epsilon_2,K)}\ket{j}\bra{j}\otimes H'(t_j)$, \text{\normalfont HOHAM}\text{\normalfont -1} is the block encoding of $\sum_{j=0}^{C(\epsilon_2,K)}\ket{j}\bra{j}\otimes H^{(j)}(1)$, \text{\normalfont HOHAM}\text{\normalfont -0} is the block encoding of $\sum_{j=0}^{C(\epsilon_2,K)}\ket{j}\bra{j}\otimes H^{(j)}(0)$. Then there exists a quantum algorithm for solving~\cref{eqn:ham_sim} up to time $t\in[0,T]$ with error at most $\varepsilon$, using $\widetilde{\mathcal{O}}\left(\left(\alpha t+  \log\left(1/\varepsilon\right)\right)^{\sigma}\right)$ queries for $\text{\normalfont HAM}'\text{\normalfont -T}$, $\text{\normalfont HOHAM}\text{\normalfont -1}$ and $\text{\normalfont HOHAM}\text{\normalfont -0}$ oracles, and can be implemented using $\widetilde{\mathcal{O}}\bigl((\alpha t + \log(1/\varepsilon))^{2\sigma}\cdot(\log(1/\varepsilon))^{\sigma}\bigr)$ additional elementary gates. 
\end{theorem}
\qed

\subsection{The control Hamiltonian case}
\subsubsection{Oracle definitions for the control Hamiltonian case}\label{subsubsec:oracle_control_case}

We now turn to the control Hamiltonian case, where the Hamiltonian takes the special form
\begin{equation}
H(s)=\sum_{j=0}^{j_{\max}-1}\alpha_j(s)M_j,
\end{equation}
with time-independent Hermitian matrices $M_j$ and Gevrey-class coefficient functions $\alpha_j(s)$. In this setting, we no longer require direct oracles for the derivatives of $H(s)$; instead, we only assume access to the HAM-T oracle
\begin{equation}
\mathrm{HAM-T} := \sum_{j=0}^{j_{\max}-1} |j\rangle\langle j| \otimes O_j,
\end{equation}
where $O_j$ is a block encoding of $M_j$ for each $j$, i.e.,
\begin{equation}
\left\| \langle 0^{n_1}| O_j |0^{n_1}\rangle - M_j \right\| \le \epsilon_4.
\end{equation}

The key observation is that the Fourier coefficients $\widetilde{H}_{-m}$ of the extended Hamiltonian can be expressed as linear combinations of the $M_j$'s with coefficients $\alpha_j^m$:
\begin{equation}
\widetilde{H}_{-m} = \sum_{j=0}^{j_{\max}-1} \alpha_j^m M_j,
\end{equation}
where
\begin{equation}
\alpha_j^m = \frac{1}{2}\int_0^2 \hat{\alpha}_j(t) e^{im\pi t}\,dt.
\end{equation}
Note that $\hat{\alpha}_j$ is the period extension of $\alpha_j$. Using the Gevrey extension constructed in \cref{sec:Periodic_Extension} and the Clenshaw Curtis quadrature, we can classically precompute approximations $\tilde{\alpha}_j^m$ satisfying
\begin{equation}
|\tilde{\alpha}_j^m - \alpha_j^m| \le \frac{\epsilon_2}{K j_{\max}}
\end{equation}
for all $m\in D^K$ and $0\le j<j_{\max}$, where $K$ is the Fourier truncation order defined in Step~1 of \cref{subsec:general_case}. The detailed computation of $\tilde{\alpha}_j^m$ is presented in Step~2 below.

With these classically computed coefficients, the block encoding of $\mathscr{H}_{\mathrm{eff,pbc}}^{4l_{\max}}$ is again constructed via the LCU technique. The prepare oracle $G_{\mathrm{prep},L}$ and $G_{\mathrm{prep},R}$ are defined as
\begin{align}
G_{\mathrm{prep},L} &= \left((|0\rangle\langle 0|)_5\otimes I_4\otimes G_{\mathrm{coef},3\sim 2,L}\otimes I_{1,b,a} + (|1\rangle\langle 1|)_5\otimes I_4\otimes \mathcal{U}_{\mathrm{ini}}^{4l_{\max}}\otimes I_{2,1,b,a}\right) \notag \\
&\qquad \cdot \left(G_{\mathrm{coef},5}\otimes I_{4\sim 1,b,a}\right),
\end{align}
\begin{align}
G_{\mathrm{prep},R} &= \left((|0\rangle\langle 0|)_5\otimes I_4\otimes G_{\mathrm{coef},3\sim 2,R}\otimes I_{1,b,a} + (|1\rangle\langle 1|)_5\otimes I_4\otimes \mathcal{U}_{\mathrm{ini}}^{4l_{\max}}\otimes I_{2,1,b,a}\right) \notag \\
&\qquad \cdot \left(G_{\mathrm{coef},5}\otimes I_{4\sim 1,b,a}\right),
\end{align}
where the auxiliary system 5 is now a single qubit that selects between the two terms in $\mathscr{H}_{\mathrm{eff,pbc}}^{4l_{\max}}$: the linear combination of $M_j$'s ($k=0$) and the linear potential term ($k=1$). The coefficient preparation unitaries $G_{\mathrm{coef},3\sim 2,L}$ and $G_{\mathrm{coef},3\sim 2,R}$ encode the classically computed coefficients $\tilde{\alpha}_j^m$ as
\begin{align}
G_{\mathrm{coef},3\sim 2,L}|0\rangle_3|0\rangle_2 &= \sum_{m\in D^{n_3}}\sum_{j=0}^{2^{n_2}-1} \tilde{\alpha}_{j,m,L} |m\rangle_3|j\rangle_2 \bigg/ \sqrt{\sum_{m\in D^{n_3}}\sum_{j=0}^{2^{n_2}-1} |\tilde{\alpha}_{j,m,L}|^2}, \\
G_{\mathrm{coef},3\sim 2,R}|0\rangle_3|0\rangle_2 &= \sum_{m\in D^{n_3}}\sum_{j=0}^{2^{n_2}-1} \tilde{\alpha}_{j,m,R} |m\rangle_3|j\rangle_2 \bigg/ \sqrt{\sum_{m\in D^{n_3}}\sum_{j=0}^{2^{n_2}-1} |\tilde{\alpha}_{j,m,R}|^2},
\end{align}
with the property that
\begin{equation}
\sum_{m\in D^K}\sum_{j=0}^{j_{\max}-1} \left| \tilde{\alpha}_j^m - \frac{(\sum_{m,j}|\tilde{\alpha}_j^m|)\cdot (\tilde{\alpha}_{j,m,L})^* \tilde{\alpha}_{j,m,R}}{\sqrt{\sum_{m,j}|\tilde{\alpha}_{j,m,L}|^2}\sqrt{\sum_{m,j}|\tilde{\alpha}_{j,m,R}|^2}} \right| \le \epsilon_3,
\end{equation}
and for all $(m,j)\notin D^K\times\{0,\ldots,j_{\max}-1\}$,
\begin{equation}
(\tilde{\alpha}_{j,m,L})^* \tilde{\alpha}_{j,m,R} = 0.
\end{equation}
The unitary $G_{\mathrm{coef},5}$ prepares the normalization factors for the two terms.

The select oracle is defined by
\begin{equation}
O_{\mathrm{select}} = (|0\rangle\langle 0|)_5 \otimes O^{(0)} + (|1\rangle\langle 1|)_5 \otimes O_{\mathrm{LP}},
\end{equation}
where $O^{(0)}$ is the controlled unitary that applies $\mathrm{Add}_m^{4l_{\max}}\otimes O_j$ conditioned on the indices $m$ and $j$, and $O_{\mathrm{LP}}$ implements the linear potential $\mathscr{H}_{\mathrm{LP}}^{4l_{\max}}$ via the comparison gate Comp as defined in Step~3 of \cref{subsec:general_case}.

With these oracles, the block encoding of $\mathscr{H}_{\mathrm{eff,pbc}}^{4l_{\max}}$ takes the form
\begin{equation}
\left\| \langle 0|_{5\sim 1}(G_{\mathrm{prep},L})^{\dagger}O_{\mathrm{select}}G_{\mathrm{prep},R}|0\rangle_{5\sim 1} - \frac{\mathscr{H}_{\mathrm{eff,pbc}}^{4l_{\max}}}{\sum_{k=0}^{1}D^{(k)}}\right\| \le \epsilon_1,
\end{equation}
where
\begin{equation}
D^{(0)} = \sum_{m\in D^K}\sum_{j=0}^{j_{\max}-1} |\tilde{\alpha}_j^m|, \qquad D^{(1)} = 4l_{\max}\omega
\end{equation}
are the normalization constants for the two terms. The construction of these oracles and the proof of the above block-encoding error are detailed in Steps~2--4 below. The total number of ancilla qubits required for this construction is
\begin{equation}
n_1 + \lceil \log_2 j_{\max}\rceil + \log_2(8l_{\max}) + 2,
\end{equation}
where $n_1$ is the number of ancilla qubits used in the block encoding of $M_j$.
\subsubsection{Algorithm implementation for the control Hamiltonian case}\label{subsec:control_case}
We consider the control Hamiltonian $ H(s) $ written by
\begin{equation}
H(s) = \sum_{j=0}^{j_{\max}-1} \alpha_j(s) M_j
\end{equation}
with coefficients $\{\alpha_j(s)\}_{j=1}^{j_{\max}}\subseteq G^{\sigma}([0,1])$ satisfying 
\begin{equation}
\exists\; C,D >0,\;\text{s.t.}\quad|\alpha_j^{(n)}(s)|\leq C D^n(n!)^{\sigma} \quad \text{for any}\;n\geq 0,\;j=0,1,...,j_{\max}-1,
\end{equation}
and operators $\{M_j\}_{j=1}^{j_{\max}}$ satisfying $\left\|M_j\right\|\leq 1$. Then
\begin{equation}
\|H^{(n)}(s)\| \le j_{\max}C\cdot D^n (n!)^{\sigma}, \quad \forall\; s\in[0,1]\;, n\geq 0.
\end{equation}
As in \cref{sec:Periodic_Extension}, we can extend the family $\{\alpha_j(s)\}_{j=1}^{j_{\max}}$ from $[0,1]$ to be $\{\hat{\alpha}_j(s)\}_{j=1}^{j_{\max}}$ defined on $[0,2]$; the analyses of \cref{sec:Four_decay} and \cref{sec:Error_analysis} then carry over verbatim, with the original constant $C$ replaced by $j_{\max}C$ and thus $l_{\max}$ in \cref{eq:def_of_lmax} replaced by $j_{\max}^{\varrho}l_{\max}$. \\
\textbf{Step 1:} As stated in step 1 of \cref{subsec:general_case}, for $\epsilon_1>0$, we still choose 
\begin{equation}
K=\left\lceil\left(2\zeta\text{log}\left(\frac{2S(2\zeta) h}{\epsilon_1}\right)\right)  ^{\varrho}\right\rceil\bigwedge \;l_{\text{max}}\;.
\end{equation}
and focus on the implementation of 
\begin{equation}
\sum_{m \in D^{K}} \text{Add}_{m}^{4l_{\max}} \otimes \widetilde{H}_{-m}.
\end{equation}\\
\textbf{Step 2:} The Fourier components are given by
\begin{equation}
\widetilde{H}_{-m} = \sum_{j=0}^{j_{\max}-1}\alpha_j^m M_j,
\end{equation}
where
\begin{align}
\alpha_j^m &= \frac{1}{2T} \int_0^{2T} \hat{\alpha}_j(t/T) e^{im\frac{2\pi}{2T} t}\, dt \notag \\
&= \frac{1}{2} \int_0^{2} \hat{\alpha}_j(t) e^{im\pi t}\, dt \notag \\
&= \frac{1}{2} \int_0^{1} \alpha_j(t) e^{im\pi t}\, dt
   + \frac{(-1)^m}{2} \int_0^{1} \hat{\alpha}_j(t+1) e^{im\pi t}\, dt \notag \\
&= \frac{1}{2} \int_0^{1} \alpha_j(t) e^{im\pi t}\, dt \notag \\
&\quad + \frac{(-1)^m}{2} \int_0^{1} \Bigl( \sum_{n=0}^\infty
          \frac{\alpha_j^{(n)}(1)}{n!} t^n \chi_\tau(R_n t) \notag \\
&\qquad\qquad\qquad\qquad\qquad
          + \sum_{n=0}^\infty
          \frac{\alpha_j^{(n)}(0)}{n!} (t-1)^n \chi_\tau(R_n (t-1)) \Bigr)
          e^{im\pi t}\, dt \notag \\
&= \frac{1}{2} \int_0^{1} \alpha_j(t) e^{im\pi t}\, dt \notag \\
&\quad + \frac{(-1)^m}{2} \sum_{n=0}^\infty C
       \Bigl( \int_0^{1} 
             \frac{D^n(n!)^\sigma}{n!} t^n \chi_\tau(R_n t) e^{im\pi t}\, dt \Bigr)
       \frac{\alpha_j^{(n)}(1)}{CD^n(n!)^\sigma} \notag \\
&\quad + \frac{(-1)^m}{2} \sum_{n=0}^\infty C
       \Bigl( \int_0^{1} 
             \frac{D^n(n!)^\sigma}{n!} (t-1)^n \chi_\tau(R_n (t-1)) e^{im\pi t}\, dt \Bigr)
       \frac{\alpha_j^{(n)}(0)}{CD^n(n!)^\sigma}.
\end{align}
As in step 2 in \cref{subsec:general_case}, note that $ \left|\frac{\alpha_j^{(n)}(1)}{CD^n(n!)^\sigma}\right|\leq 1$ and
\begin{equation}
\sum_{n=0}^\infty C\left|
        \int_0^{1} 
             \frac{D^n(n!)^\sigma}{n!} t^n \chi_\tau(R_n t) e^{im\pi t}\, dt \right|\leq 
C\sum_{n=0}^\infty \left(\frac{1}{4e^{\sigma-1}}\right)^n.
\end{equation}
For $0<\epsilon_2<\frac{1}{2}$, we pick $C(\epsilon_2,K)$
such that
\begin{equation}\label{eq:def_of_C_eps2_control}
C\sum_{n=C(\epsilon_2,K)+1}^\infty \left(\frac{1}{4e^{\sigma-1}}\right)^n\leq \frac{\epsilon_2}{2Kj_{\max}}.
\end{equation}
Assume further that we can classically compute approximations $\tilde\alpha_j^m$ satisfying
\begin{align}
\Big|\tilde{\alpha}_j^m
&-\Big(
\frac{1}{2} \int_0^{1} \alpha_j(t) e^{im\pi t}\, dt \notag\\
&\quad + \frac{(-1)^m}{2} \sum_{n=0}^{C(\epsilon_2,K)} C
       \Bigl( \int_0^{1}
             \frac{D^n(n!)^\sigma}{n!} t^n \chi_\tau(R_n t) e^{im\pi t}\, dt \Bigr)
       \frac{\alpha_j^{(n)}(1)}{CD^n(n!)^\sigma} \notag\\
&\quad + \frac{(-1)^m}{2} \sum_{n=0}^{C(\epsilon_2,K)} C
       \Bigl( \int_0^{1} 
             \frac{D^n(n!)^\sigma}{n!} (t-1)^n \chi_\tau(R_n (t-1)) e^{im\pi t}\, dt \Bigr)
       \frac{\alpha_j^{(n)}(0)}{CD^n(n!)^\sigma}\Big)\Big|\leq \frac{\epsilon_2}{2Kj_{\max}}.\label{eq:def_of_alpha}
\end{align}
Hence
\begin{equation}
|\tilde{\alpha}_j^m-\alpha_j^m|\leq\frac{\epsilon_2}{Kj_{\max}}.
\end{equation}
Then we have
\begin{align}
\left\|\sum_{m \in D^{K}}\sum_{j=0}^{j_{\max}-1}\tilde{\alpha}_j^m \text{Add}_{m}^{4l_{\max}} \otimes M_j-\sum_{m \in D^{K}} \text{Add}_{m}^{4l_{\max}} \otimes \widetilde{H}_{-m}\right\|
&=\left\|\sum_{m \in D^{K}}\sum_{j=0}^{j_{\max}-1}(\tilde{\alpha}_j^m-\alpha_j^m )\text{Add}_{m}^{4l_{\max}} \otimes M_j\right\|\leq 2\epsilon_2.
\end{align}
Then we only need to implement
\begin{equation}
\widetilde{\mathscr{H}}_{\text{eff, pbc}}^{4l_{\max}}:=\sum_{m \in D^{K}}\sum_{j=0}^{j_{\max}-1}\tilde{\alpha}_j^m \text{Add}_{m}^{4l_{\max}} \otimes M_j-\mathscr{H}_{\text{LP}}^{4l_{\max}},
\end{equation}
with the coefficients satisfying
\begin{equation}
\sum_{m \in D^{K}} \sum_{j=0}^{j_{\max}-1}|\tilde{\alpha}_j^m|\leq \sum_{m \in D^{K}} \sum_{j=0}^{j_{\max}-1}|\alpha_j^m|+2\epsilon_2\leq j_{\max}(h_4+2h_3\text{log}(\lceil A/\pi\rceil-1)+2h_2S(\zeta)+1)+1,
\end{equation}
where we use the decay estimates in \cref{sec:Four_decay}.\\
\textbf{Step 3:} For simplicity we consider the case $\text{log}_2\,(8l_{\max})$ is an integer. We prepare a system consisting of two working systems labeled by $a, b$ and five kinds of auxiliary systems labeled by 1$\sim$5. The working system $a$ is a $\lceil\text{log}_2\,d\rceil$-qubit system used to accommodate the matrix $M_j$. The working system $b$ is a $\text{log}_2\,(8l_{\max})$-qubit system used to accommodate the unitary matrix $\text{Add}_m^{4l_{\text{max}}}$ and $\text{Comp}_{\text{working}}$, where $\text{Comp}_{\text{working}}$ represents the working-qubit part of Comp. The auxiliary system 1 is an $n_1$-qubit system prepared for the block-encoding of $M_j$, i.e.
\begin{equation}\label{eq:def_of_O_4.2}
\left\|\bra{0^{n_1}}O_j\ket{0^{n_1}}-M_j\right\|\leq \epsilon_4.
\end{equation}
The auxiliary system 2 is a $\lceil\text{log}_2\,j_{\max}\rceil=:n_2$-qubit system characterized by $\{\ket{j}\}_{j=0}^{2^{n_2}-1 }$, which represent the number of time-independent operator $M_j$. The auxiliary system 3 is a $\text{log}_2\,(8l_{\max})=:n_3$-qubit system characterized by $\{\ket{m}\}_{m\in D^{n_3}}:=\{\ket{-2^{n_3-1}+1},\ket{-2^{n_3-1}+2},...,\ket{2^{n_3-1}-1},\ket{2^{n_3-1}}\}$, which is used to form the linear combination of $\text{Add}_m^{4l_{\text{max}}}$ and accommodate $\text{Comp}_{\text{ancilla}}$, where $\text{Comp}_{\text{ancilla}}$ represents the ancilla-qubit part of Comp. The auxiliary system 4 is a 1-qubit system used to accommodate $\text{Comp}_{\text{single}}$ and $Z_{\text{single}} $, where $\text{Comp}_{\text{single}}$ represents the single-qubit part of Comp. The auxiliary system 5 is a $\lceil\text{log}_2\,(2)\rceil=1$ qubit system to combine the two terms in $\widetilde{\mathscr{H}}_{\text{eff, pbc}}^{4l_{\max}}$. The prepare oracle is defined by
\begin{align}
G_{\text{prep},L}
&= \Bigl((\ket{0}\bra{0})_5 \otimes I_4 \otimes G_{\text{coef},3\sim 2,L}
         \otimes I_{1,b,a}  + (\ket{1}\bra{1})_5 \otimes I_4 \otimes \mathscr{U}_{\text{ini}}^{4l_{\max}}
         \otimes I_{2,1,b,a} \Bigr) \notag \\
&\quad\quad\quad\cdot \Bigl( G_{\text{coef,5}}   \otimes I_{4\sim1,b,a} \Bigr),\label{eq:def_of_prep_4.2L}
\end{align}
\begin{align}
G_{\text{prep},R}
&= \Bigl((\ket{0}\bra{0})_5 \otimes I_4 \otimes G_{\text{coef},3\sim2,R}
         \otimes I_{1,b,a} + (\ket{1}\bra{1})_5 \otimes I_4 \otimes \mathscr{U}_{\text{ini}}^{4l_{\max}}
         \otimes I_{2,1,b,a} \Bigr) \notag\\
&\quad\quad\quad\cdot \Bigl( G_{\text{coef,5}}  \otimes I_{4\sim1,b,a} \Bigr). \label{eq:def_of_prep_4.2R}
\end{align}
where
\begin{equation}\label{eq:cond_of_G_4.2_1}
G_{\text{coef},3\sim 2,L}\ket{0}_3\ket{0}_2=\sum_{m\in D^{n_3}}\sum_{j=0}^{2^{n_2}-1}\tilde{\alpha}_{j,m,L}\ket{m}_3\ket{j}_2\Big/\sqrt{\sum_{m\in D^{n_3}}\sum_{j=0}^{2^{n_2}-1}\left|\tilde{\alpha}_{j,m,L}\right|^2},
\end{equation}
\begin{equation}\label{eq:cond_of_G_4.2_2}
G_{\text{coef},3\sim 2,R}\ket{0}_3\ket{0}_2=\sum_{m\in D^{n_3}}\sum_{j=0}^{2^{n_2}-1}\tilde{\alpha}_{j,m,R}\ket{m}_3\ket{j}_2\Big/\sqrt{\sum_{m\in D^{n_3}}\sum_{j=0}^{2^{n_2}-1}\left|\tilde{\alpha}_{j,m,R}\right|^2},
\end{equation}
satisfy
\begin{equation}\label{eq:cond_of_G_4.2_3}
\sum_{m\in D^{K}}\sum_{j=0}^{j_{\max}-1}\left|\tilde{\alpha}^m_j- \frac{(\sum_{m\in D^{K}}\sum_{j=0}^{j_{\max}-1}|\tilde{\alpha}^m_j|)\cdot(\tilde{\alpha}_{j,m,L})^*\tilde{\alpha}_{j,m,R}}{\sqrt{\sum_{m\in D^{n_3}}\sum_{j=0}^{2^{n_2}-1}\left|\tilde{\alpha}_{j,m,L}\right|^2}\sqrt{\sum_{m\in D^{n_3}}\sum_{j=0}^{2^{n_2}-1}\left|\tilde{\alpha}_{j,m,R}\right|^2}} \right|\leq \epsilon_3,
\end{equation}
and for all $(m,j)\notin D^K\times\{0,1,2,...,j_{\max}-1\}$,
\begin{equation}\label{eq:cond_of_G_4.2_4}
(\tilde{\alpha}_{j,m,L})^*\tilde{\alpha}_{j,m,R}=0.
\end{equation}
In addition, let
\begin{equation}\label{eq:cond_of_G5_4.2_1}
G_{\text{coef,5}}\ket{0}_5=\frac{\sum_{k=0}^1\sqrt{\tilde{D}^{(k)}}\ket{k}_5}{\sqrt{\sum_{k=0}^1\tilde{D}^{(k)}}},
\end{equation}
satisfy
\begin{equation}\label{eq:cond_of_G5_4.2_2}
\sum_{k=0}^1\left|D^{(k)}-\frac{(\sum_{k=0}^1D^{(k)})\tilde{D}^{(k)}}{(\sum_{k=0}^1\tilde{D}^{(k)})}\right|\leq \epsilon_3,
\end{equation}
where 
\begin{equation}\label{eq:def_of_D_4.2}
D^{(0)}=\sum_{m \in D^{K}} \sum_{j=0}^{j_{\max}-1}|\tilde{\alpha}_j^m|,\quad\quad\quad D^{(1)}=4j_{\max}^{\varrho}l_{\max}\omega.
\end{equation}
The select oracle is defined by
\begin{equation}\label{eq:def_of_selec_4.2}
O_{\text{select}}= (\ket{0}\bra{0})_5 \otimes O + (\ket{1}\bra{1})_5 \otimes O_{\text{LP}},
\end{equation}
where
\begin{align}
O&=I_4\otimes\sum_{m\in D^K}(\ket{m}\bra{m})_3\otimes\sum_{j=0}^{j_{\max}-1}(\ket{j}\bra{j})_2\otimes \tikzmarknode{O_{j,1}}{O_{j,1}} \otimes(\text{Add}_m^{4l_{\text{max}}})_b\otimes\tikzmarknode{O_{j,a}}{O_{j,a}} \notag \\
&\quad+I_4\otimes\sum_{(m,j)\notin D^K\times\{0,1,2,...,j_{\max}-1\}}(\ket{m,j}\bra{m,j})_{3,2}\otimes I_{1,b,a},
\end{align}
\begin{tikzpicture}[overlay, remember picture]
  \draw[thick] ([yshift=-3pt]O_{j,1}.south) to[out=-30,in=-150] ([yshift=-3pt]O_{j,a}.south);
\end{tikzpicture}
and $O_{\text{LP}}$ is defined as in \cref{eq:def_of_O_LP}. Then by exactly the same way in step 4 of \cref{subsec:general_case},
\begin{align}
\left\|\bra{0}_{5\sim 1}(G_{\text{prep},L})^{\dagger}O_{\text{select}}G_{\text{prep},R}\ket{0}_{5\sim 1}-\frac{\widetilde{\mathscr{H}}_{\text{eff, pbc}}^{4l_{\max}}}{\sum_{k=0}^1D^{(k)}}\right\|\leq\frac{2\epsilon_3}{\sum_{k=0}^{1}D^{(k)}}
     + \frac{ D^{(0)} \epsilon_4}{\sum_{k=0}^{1}D^{(k)}} .
\end{align}
For any $\varepsilon_1>0$, let
\begin{equation}\label{eq:def_of_epsilon_4.2}
\epsilon_1=\frac{\sum_{k=0}^{1}D^{(k)}}{4}\varepsilon_1,\quad\quad\epsilon_2=\frac{\sum_{k=0}^{1}D^{(k)}}{8}\varepsilon_1,\quad\quad\epsilon_3=\frac{\sum_{k=0}^{1}D^{(k)}}{8}\varepsilon_1,\quad\quad\epsilon_4=\frac{\sum_{k=0}^{1}D^{(k)}}{4D^{(0)}}\varepsilon_1,
\end{equation}
then
\begin{equation}
\left\|\bra{0}_{5\sim 1}(G_{\text{prep},L})^{\dagger}O_{\text{select}}G_{\text{prep},R}\ket{0}_{5\sim 1}-\frac{\mathscr{H}_{\text{eff, pbc}}^{4l_{\max}}}{\sum_{k=0}^1D^{(k)}}\right\|\leq\varepsilon_1.
\end{equation}
To summarize, we get
\begin{theorem}[BE of refined effective Hamiltonian in the control-Hamiltonian case]\label{th:BE_of_refined_effective_Hamiltonian_cont}
Consider $H(s) = \sum_{j=0}^{j_{\max}-1} \alpha_j(s) M_j$ on $[0,1]$ and the Gevrey extension in \cref{sec:Periodic_Extension}, the decay of Fourier coefficients in \cref{sec:Four_decay}, $\varepsilon>0$ and the definition of $l_{\max}$ and $\mathscr{H}_{\text{\normalfont eff, pbc}}^{4l_{\max}}$ in \cref{sec:Error_analysis,sec:Error_analysis_5}. For any $\varepsilon_1>0$, $\tilde{\alpha}_j^m$ in \cref{eq:def_of_alpha}; $D^{(k)}\;(k=0,1)$ defined in \cref{eq:def_of_D_4.2}; $\epsilon_k\;(k=1,2,3,4)$ defined in \cref{eq:def_of_epsilon_4.2}. Given $G_{\text{\normalfont coef},3\sim 2,L}$ and $G_{\text{\normalfont coef},3\sim 2,R}$ satisfying \cref{eq:cond_of_G_4.2_1,eq:cond_of_G_4.2_2,eq:cond_of_G_4.2_3,eq:cond_of_G_4.2_4}; $G_{\text{\normalfont coef},5}$ satisfying \cref{eq:cond_of_G5_4.2_1,eq:cond_of_G5_4.2_2}; $O_j^{(k)}$ satisfying \cref{eq:def_of_O_4.2}; $\text{\normalfont Add}_m^{4l_{\max}}$ in \cref{eq:def_of_Add}  and $\sum_{m\in D^K}\ket{m}\bra{m}\otimes\text{\normalfont Add}_m^{4l_{\max}}$ which both can be implemented by $\mathcal{O}(\text{\normalfont log}\,l_{\max})$ elementary gates; $\text{\normalfont Comp}$ in \cref{eq:def_of_Comp} which can be implemented by $\mathcal{O}(\text{\normalfont log}\,l_{\max})$ elementary gates. We can construct prepare oracles $G_{\text{\normalfont prep},L}$ and $G_{\text{\normalfont prep},R}$ as in \cref{eq:def_of_prep_4.2L,eq:def_of_prep_4.2R} and a select oracle $O_{\text{\normalfont select}}$ as in \cref{eq:def_of_selec_4.2}, with five ancillary-qubit systems together comprising 
\begin{equation}
n_1+\lceil\text{\normalfont log}_2\,j_{\max}\rceil+\text{\normalfont log}_2(8l_{\max})+2
\end{equation}
ancilla qubits, such that
\begin{equation}
\Bigl\| \bra{0}_{5\sim 1} (G_{\text{\normalfont prep},L})^{\dagger} O_{\text{\normalfont select}} G_{\text{\normalfont prep},R} \ket{0}_{5\sim 1}
        - \frac{\mathscr{H}_{\text{\normalfont eff, pbc}}^{4l_{\max}}}{\sum_{k=0}^{1}D^{(k)}} \Bigr\| \leq \varepsilon_1,
\end{equation}
using 1 query of
\begin{equation}\label{eq:def_of_HAM_T}
\text{\normalfont HAM}\text{\normalfont -T}:=\sum_{j=0}^{j_{\max}-1}\ket{j}\bra{j}\otimes O_j.
\end{equation}
\end{theorem}
\qed\\
\textbf{Step 4:} Under \cref{th:BE_of_refined_effective_Hamiltonian_cont} and step 3 in \cref{subsec:general_case}, denote
\begin{equation}
U_1=(G_{\text{prep},L})^{\dagger} O_{\text{select}} G_{\text{prep},R},
\end{equation}
\begin{equation}\label{eq:def_of_U_2_4.2}
U_2=(I_4\otimes \mathscr{U}_{\text{ini}}^{4l_{\max}}\otimes I_{b,a})^{\dagger}(\text{Comp}\otimes I_a)^\dagger (Z_{4}\otimes I_{3,b,a})  (\text{Comp}\otimes I_a)(I_4\otimes \mathscr{U}_{\text{ini}}^{4l_{\max}}\otimes I_{b,a}).
\end{equation}
We can then construct a block encoding of $e^{-i\mathscr{H}_{\text{eff, pbc}}^{4l_{\max}}t}$ and $e^{-i\mathscr{H}_{\text{LP}}^{4l_{\max}}t}$ using QSVT such that
\begin{equation}
\left\|V_1(t)-e^{-i\mathscr{H}_{\text{eff, pbc}}^{4l_{\max}}t}\right\|\leq 2t(\sum_{k=0}^{1}D^{(k)})\varepsilon_1,
\end{equation}
\begin{equation}
\left\|V_2(t)-e^{-i\mathscr{H}_{\text{LP}}^{4l_{\max}}t}\right\|\leq 2tD^{(1)}\varepsilon_2.
\end{equation}
Where $V_1$ is comprised of 
\begin{equation}
6(\sum_{k=0}^{1}D^{(k)})t+9\text{log}\left(\frac{12}{2t(\sum_{k=0}^{1}D^{(k)})\varepsilon_1}\right)
\end{equation}
uses of $U_1$ or its inverse, 3 uses of controlled-$U_1$ or its inverse, with
\begin{equation}
n_1+\lceil\text{log}_2\,j_{\max}\rceil+\text{log}_2(8l_{\max})+4
\end{equation}
ancilla qubits. $V_2$ is comprised of 
\begin{equation}
6D^{(1)}t+9\text{log}\left(\frac{6}{2tD^{(1)}\varepsilon_2}\right)
\end{equation}
uses of $U_2$ or its inverse, 3 uses of controlled-$U_2$ or its inverse, with
\begin{equation}
\text{log}_2(8l_{\max})+3
\end{equation}
ancilla qubits. Let
\begin{align}
\widetilde{\mathscr{U}}_{\text{amp1,pbc}}^{l_{\max}}(t) 
&= (\mathscr{U}_{\text{ini}}^{4l_{\max}})^{\dagger} V_2(t) V_1(t)\mathscr{U}_{\text{ini}}^{l_{\max}} , \\
\widetilde{\mathscr{U}}_{\text{amp2,pbc}}^{l_{\max}}(t)&= -\widetilde{\mathscr{U}}_{\text{amp1,pbc}}^{l_{\max}}(t) \mathscr{R}[\widetilde{\mathscr{U}}_{\text{amp1,pbc}}^{l_{\max}}(t) ]^{\dagger}\mathscr{R}\widetilde{\mathscr{U}}_{\text{amp1,pbc}}^{l_{\max}}(t) .
\end{align}
Then
\begin{equation}
\left\|\widetilde{\mathscr{U}}_{\text{amp2,pbc}}^{l_{\max}}(t)-\mathscr{U}_{\text{amp2,pbc}}^{l_{\max}}(t)\right\|\leq 6t(\sum_{k=0}^{1}D^{(k)})\varepsilon_1+6tD^{(1)}\varepsilon_2.
\end{equation}
Denote
\begin{equation}\label{eq:def_of_E_4.2}
E:= j_{\max}(h_4+2h_3\text{log}(\lceil A/\pi\rceil-1)+2h_2S(\zeta)+1)+1\geq D^{(0)}.
\end{equation}
Pick
\begin{equation}
\varepsilon_1=\frac{\varepsilon}{24t(E+D^{(1)})},\quad\varepsilon_2=\frac{\varepsilon}{24tD^{(1)}}.
\end{equation}
Combining \cref{th:Refined_effective_Hamiltonian}, we finally get that
\begin{equation}
\left\| \widetilde{\mathscr{U}}_{\text{amp2, pbc}}^{l_{\max}}(t) | 0 \rangle | \psi(0) \rangle - |0\rangle |\psi(t)\rangle \right\| \leq  \varepsilon.
\end{equation}\\
\textbf{Step 5:} We now focus on the query complexity of $U_1$ given by
\begin{equation}
18(\sum_{k=0}^{1}D^{(k)})t+27\text{log}\left(\frac{144(E+D^{(1)})}{(\sum_{k=0}^{1}D^{(k)})\varepsilon}\right).
\end{equation}
Although substituting the definition of $D^{(k)}$ in \cref{eq:def_of_D_4.2} and $E$ in \cref{eq:def_of_E_4.2} into the formula above yields an upper bound on the query complexity, we here derive the asymptotic dependence. Note that
\begin{equation}
D^{(0)}\leq E=\mathcal{O}\left(j_{\max}\alpha \text{log}\left(\frac{1}{\tau-1}\right)\right)=\mathcal{O}\left(j_{\max}\alpha \text{log}\text{log}\left(\alpha t+\text{log}\left(\frac{1}{\varepsilon}\right)\right)\right),
\end{equation}
and
\begin{equation}
D^{(1)}=\mathcal{O}\left(j_{\max}^{\varrho}\omega\left(\alpha t+  \log\left(\frac{1}{\varepsilon}\right)\right)^{\sigma}\text{log}\left(\alpha t+ \log\left(\frac{1}{\varepsilon}\right)\right)\left(\text{log}\text{log}\left(\alpha t+ \log\left(\frac{1}{\varepsilon}\right)\right)\right)^{\sigma}\right).
\end{equation}
Hence
\begin{align}
&18(\sum_{k=0}^{1}D^{(k)})t+27\text{log}\left(\frac{144(E+D^{(1)})}{(\sum_{k=0}^{1}D^{(k)})\varepsilon}\right)\notag\\
&\quad\quad=\mathcal{O}\left((\sum_{k=0}^{1}D^{(k)})t+\text{log}\left(\frac{144(E+D^{(1)})}{(\sum_{k=0}^{1}D^{(k)})\varepsilon}\right)\right)\notag \\
&\quad\quad=\mathcal{O}\left(j_{\max}^{\varrho}\left(\alpha t+  \log\left(\frac{1}{\varepsilon}\right)\right)^{\sigma}\text{log}\left(\alpha t+ \log\left(\frac{1}{\varepsilon}\right)\right)\left(\text{log}\text{log}\left(\alpha t+ \log\left(\frac{1}{\varepsilon}\right)\right)\right)^{\sigma}\right)\notag \\
&\quad\quad=\widetilde{\mathcal{O}}\left(j_{\max}^{\varrho}\left(\alpha t+  \log\left(\frac{1}{\varepsilon}\right)\right)^{\sigma}\right).
\end{align}
By exactly the same derivation, we find that the query complexity of $U_2$ is also 
$\widetilde{\mathcal{O}}\left(j_{\max}^{\varrho}\left(\alpha t+  \log\left(\frac{1}{\varepsilon}\right)\right)^{\sigma}\right)$.

As in the case of general Hamiltonians in \cref{subsec:general_case}, the additional gate complexity for the controlled Hamiltonian case also consists of three contributions: the prepare oracles for the linear combination of unitary matrices in \cref{eq:def_of_prep_4.2L,eq:def_of_prep_4.2R}, the adder gates and compare gates used in the select oracle in \cref{eq:def_of_selec_4.2}, and the compare gate $\mathrm{Comp}$ and the Hadamard gate $U_{\text{ini}}^{4l_{\max}}$ in $U_2$ in \cref{eq:def_of_U_2_4.2}. The first part contributes 
$\widetilde{\mathcal{O}}\left(j_{\max}^{\varrho+1}\left(\alpha t + \log(1/\varepsilon)\right)^{2\sigma}\right)$, 
while the second and third parts each contribute 
$\widetilde{\mathcal{O}}\left(j_{\max}^{\varrho}\left(\alpha t + \log(1/\varepsilon)\right)^{\sigma}\right)$. The additional gate complexity incurred by the QSVT does not change the asymptotic complexity of the total additional gates as in \cref{subsec:general_case}. Therefore, the additional gate complexity in the controlled Hamiltonian case is 
$\widetilde{\mathcal{O}}\left(j_{\max}^{\varrho+1}\left(\alpha t + \log(1/\varepsilon)\right)^{2\sigma}\right)$. From the above discussion, we obtain the following theorem.

\begin{theorem}[Query and gate complexity in the control-Hamiltonian case]\label{th:query_gate_complex_cont}
Let $1 \leq \sigma < 2$ and $H(s) = \sum_{j=0}^{j_{\max}-1} \alpha_j(s) M_j$ with coefficients $\{\alpha_j(s)\}_{j=1}^{j_{\max}}\subseteq G^{\sigma}([0,1])$ and $|\alpha_j(s)| \leq \alpha$ for each $j$. Define $\widetilde{\alpha}_j(t) = \alpha_j(t/T)$ for each $j$.
Suppose that we are given access to the $\text{\normalfont HAM}\text{\normalfont -T}$ oracles defined by $\text{\normalfont HAM}\text{\normalfont -T}=\sum_{j=0}^{j_{\max}-1}\ket{j}\bra{j}\otimes O_j$, where $O_j$ is the block encoding of $M_j$ for each $j$. Then there exists a quantum algorithm for solving~\cref{eqn:ham_sim} up to time $t\in[0,T]$ with error at most $\varepsilon$, using $\widetilde{\mathcal{O}}\left(j_{\max}^{\varrho}\left(\alpha t+  \log\left(1/\varepsilon\right)\right)^{\sigma}\right)$ $\text{\normalfont HAM}\text{\normalfont -T}$ queries, and can be implemented using $\widetilde{\mathcal{O}}\left(j_{\max}^{\varrho+1}\left(\alpha t + \log(1/\varepsilon)\right)^{2\sigma}\right)$ additional elementary gates, where $\varrho = \sigma + \frac{1}{\text{\normalfont log}\bigl(\alpha t +e+ \text{\normalfont log}(1/\varepsilon)\bigr)}$. 
\end{theorem}

\section{Application to slow linear differential equations}\label{sec:linear_odes}

We now apply our Hamiltonian simulation results developed above to general non-unitary linear dynamics. 
Consider the homogeneous linear differential equation
\begin{equation}\label{eq:linear_ode}
    \frac{d}{dt}u(t)
    =-\widetilde A(t)u(t),
    \qquad
    \widetilde A(t)=A(t/T),
    \qquad
    u(0)=u_0,
\end{equation}
where $A(s)$ is analytic but does not need to be Hermitian. 
Its solution is
\begin{equation}\label{eq:linear_ode_propagator}
    u(T)=\mathcal V_A(T) u_0,
    \qquad
    \mathcal V_A(T):=
    \mathcal T\exp\left(-\int_0^T \widetilde A(t)\,dt\right).
\end{equation}
We use the Cartesian decomposition
\begin{equation}\label{eq:cartesian_decomposition}
    A(s)=L(s)+iH(s),
    \qquad
    L(s)=\frac{A(s)+A(s)^\dagger}{2},
    \qquad
    H(s)=\frac{A(s)-A(s)^\dagger}{2i},
\end{equation}
and assume $L(s)\succeq0$ for all $s\in[0,1]$. 
This is precisely the non-positive-logarithmic-norm condition on the generator $-A(s)$ and implies that $\|u(t)\|$ is nonincreasing. 

To solve~\cref{eq:linear_ode_propagator}, existing works~\cite{AnLiuLin2023,AnChildsLin2023,LowSomma2025} have proposed the LCHS approach, which represents the time evolution operator of~\cref{eq:cartesian_decomposition} as a linear combination of unitary operators and implements it through LCU of Hamiltonian simulation algorithms. 
Specifically, the optimal LCHS construction, proposed in~\cite{LowSomma2025}, approximates the non-unitary propagator by a weighted integral of unitary propagators,
\begin{equation}\label{eq:lchs_representation}
    \mathcal V_A(T)
    \approx
    \mathcal O_R(T):=
    \frac{1}{\sqrt{2\pi}}
    \int_{-R}^{R}\widehat f(k)U_k(T)\,dk,
    \qquad
    U_k(T):=
    \mathcal T\exp\left(
        -i\int_0^T\bigl(k\widetilde L(t)+\widetilde H(t)\bigr)\,dt
    \right),
\end{equation}
where $\widetilde L(t)=L(t/T)$ and $\widetilde H(t)=H(t/T)$, and $\widehat f(k)$ is a scalar-valued kernel function with closed-form expression. 
Notice that the Hamiltonian $k\widetilde{L}(t)+\widetilde{H}(t)$ remains slow and analytic. 
This observation allows us to insert our slow analytic Hamiltonian simulation algorithm into the LCHS construction.

\begin{corollary}[Slow Gevrey linear differential equations]\label{th:slow_ode_complexity}
    Consider the problem of linear differential equation~\cref{eq:linear_ode}. 
    Assume that $\widetilde{A}(t) = A(t/T)$ for an analytic matrix-valued function $A(s)$ with $\|A(s)\| \leq \alpha$, and $A(s)$ has the Cartesian decomposition $A(s) = L(s) + iH(s)$ with $L(s)\succeq 0$. 
    Suppose further that either the derivative-access oracles appearing in \cref{th:query_gate_complex_gene} or the control Hamiltonian oracles in~\cref{th:query_gate_complex_cont} are available for $A(s)$, and we are given access to the state preparation oracle of the initial condition $\ket{u_0}$. 
    Then, for any $0<\varepsilon<1$, a state $|\widetilde{u}(T)\rangle$ satisfying
    \begin{equation}
        \left\|
        |\widetilde{u}(T)\rangle
        -\frac{u(T)}{\|u(T)\|}
        \right\|\leq\varepsilon
    \end{equation}
    can be prepared with constant success probability using
    \begin{equation}\label{eq:ode_state_query_complexity}
        \widetilde{\mathcal O}\!\left(
    \frac{\|u(0)\|}{\|u(T)\|}\alpha T \log\left(\frac{1}{\varepsilon}\right)
    \right)
    \end{equation}
    queries to the matrix oracles and $\widetilde{\mathcal O}\!\left(
    \|u(0)\|/\|u(T)\|
    \right)$ queries to the state preparation of $\ket{u_0}$. 
\end{corollary}

\begin{proof}
We first construct a block-encoding of the time evolution operator $\mathcal{V}_A(T)$ up to a target operator error $\delta$. 
We choose the LCHS truncation parameter $R$ and quadrature to discretize the integral $\mathcal{O}_R(T)$, and uses LCU to implement $\mathcal{O}_R(T)$. 
By the optimal-scaling LCHS construction~\cite{LowSomma2025}, 
the kernel in Ref.~\cite{LowSomma2025} can be chosen so that
\begin{equation}\label{eq:lchs_parameters}
    R=\mathcal O\!\left(\log\left(\frac{1}{\delta}\right)\right),
    \qquad
    \frac{1}{\sqrt{2\pi}}\int_{-R}^{R}|\widehat f(k)|\,dk=\mathcal O(1),
\end{equation}
and a uniform quadrature converges exponentially. 
Thus the LCU normalization is constant, and the overall query complexity of constructing a block-encoding of $\mathcal{V}_A(T)$ is to cost of simulating $\{U_k(T):|k|\leq R\}$. 
Since the LCHS index register can coherently control the LCU constructions of $kL+H$, so the query cost equals that of the most expensive node $|k|=R$, rather than the number of quadrature nodes. 
Applying \cref{th:query_gate_complex_gene} with simulation error $\mathcal O(\delta)$ therefore yields the cost to be
\begin{equation}
    \widetilde{\mathcal O}\!\left(
    \alpha R T+\log\left(\frac{1}{\delta}\right)
    \right) = \widetilde{\mathcal O}\!\left(
    \alpha T \log\left(\frac{1}{\delta}\right)
    \right)
\end{equation}
queries to the matrix oracles, with $\mathcal{O}(1)$ block-encoding normalization factor. 

To prepare an approximation of the normalized final solution $\ket{u(T)}$, we can simply apply the block-encoding of $\mathcal{V}_A(T)$ to the initial state $\ket{u_0}$ and post select on the ancilla qubits to extract the good subspace. 
The naive post-selection success probability is $\Omega( (\|u(T)\|/\|u(0)\|)^2 )$ and can be boosted to $\Omega(1)$ with $\mathcal{O}(\|u(0)\|/\|u(T)\|)$ rounds of amplitude amplification. 
Therefore, the overall query complexity becomes $\widetilde{\mathcal O}\!\left(
    \frac{\|u(0)\|}{\|u(T)\|}\alpha T \log\left(\frac{1}{\delta}\right)
    \right)$
queries to the matrix oracles and $\widetilde{\mathcal O}\!\left(
    \frac{\|u(0)\|}{\|u(T)\|}
    \right)$ queries to the initial state preparation. 
To make sure the output quantum state is an $\varepsilon$ approximation, it suffices to choose $\delta = \Theta(\epsilon \|u(T)\|/\|u(0)\|)$, which yields the claimed complexity. 
\end{proof}

The above discussions mainly focus on the homogeneous differential equations. 
The inhomogeneous equation $\dot{u}(t)=-\widetilde A(t) u(t)+b(t)$ can be treated based on the Duhamel principle, by applying the same propagator block encoding to the extra term in the Duhamel principle and then using an outer LCU over the time, as in~\cite{LowSomma2025}.

\printbibliography

\appendix
\section{Proofs for \texorpdfstring{\cref{sec:Error_analysis_5}}{Section 5}}
\label{app:main}

This appendix collects the complete proofs of~\cref{lem:translation_symmetry}, ~\cref{th:Amplification_by_symmetry},~\cref{th:OAA} and~\cref{th:Refined_effective_Hamiltonian} in~\cref{sec:Error_analysis_5}. All of these proofs are essentially generalizations of the techniques in \cite{mizuta2023optimal} to the Gevrey-class Hamiltonian setting.

\subsection{Proof of \texorpdfstring{\cref{lem:translation_symmetry}}{Lemma 5.1}: approximate translation symmetry}

\begin{proof}[Proof of~\cref{lem:translation_symmetry}]
We consider a perturbation $\tilde{\mathscr{H}}_b(t)$ designated by
\begin{equation}
\tilde{\mathscr{H}}_b(t) = \sum_{(l,m) \in \partial \widetilde{F}^{4l_{\max}}} |l\rangle \langle l \oplus m| \otimes e^{8il_{\max} \omega t} \widetilde{H}_{m} + \left(|l\rangle \langle l \oplus m| \otimes e^{8il_{\max} \omega t} \widetilde{H}_{m}\right)^{\dagger},
\end{equation}
with $\partial \tilde{F}^{4l_{\max}} = \{(l,m)\mid l \in D^{4l_{\max}},\; 8l_{\max} - l + 1 \leq m \leq 8l_{\max} - 1\}$, and $l\oplus m \in D^{4l_{\max}}$ defined as modulo $8l_{\max}$. This Hamiltonian indicates hopping terms that go across the boundaries $|4l_{\max}\rangle$ and $|-4l_{\max}+1\rangle$. Let $\tilde{\mathscr{U}}_{\text{pert}}(t)$ denote a time evolution operator under $\mathscr{H}_{\text{eff}}^{4l_{\max}}+\tilde{\mathscr{H}}_b(t)$. Then, due to the exact translation symmetry in the interaction picture, the transition amplitude $\langle l | \tilde{\mathscr{U}}_{\text{pert}}(t) | l'\rangle$ satisfies
\begin{equation}
\langle l|\tilde{\mathscr{U}}_{\text{pert}}(t)|l'\rangle = e^{i l'\omega t} \langle l \ominus l'|\tilde{\mathscr{U}}_{\text{pert}}(t)|0\rangle.
\end{equation}
Hence we have
\begin{align}
&\left\| \langle l | e^{-i \mathscr{H}_{\text{eff}}^{4l_{\max}}t} | l' \rangle - e^{i l' \omega t} \langle l \ominus l' | e^{-i \mathscr{H}_{\text{eff}}^{4l_{\max}}t} | 0 \rangle \right\| \notag \\
&\quad\leq \left\| \langle l| \bigl( \tilde{\mathscr{U}}_{\text{pert}}(t) - e^{-i\mathscr{H}_{\text{eff}}^{4l_{\max}}t} \bigr) |l'\rangle \right\| 
      + \left\| \langle l| \tilde{\mathscr{U}}_{\text{pert}}(t)|l'\rangle 
         - e^{i l' \omega t} \langle l \ominus l' | e^{-i \mathscr{H}_{\text{eff}}^{4l_{\max}}t} | 0 \rangle \right\| \notag \\
&\quad= \left\| \langle l| \bigl( \tilde{\mathscr{U}}_{\text{pert}}(t) - e^{-i\mathscr{H}_{\text{eff}}^{4l_{\max}}t} \bigr) |l'\rangle \right\| 
      + \left\| e^{i l'\omega t} \langle l \ominus l'|\tilde{\mathscr{U}}_{\text{pert}}(t)|0\rangle 
         - e^{i l' \omega t} \langle l \ominus l' | e^{-i \mathscr{H}_{\text{eff}}^{4l_{\max}}t} | 0 \rangle \right\| \notag \\ 
&\quad =\left\| \langle l| \bigl( \tilde{\mathscr{U}}_{\text{pert}}(t) - e^{-i\mathscr{H}_{\text{eff}}^{4l_{\max}}t} \bigr) |l'\rangle \right\|+\left\|\langle l \ominus l'|\left(\tilde{\mathscr{U}}_{\text{pert}}(t)|0\rangle-    e^{-i \mathscr{H}_{\text{eff}}^{4l_{\max}}t}\right) | 0 \rangle\right\|.
\end{align}
The difference of transition amplitudes between $\mathscr{H}_{\text{eff}}^{4l_{\max}}+\tilde{\mathscr{H}}_b(t)$ and $\mathscr{H}_{\text{eff}}^{4l_{\max}}$ is bounded in a similar way to that in \cref{th:Floquet-Hilbert_space_truncation}. The difference survives only when the trajectory $|l'\rangle \rightarrow |l_1\rangle \rightarrow \cdots \rightarrow |l_{n-1}\rangle \rightarrow |l\rangle$ pass through the boundaries $|4l_{\max}\rangle$ and $|-4l_{\max}+1\rangle$ via $\tilde{\mathscr{H}}_b(t)$, and then its length $\sum_{i=1}^n |m_i|$ with $m_i = l_i - l_{i-1}$ should be equal to or larger than $(4l_{\max} - |l|) + (4l_{\max} - |l'|)$. Then we can obtain its upper bound in a similar way to \cref{th:Floquet-Hilbert_space_truncation}:
\begin{equation}
\left\| \langle l| \bigl( \tilde{\mathscr{U}}_{\text{pert}}(t) - e^{-i\mathscr{H}_{\text{eff}}^{4l_{\max}}t} \bigr) |l'\rangle \right\|\leq \sum_{n=0}^{\infty} \frac{(2 t)^n}{n!} S_n(8l_{\max} - |l| - |l'|) \leq e^{2\beta  t-(8l_{\max}-|l|-|l'|)^{\frac{1}{\varrho}}/2\zeta }.
\end{equation}
In the same way we have 
\begin{equation}
\left\|\langle l \ominus l'|\left(\tilde{\mathscr{U}}_{\text{pert}}(t)|0\rangle-    e^{-i \mathscr{H}_{\text{eff}}^{4l_{\max}}t}\right) | 0 \rangle\right\| \leq \sum_{n=0}^{\infty} \frac{(2 t)^n}{n!} S_n(8l_{\max} - |l| - |l'|) \leq e^{2\beta  t-(8l_{\max}-|l|-|l'|)^{\frac{1}{\varrho}}/2\zeta }.
\end{equation}
Combining the two estimates yields the desired bound.
\end{proof}

\subsection{Proof of \texorpdfstring{\cref{th:Amplification_by_symmetry}}{Theorem 5.2}: amplification by symmetry}

\begin{proof}[Proof of~\cref{th:Amplification_by_symmetry}]
Note that 
\begin{align}
\langle 0| \mathscr{U}_{\text{amp1}}^{l_{\max}}(t) | 0 \rangle | \psi(0) \rangle
&= \frac{1}{4l_{\max}} \sum_{l' \in D^{l_{\max}}} \sum_{l \in D^{4l_{\max}}} e^{-i l \omega t} \langle l | e^{-i \mathscr{H}_{\text{eff}}^{4l_{\max}} t} | l' \rangle | \psi(0) \rangle.
\end{align}
We separate the summation over $l \in D^{4l_{\max}}$ in the above formula by
\begin{equation}
\sum_{l \in D^{4l_{\max}}} = \sum_{l;\; l-l' \in D^{3l_{\max}}} + \sum_{l \in D^{4l_{\max}};\; l-l' \notin D^{3l_{\max}}}.
\end{equation}
Let us focus on the first summation. For each $l'\in D^{l_{\max}}$, the summation can be approximated as
\begin{align}
& \left\| \sum_{l;\; l-l' \in D^{3l_{\max}}} e^{-i l \omega t} \langle l | e^{-i \mathscr{H}_{\text{eff}}^{4l_{\max}} t} | l' \rangle | \psi(0) \rangle - | \psi(t) \rangle \right\| \notag \\
&\quad \leq \left\| \sum_{l;\; l-l' \in D^{3l_{\max}}} e^{-i l \omega t} \langle l | e^{-i \mathscr{H}_{\text{eff}}^{4l_{\max}} t} | l' \rangle | \psi(0) \rangle 
      - \sum_{l;\; l-l' \in D^{3l_{\max}}} e^{-i (l- l') \omega t} \langle l \ominus l'| e^{-i \mathscr{H}_{\text{eff}}^{4l_{\max}} t} | 0 \rangle | \psi(0) \rangle \right\| \notag \\
&\qquad + \left\| \sum_{l;\; l-l' \in D^{3l_{\max}}} e^{-i (l- l') \omega t} \langle l \ominus l'| e^{-i \mathscr{H}_{\text{eff}}^{4l_{\max}} t} | 0 \rangle | \psi(0) \rangle 
      - | \psi(t) \rangle \right\| \notag \\
&\quad = \left\| \sum_{l;\; l-l' \in D^{3l_{\max}}} e^{-i l \omega t} \langle l | e^{-i \mathscr{H}_{\text{eff}}^{4l_{\max}} t} | l' \rangle | \psi(0) \rangle 
      - \sum_{l;\; l-l' \in D^{3l_{\max}}} e^{-i (l- l') \omega t} \langle l \ominus l'| e^{-i \mathscr{H}_{\text{eff}}^{4l_{\max}} t} | 0 \rangle | \psi(0) \rangle \right\| \notag \\
&\qquad + \left\| \sum_{l \in D^{3l_{\max}}} e^{-i l \omega t} \langle l | e^{-i \mathscr{H}_{\text{eff}}^{4l_{\max}} t} | 0 \rangle | \psi(0) \rangle 
      - | \psi(t) \rangle \right\| \notag \\
&\quad \leq \sum_{l;\; l-l' \in D^{3l_{\max}}} \left\| \langle l | e^{-i \mathscr{H}_{\text{eff}}^{4l_{\max}} t} | l' \rangle | \psi(0) \rangle 
      - e^{i  l'\omega t} \langle l \ominus l'| e^{-i \mathscr{H}_{\text{eff}}^{4l_{\max}} t} | 0 \rangle | \psi(0) \rangle \right\| \notag \\
&\qquad + \left\| \sum_{l \in D^{3l_{\max}}} e^{-i l \omega t} \langle l | e^{-i \mathscr{H}_{\text{eff}}^{4l_{\max}} t} | 0 \rangle | \psi(0) \rangle 
      - | \psi(t) \rangle \right\| \notag \\
&\quad \leq \sum_{l;\; l-l' \in D^{3l_{\max}}} 2e^{2\beta  t - \frac{(8l_{\max}-|l|-|l'|)^{\frac{1}{\varrho}}}{2\zeta} } 
      + \left\| \sum_{l \in D^{3l_{\max}}} e^{-i l \omega t} \langle l | e^{-i \mathscr{H}_{\text{eff}}^{4l_{\max}} t} | 0 \rangle | \psi(0) \rangle 
      - | \psi(t) \rangle \right\| \notag \\
&\quad \leq 4\left(\sum_{m=0}^{\infty} e^{ - \frac{m^{\frac{1}{\varrho}}}{4\zeta} } \right)e^{2\beta  t - \frac{(3l_{\max})^{\frac{1}{\varrho}}}{4\zeta}}
      + \left\| \sum_{l \in D^{3l_{\max}}} e^{-i l \omega t} \langle l | e^{-i \mathscr{H}_{\text{eff}}^{4l_{\max}} t} | 0 \rangle | \psi(0) \rangle 
      - | \psi(t) \rangle \right\| \notag \\
&\quad \leq 8\left(\sum_{m=0}^{\infty} e^{ - \frac{m^{\frac{1}{\varrho}}}{4\zeta} } \right)e^{2\beta  t - \frac{(3l_{\max})^{\frac{1}{\varrho}}}{4\zeta}}\leq \frac{8\varepsilon^{3^{1/\varrho}}}{4^{ 3^{1/\varrho}}(S(4\zeta))^{3^{1/\varrho}-1}}.
\end{align}
Here the third inequality invokes Lemma~\ref{lem:translation_symmetry}, the penultimate one follows the same argument as Theorem~\ref{th:Floquet-Hilbert_space_truncation}, and the last one relies on the choice of~$l_{\max}$. 

We next compute the second summation, which is taken over $l \in D^{4l_{\max}}$ satisfying $l - l' \notin D^{3l_{\max}}$. The Lieb-Robinson bound, dictated by \cref{th:transition_rate}, immediately concludes its upper bound by
\begin{align}
& \left\| \sum_{l;\; l-l' \notin D^{3l_{\max}}} e^{-i l \omega t} \langle l | e^{-i \mathscr{H}_{\text{eff}}^{4l_{\max}} t} | l' \rangle | \psi(0) \rangle \right\| \notag \\
&\quad\leq \sum_{l;\; l-l' \notin D^{3l_{\max}}} e^{2\beta  t - |l-l'|^{\frac{1}{\varrho}}/2\zeta } \notag \\
&\quad\leq 2\left(\sum_{m=0}^{\infty} e^{ - \frac{m^{\frac{1}{\varrho}}}{4\zeta} } \right)  e^{2\beta  t - (3l_{\max})^{\frac{1}{\varrho}}/4\zeta } \leq \frac{2\varepsilon^{3^{1/\varrho}}}{4^{ 3^{1/\varrho}}(S(4\zeta))^{3^{1/\varrho}-1}}.
\end{align}
By combining these results we complete the proof.
\end{proof}

\subsection{Proof of \texorpdfstring{\cref{th:OAA}}{Theorem 5.3}: oblivious amplitude amplification}

\begin{proof}[Proof of~\cref{th:OAA}]
The first one $\mathscr{R}$ reverses the sign of $|l\rangle$ for $l \neq 0$, and it is implemented with $\mathcal{O}(\log l_{\max})$ gates. The second one $\mathscr{U}_{\text{amp2}}^{l_{\max}}(t)$ plays a role in enhancing the amplitude of $|\psi(t)\rangle$ up to $1 - \mathcal{O}(\varepsilon)$. Its action on any initial state $|0\rangle |\psi(0)\rangle$ is computed as follows:
\begin{align}
\mathscr{U}_{\text{amp2}}^{l_{\max}}(t) |0\rangle |\psi(0)\rangle
&= \mathscr{U}_{\text{amp1}}^{l_{\max}}(t) \mathscr{R} (\mathscr{U}_{\text{amp1}}^{l_{\max}}(t))^{\dagger}
   \left(-\frac{1}{2} |0\rangle |\psi(t)\rangle + |\Psi^{\perp}\rangle\right)- \mathscr{U}_{\text{amp1}}^{l_{\max}}(t) \mathscr{R}(\mathscr{U}_{\text{amp1}}^{l_{\max}}(t))^{\dagger}|0\rangle\Delta. \notag \\
&= \mathscr{U}_{\text{amp1}}^{l_{\max}}(t) \mathscr{R} \left\{ |0\rangle |\psi(0)\rangle
   - (\mathscr{U}_{\text{amp1}}^{l_{\max}}(t))^{\dagger} |0\rangle |\psi(t)\rangle \right\} -2 \mathscr{U}_{\text{amp1}}^{l_{\max}}(t) \mathscr{R}(\mathscr{U}_{\text{amp1}}^{l_{\max}}(t))^{\dagger}|0\rangle\Delta.
\end{align}
In the second equality, we use the relation obtained by applying $(\mathscr{U}_{\text{amp1}}^{l_{\max}}(t))^{\dagger}$ to 
\begin{equation}
\mathscr{U}_{\text{amp1}}^{l_{\max}}(t)\,|0\rangle\,|\psi(0)\rangle 
=\frac{1}{2}\,|0\rangle\,|\psi(t)\rangle+|0\rangle\Delta+|\Psi^{\perp}\rangle.
\end{equation}
Next, we evaluate
\begin{align}
\mathscr{R} (\mathscr{U}_{\text{amp1}}^{l_{\max}}(t))^{\dagger} |0\rangle|\psi(t)\rangle
&= 2 |0\rangle \left( \langle 0| (\mathscr{U}_{\text{amp1}}^{l_{\max}}(t))^{\dagger} |0\rangle \right) |\psi(t)\rangle
- (\mathscr{U}_{\text{amp1}}^{l_{\max}}(t))^{\dagger} |0\rangle|\psi(t)\rangle \notag \\
&= |0\rangle U^{\dagger}(t)|\psi(t)\rangle+ 2 |0\rangle \left( \langle 0| (\mathscr{U}_{\text{amp1}}^{l_{\max}}(t))^{\dagger} |0\rangle-\frac{1}{2}U^{\dagger}(t) \right) |\psi(t)\rangle
- (\mathscr{U}_{\text{amp1}}^{l_{\max}}(t))^{\dagger} |0\rangle|\psi(t)\rangle \notag \\
&= |0\rangle |\psi(0)\rangle
- (\mathscr{U}_{\text{amp1}}^{l_{\max}}(t))^{\dagger} |0\rangle|\psi(t)\rangle+2 |0\rangle \left( \langle 0| (\mathscr{U}_{\text{amp1}}^{l_{\max}}(t))^{\dagger} |0\rangle-\frac{1}{2}U^{\dagger}(t) \right) |\psi(t)\rangle.
\end{align}
Hence we have
\begin{align}
&\mathscr{U}_{\text{amp2}}^{l_{\max}}(t) |0\rangle |\psi(0)\rangle \notag \\
&\quad= \mathscr{U}_{\text{amp1}}^{l_{\max}}(t) |0\rangle |\psi(0)\rangle
   - \mathscr{U}_{\text{amp1}}^{l_{\max}}(t) \mathscr{R} (\mathscr{U}_{\text{amp1}}^{l_{\max}}(t))^{\dagger} |0\rangle |\psi(t)\rangle
   - 2 \mathscr{U}_{\text{amp1}}^{l_{\max}}(t) \mathscr{R} (\mathscr{U}_{\text{amp1}}^{l_{\max}}(t))^{\dagger} |0\rangle \Delta \notag \\
&\quad= \mathscr{U}_{\text{amp1}}^{l_{\max}}(t) |0\rangle |\psi(0)\rangle
   - \mathscr{U}_{\text{amp1}}^{l_{\max}}(t) |0\rangle |\psi(0)\rangle
   + |0\rangle |\psi(t)\rangle \notag \\
&\qquad - 2 \mathscr{U}_{\text{amp1}}^{l_{\max}}(t) |0\rangle \left( \langle 0| (\mathscr{U}_{\text{amp1}}^{l_{\max}}(t))^{\dagger} |0\rangle
   - \frac{1}{2} U^{\dagger}(t) \right) |\psi(t)\rangle
   - 2 \mathscr{U}_{\text{amp1}}^{l_{\max}}(t) \mathscr{R} (\mathscr{U}_{\text{amp1}}^{l_{\max}}(t))^{\dagger} |0\rangle \Delta \notag \\
&\quad= |0\rangle |\psi(t)\rangle
   - 2 \mathscr{U}_{\text{amp1}}^{l_{\max}}(t) |0\rangle \left( \langle 0| (\mathscr{U}_{\text{amp1}}^{l_{\max}}(t))^{\dagger} |0\rangle
   - \frac{1}{2} U^{\dagger}(t) \right) |\psi(t)\rangle
   - 2 \mathscr{U}_{\text{amp1}}^{l_{\max}}(t) \mathscr{R} (\mathscr{U}_{\text{amp1}}^{l_{\max}}(t))^{\dagger} |0\rangle \Delta.
\end{align}
\cref{th:Amplification_by_symmetry} essentially indicates that the time evolution operator $U(t)$ is approximated as
\begin{equation}
\left\| \langle 0| \mathscr{U}_{\text{amp1}}^{l_{\max}}(t)|0\rangle - \frac{1}{2} U(t) \right\| \leq \frac{10\varepsilon^{3^{1/\varrho}}}{4^{ 3^{1/\varrho}}(S(4\zeta))^{3^{1/\varrho}-1}},
\end{equation}
and hence the relation
\begin{equation}
\left\| \langle 0| [\mathscr{U}_{\text{amp1}}^{l_{\max}}(t)]^{\dagger}|0\rangle - \frac{1}{2} U(t)^{\dagger} \right\| \leq \frac{10\varepsilon^{3^{1/\varrho}}}{4^{ 3^{1/\varrho}}(S(4\zeta))^{3^{1/\varrho}-1}}
\end{equation}
is also satisfied. This means that  
\begin{equation}
\left\| 2 \mathscr{U}_{\text{amp1}}^{l_{\max}}(t) |0\rangle \left( \langle 0| (\mathscr{U}_{\text{amp1}}^{l_{\max}}(t))^{\dagger} |0\rangle-\frac{1}{2}U^{\dagger}(t) \right) |\psi(t)\rangle\right\|< \frac{20\varepsilon^{3^{1/\varrho}}}{4^{ 3^{1/\varrho}}(S(4\zeta))^{3^{1/\varrho}-1}},
\end{equation}
and with 
\begin{equation}
\left\|\mathscr{U}_{\text{amp1}}^{l_{\max}}(t) \mathscr{R} (\mathscr{U}_{\text{amp1}}^{l_{\max}}(t))^{\dagger}|0\rangle\Delta\right\|=\left\||0\rangle\Delta\right\|\leq\frac{10\varepsilon^{3^{1/\varrho}}}{4^{ 3^{1/\varrho}}(S(4\zeta))^{3^{1/\varrho}-1}},
\end{equation}
we get
\begin{equation}
\left\|\mathscr{U}_{\text{amp2}}^{l_{\max}}(t) |0\rangle |\psi(0)\rangle - |0\rangle |\psi(t)\rangle\right\| \leq \frac{40\varepsilon^{3^{1/\varrho}}}{4^{ 3^{1/\varrho}}(S(4\zeta))^{3^{1/\varrho}-1}}
\end{equation}
for an arbitrary initial state $|\psi(0)\rangle$. This result indicates that the operation $\mathscr{U}_{\text{amp2}}^{l_{\max}}(t)$ generates the time-evolved state $|\psi(t)\rangle$ with the amplitude $1-\mathcal{O}(\varepsilon)$, which completes the proof. 
\end{proof}

\subsection{Proof of \texorpdfstring{\cref{th:Refined_effective_Hamiltonian}}{Theorem 5.4}: refined effective Hamiltonian}
\label{app:refined_proof}

\begin{proof}[Proof of~\cref{th:Refined_effective_Hamiltonian}]
First, we evaluate the difference of the transition rates between $\mathscr{H}_{\text{eff, pbc}}^{4l_{\max}}$ and $\mathscr{H}_{\text{eff}}^{4l_{\max}}$. We replace the perturbation $\tilde{\mathscr{H}}_b(t)$ by the boundary term $\tilde{\mathscr{H}}_b$ in the proof of \cref{lem:translation_symmetry}. We obtain
\begin{equation}
\left\| \langle l | e^{-i\mathscr{H}_{\text{eff, pbc}}^{4l_{\max}}t} | l' \rangle - \langle l | e^{-i\mathscr{H}_{\text{eff}}^{4l_{\max}}t} | l' \rangle \right\| \leq  e^{2\beta  t-(8l_{\max}-|l|-|l'|)^{\frac{1}{\varrho}}/2\zeta }.
\end{equation}
Once we obtain this bound, we can evaluate the deviation from $\mathscr{U}_{\text{amp1}}^{l_{\max}}(t)$ as
\begin{align}
\left\| \langle 0 | \bigl( \mathscr{U}_{\text{amp1, pbc}}^{l_{\max}}(t) - \mathscr{U}_{\text{amp1}}^{l_{\max}}(t) \bigr) | 0 \rangle | \psi(0) \rangle \right\|
&\leq\frac{1}{4l_{\max}} \sum_{l' \in D^{l_{\max}}} \sum_{l \in D^{4l_{\max}}}  \left\|\langle l | e^{-i \mathscr{H}_{\text{eff,pbc}}^{4l_{\max}} t}-e^{-i \mathscr{H}_{\text{eff}}^{4l_{\max}} t} | l' \rangle | \psi(0) \rangle\right\| \notag \\
&\leq \sum_{l' \in D^{l_{\max}},\;l \in D^{4l_{\max}} } \frac{e^{2\beta  t-(8l_{\max}-|l|-|l'|)^{{\frac{1}{\varrho}}}/2\zeta}}{4l_{\max}} \notag \\
&\leq  \left(\sum_{m=0}^{\infty} e^{- \frac{m^{{\frac{1}{\varrho}}}}{4\zeta} }\right) e^{2\beta  t-(3l_{\max})^{\frac{1}{\varrho}}/4\zeta}
\leq \frac{\varepsilon^{3^{1/\varrho}}}{4^{ 3^{1/\varrho}}(S(4\zeta))^{3^{1/\varrho}-1}}.
\end{align}
Since $\langle 0 | \mathscr{U}_{\text{amp1}}^{l_{\max}}(t) |0\rangle |\psi(0)\rangle$ accurately provides the state $|\psi(t)\rangle / 2$ as \cref{th:Amplification_by_symmetry}, a triangle inequality concludes the first inequality. Since the oblivious amplitude amplification generates the error at most four times, the second inequality is immediately derived. The last assertion comes from the fact that for $1<\varrho<3$,
\begin{align}
4^{3^{1/\varrho}}\bigl(S(4\zeta)\bigr)^{3^{1/\varrho}-1}
&\geq 4^{3^{1/\varrho}}\Bigl(\Gamma(\varrho+1)(4\zeta)^{\varrho}\Bigr)^{3^{1/\varrho}-1} \notag \\
&\geq 4^{3^{1/\varrho}}\Bigl(\Gamma(\varrho+1)4^{\varrho}
    \Bigl(\frac{2}{\varrho}\Bigr)^{\varrho}\cdot\frac{A}{\pi}\Bigr)^{3^{1/\varrho}-1} \notag \\
&\geq 4^{3^{1/\varrho}}\Bigl(\frac{64\cdot 32\,e^{2+1/e}}{\pi}\Bigr)^{3^{1/\varrho}-1} \notag \\
&\geq 4^{3^{1/3}}\Bigl(\frac{64\cdot 32\,e^{2+1/e}}{\pi}\Bigr)^{3^{1/3}-1} \notag \\
&> 369.
\end{align}
This completes the proof.
\end{proof}

\end{document}